\documentclass[aps,showpacs,amssymb,superscriptaddress,onecolumn,pra,notitlepage,nofootinbib]{revtex4-1}

\usepackage{bbm,mathrsfs}
\usepackage{graphicx}
\usepackage{amsfonts,amsmath}
\usepackage{amsthm}
\usepackage{times}
\usepackage{color}
\def\textbf#1{{\bf #1}}
\def\be{\begin{equation}}
\def\ee{\end{equation}}
\def\ben{\begin{eqnarray}}
\def\een{\end{eqnarray}}
\def\eea{\end{array}}
\def\bea{\begin{array}}
\newcommand{\Tr}[0]{\mathrm{Tr}}
\newcommand{\ot}[0]{\otimes}
\newcommand{\bei}{\begin{itemize}}
\newcommand{\eei}{\end{itemize}}
\newcommand{\ket}[1]{|#1\rangle}
\newcommand{\bra}[1]{\langle#1|}

\newcommand{\proj}[1]{\ket{#1}\!\bra{#1}}

\newcommand{\nmss}{\negmedspace\negmedspace}
\newcommand{\nmsss}{\negmedspace\negmedspace\negmedspace}

\usepackage{lipsum}

\usepackage{hyperref}
\definecolor{myurlcolor}{rgb}{0,0,0.7}
\definecolor{myrefcolor}{rgb}{0.8,0,0}

\hypersetup{colorlinks, linkcolor=myrefcolor, citecolor=myurlcolor, urlcolor=myurlcolor}

\newtheorem{thm}{Theorem}

\theoremstyle{definition}

\begin{document}

\title{Nonlocality in many-body quantum systems detected with two-body correlators}

\author{J. Tura}
\affiliation{ICFO -- Institut de Ciencies Fotoniques, Mediterranean Technology Park,
08860 Castelldefels (Barcelona), Spain}
\author{R. Augusiak}
\affiliation{ICFO -- Institut de Ciencies Fotoniques, Mediterranean Technology Park,
08860 Castelldefels (Barcelona), Spain}
\author{A. B. Sainz}
\affiliation{ICFO -- Institut de Ciencies Fotoniques, Mediterranean Technology Park,
08860 Castelldefels (Barcelona), Spain}
\author{B. L\"ucke}
\affiliation{Institut f\"ur Quantenoptik, Leibniz Universit\"at Hannover, Welfengarten 1, D-30167 Hannover, Germany}
\author{C. Klempt}
\affiliation{Institut f\"ur Quantenoptik, Leibniz Universit\"at Hannover, Welfengarten 1, D-30167 Hannover, Germany}
\author{M. Lewenstein}
\affiliation{ICFO -- Institut de Ciencies Fotoniques, Mediterranean Technology Park,
08860 Castelldefels (Barcelona), Spain}
\affiliation{ICREA -- Instituci\'o Catalana de Recerca i Estudis Avan\c cats, Lluis Campanys 3, 08010 Barcelona, Spain}
\author{A. Ac\'in}
\affiliation{ICFO -- Institut de Ciencies Fotoniques, Mediterranean Technology Park,
08860 Castelldefels (Barcelona), Spain}
\affiliation{ICREA -- Instituci\'o Catalana de Recerca i Estudis Avan\c cats, Lluis Campanys 3, 08010 Barcelona, Spain}

%-------------------------------------------------------------------------------------
\begin{abstract}
Contemporary  understanding of correlations in quantum many-body
systems and in  quantum phase transitions is based to a large extent
on the recent intensive  studies of entanglement in many-body
systems. In contrast, much less is known about the role of quantum
nonlocality in these systems, mostly because the available
multipartite Bell inequalities involve high-order correlations
among many particles, which are hard to access theoretically, and
even harder experimentally. Standard, "theorist- and
experimentalist-friendly"  many-body observables involve
correlations among only few (one, two, rarely three...)
particles. Typically, there is no multipartite Bell inequality
for this scenario based on such low-order correlations. Recently,
however, we have succeeded  in constructing multipartite Bell
inequalities that involve two- and  one-body correlations only,  and
showed how they revealed the nonlocality in many-body systems
relevant for nuclear and atomic physics [Science {\bf 344}, 1256
(2014)]. With the present contribution we continue our work on this problem.
On the one hand, we present a detailed derivation of the above Bell inequalities, pertaining
to permutation symmetry among the involved parties. On the other hand, we present a couple of
new results concerning such Bell inequalities. First, we characterize their tightness. We then
discuss maximal quantum violations of these inequalities in the general case, and their scaling
with the number of parties. Moreover, we provide new classes of two-body Bell inequalities which
reveal nonlocality of the Dicke states---ground states of physically relevant and experimentally realizable
Hamiltonians. Finally, we shortly discuss various scenarios for nonlocality detection
in mesoscopic systems of trapped ions or atoms, and by atoms  trapped in the vicinity of designed nanostructures.
\end{abstract}

%-------------------------------------------------------------------------------------

\maketitle

\section{Introduction}

Nonlocality is a property of correlations that go beyond the
paradigm of local
realism~\cite{EPR,Bell64,Bellbook,NLreview,Bellissue}. According to
the celebrated theorem of J.~S.~Bell,  correlations between the results
of measurements of local observables on some quantum states might be
nonlocal, i.e. they do not admit a {\it local hidden variables}
(LHV) model \`a la Einstein-Podolski-Rosen (EPR)~\cite{EPR, Bohmbook}.
In such instances we talk about Bell nonlocality.
Operationally, such correlations
cannot be simulated by correlated classical random variables.
This can be detected by means
of the so-called Bell inequalities \cite{Bell64} - the celebrated example of such is
the famous Clauser-Horne-Shimony-Holt inequality (CHSH)~\cite{CHSH}.
Bell nonlocality is interesting for at least  three reasons:
\begin{itemize}

\item It is a resource for certain information--processing tasks, such as secure key
distribution~\cite{key1,key2,key3} or certified quantum randomness
generation~\cite{rand1,rand2,rand3}. Hence, revealing the
nonlocality of a given composite quantum state is one of the
central problems of quantum information theory,
of relevance for future quantum technologies.

\item It lies at the heart of fundamental aspects of quantum
physics~\cite{ZurekWheeler,Bellbook,Bohmbook}, e.g.,
playing a fundamental role in our understanding of
intrinsic quantum randomness, and sometimes even in philosophical contexts
such as the {\it free will problem} \cite{Heisenberg-free,FreeWillbook}.

\item It is extremely challenging mathematically: characterization
of non-locality is a complex and  difficult problem.
In fact it was proven that determining whether a given set of
correlations belongs or not to the set of local
EPR correlations, is NP-complete, whereas specifying all Bell
inequalities for growing number of involved parties, and/or
measured observables, and/or outcomes of measurements is NP-hard
(\cite{Babai,Avis}; see also \cite{NLreview,Bellissue} and
references therein).
\end{itemize}

Quantum states that violate a Bell inequality are
necessarily entangled, or in other words, they are not separable
and cannot be represented as mixtures of projections onto
product states \cite{Werner1989} (for a review on entanglement see
\cite{Horodeccy}); the opposite  is not necessarily true. Already in 1991 Gisin
proved \cite{Gisin1991} that any entangled pure state of
two parties violates a Bell inequality (in fact, they violate the CHSH inequality).
This result was extended to
pure entangled states of an arbitrary number of parties by Popescu and Rohrlich
\cite{PR}. However, the situation changes completely when moving to mixed states:
In his seminal paper from 1989 \cite{Werner1989} Werner
constructed examples of mixed bipartite states that admit a LHV
model for all local projective measurements, and nevertheless are
entangled. This result was then generalized by Barrett to
arbitrary measurements \cite{Barrett}. Very recently,
some of us have shown that entanglement and nonlocality are inequivalent
for any number of parties, i.e. that there exist genuinely
multipartite entangled states that do not display genuinely
multipartite nonlocality \cite{Demian2014}.
It is then clear that entanglement is a weaker property
of quantum states than nonlocality.

Over the last decade, entanglement has proven
to be very useful to characterize properties of
many-body systems and the nature of quantum phase transitions
(QPT) \cite{Sachdev}.  Let us illustrate this statement for
lattice spins models described by local Hamiltonians, that is,
those involving finite range interactions, or interactions
rapidly decaying with distance. Most of the results listed below
also hold for lattice Bose and Fermi models, and even for
quantum field theories. In particular, in the ground states (GS)
(low energy states) of such models the following properties are
true (for general reviews see Refs. \cite{Augusiak2010,Lewenstein2012}):
\begin{itemize}

\item If one calculates the reduced density matrix for two spins, this density matrix
typically exhibits  entanglement for short separations of the spins only.
This occurs even at criticality; still entanglement measures show signatures of QPTs
\cite{Amico,Osborne2};

\item If, instead of calculating a reduced density matrix, one tries to perform optimized measurements on
the rest of the system, one concentrates then entanglement in the
chosen two spins. One obtains in this way {\it localizable
entanglement} \cite{Frank1}; the corresponding entanglement length
diverges when the standard correlation length diverges, i.e., at
standard QPTs. Interestingly, there exist critical systems for
which the correlation length remains finite, while the entanglement
localization length goes to infinity \cite{Frank2};

\item For non-critical systems, GSs and low energy states exhibit the, so-called, area laws: the
von Neuman (or R\'enyi) entropy of the reduced density matrix of a
block of size $I$ scales as the size of the boundary of the block,
$\partial I$; at criticality a logarithmic divergence occurs
frequently \cite{VidalLatorre} (for a review see \cite{CardyJPA,EisertRMP}).
These results are very well established in 1D, while there are plenty of open questions in 2D
and higher dimensions;

\item GSs and low energy states away from criticality can be efficiently described by the, so-called, matrix product states, or more generally tensor network states (cf. \cite{Frank/Cirac2}); even at criticality MPS are efficient in 1D.

\item Topological order (at least for gapped systems in 1D and 2D) exhibits itself
frequently in the properties of the, so-called, {\it entanglement spectrum}, i.e.,
the spectrum of the logarithm of the reduced density matrix of a block $I$
\cite{Haldane}, and in 2D in the appearance of the, so-called,
{\it topological entropy}, i.e., a negative constant correction to
the area laws \cite{Preskill,Wen}.
\end{itemize}

A natural question thus arises: \textit{Does nonlocality also play an important role in characterizing
correlations in many-body systems?}
Apart from its fundamental interest, so far the role of
nonlocality in such systems has hardly been explored (see \textit{e.g.} \cite{Oliveira}). Entanglement
and nonlocality are known to be inequivalent quantum resources. In
principle, a generic many-body state, say a ground state of a
local Hamiltonian, is pure, entangled and, because all pure
entangled states violate a Bell inequality~\cite{PR}, it is also
nonlocal. However, this result cannot be easily verified in
experiments, because the known Bell inequalities (see, e.g.,
\cite{nier1,nier2,Sliwa,JD1,nier4}) usually involve products of
observables of all observers (also referred to as parties).
Unfortunately, measurements of such observables, although in
principle possible \cite{Greiner-single,Bloch-single}, are
technically extremely difficult; instead  one has typically "easy"
access to few-body correlations, say one- and  two-body, in
generic  many-body systems. Thus, the physically relevant question
concerning the nonlocality of many-body quantum states is whether
its detection is possible using only two-body correlations.

When tackling this problem one has to face the following technical
challenges. First of all, finding Bell inequalities that are valid
for an arbitrary number of parties is usually a difficult task, as
the complexity of characterizing the set of classical correlations
scales doubly exponentially with the number of parties. Then, it is natural to expect
that Bell inequalities based on only few-body correlations are, in general, weaker in revealing
nonlocality than those constructed from high-order correlations (that is, correlations among
many of the observers). Whereas Bell inequalities
involving correlations among all but one parties have been
constructed in \cite{BSV,Wiesniak,Lars}, thus proving that
all--partite correlations are not necessary to reveal nonlocality,
here we have to consider the more demanding question of whether
nonlocality detection is possible for systems of an arbitrary
number of parties from the minimal information achievable in a
Bell test, i.e., two--body correlations. We  answered  this
question positively in Ref. \cite{Science} by proposing classes of
Bell inequalities constructed from one- and two-body expectation
values, and, more importantly, providing examples of physically
relevant  many-body quantum states (i.e., GSs or low energy states
of physically relevant Hamiltonians) that violate these
inequalities. Note that finding and classifying such states is an
interesting task in itself, especially in view of the fact that
many genuinely entangled quantum many-body states have two-body
reduced density matrices (or, in other words, covariance matrices)
that are compatible with two-body reduced density matrices of some
separable state. This is the case of the so-called {\it graph
states}, as demonstrated in Ref. \cite{Gittsovich}; obviously one
cannot detect entanglement of such states with two-body
correlators, not even mentioning nonlocality.

It should also be stressed that an analogous question was already explored in the case of
entanglement, which next to nonlocality is a key resource quantum information theory.
Several entanglement criteria relying solely on two-body expectation values have been proposed
\cite{EntTwoBody}. In particular, in \cite{coll} the possibility of addressing two-body statistics
(although not individually) \textit{via} collective observables was exploited.

The main aims of this work are two-fold. First, we present a detailed derivation of
multipartite Bell inequalities involving only one- and two-body correlators, and pertaining to
permutation symmetry between the involved parties presented in \cite{Science}. To this end we
characterize the polytope of LHV correlations in such case. We then discuss
the maximal quantum violations of these inequalities in the general case, including its
scaling with the number of parties. We pay particular attention to
Bell inequalities violated by the family of Dicke states that are
the ground states of physically relevant and experimentally
realizable Hamiltonians. We discuss shortly various scenarios for
detection of non-locality in mesoscopic systems of trapped ions or
atoms, and by atoms trapped in the vicinity of designed
nanostructures. The second aim is to strongly develop and generalize the results of Ref. \cite{Science}.
First, we impose conditions under which our Bell inequalities are tight. Then, we provide a new class of
two-body Bell inequalities which reveal nonlocality in all the entangled Dicke states, not only those studied in
\cite{Science}. Finally, we construct and characterize an analytical class of pure states
that well approximate states maximally violating some of our Bell inequalities. With their aid
we can analytically compute the maximal violation of these Bell inequalities in the thermodynamic limit.

The manuscript is organized in sections, in which we define the problem
and present the results. All the technical proofs and
demonstrations, as well as a discussion of concrete experiments
aimed at generation of the Dicke states, are shifted to the three
Appendices. In Section II we describe the generic Bell experiment,
using the contemporary language of device-independent quantum
information theory in terms of conditional probabilities for $n$ parties
compatible with the no-signalling principle. The structure of the
polytope defined by the local correlations (local Bell
polytope) is discussed in Section \ref{sec:projecting}. Here we first describe the
general case, then how we can restrict the analysis to a smaller polytope using a
general symmetry group, and finally we focus on  the permutationally
invariant Bell polytope. In Section \ref{sec:classes} various classes of Bell
inequalities are derived and discussed. Section \ref{sec:quantum} is devoted to results concerning maximal quantum
violation of the derived inequalities. We present here also analytical expressions for a class of states that violate
the derived Bell inequalities, and such that the analytical
violation converges to the maximal numerical one as $n$ becomes large.
Various aspects of the robustness of these inequalities and their
violations are analyzed here, with an emphasis on the limit of a large number of parties. In
Section \ref{sec:Dicke} we focus on inequalities detecting nonlocality of the entangled Dicke states, and discuss also their parameter dependence. Section
\ref{sec:conclusions} contains a short conclusion, and a discussion of several
promising experimental systems in which the proposed Bell test
could be successfully implemented: mesoscopic
systems of ultracold trapped ions, neutral atoms in a single, or
double well trap, or atoms near  tapered optical fibers and
nanostructured optical crystals.

Appendix \ref{AppA} contains the proofs of all the theorems from
Sections \ref{sec:projecting}-\ref{sec:Dicke}, while Appendix \ref{AppB} presents some details
concerning the block-diagonalization of the Bell operator, necessary to
determine the maximal quantum violation. As mentioned above,
Appendix \ref{AppC} contains detailed calculations concerning the
robustness to experimental errors of the 2-body Bell inequalities.
These calculations are relevant for the potential implementation of the derived Bell inequalities,
for instance in the recent experiment of Ref. \cite{Klempt}.

\section{The Bell Experiment}
\label{sec:Bellexp}

Let us consider the standard multipartite Bell-type scenario in which $n$ spatially separated observers, denoted $A_0,\ldots,A_{n-1}$ share some $n$-partite resource. On their share of this resource every observer $A_i$ can perform one of $m$ measurements, each having $d$ outcomes. In what follows these measurements are denoted by $\mathcal{M}_{x_i}^{(i)}$ with $x_i\in\{0,\ldots,m-1\}$, while their outcomes by $a_i\in\{0,\ldots,d-1\}$.
This scenario is usually referred to as $(n,m,d)$.

The correlations between results obtained in the above experiment are described in terms of the conditional probabilities
\begin{equation}
 \label{eq:ProbabilitatCondicionada}
 P(a_0,\ldots,a_{n-1}|x_0,\ldots,x_{n-1}), \qquad 0 \leq a_i < d, \qquad 0 \leq x_i < m,
\end{equation}
that the $i$th party obtained $a_i$ upon performing the measurement $\mathcal{M}_{x_i}^{(i)}$.

It is convenient to collect all the probabilities in (\ref{eq:ProbabilitatCondicionada}) into a vector with $(md)^n$ components, each corresponding to a different combination of outcomes $a_0,\ldots,a_{n-1}$ and measurements choices $x_0,\ldots,x_{n-1}$. Let us call such vector $\vec{p}$.
Clearly, one must ask for probabilities to be non-negative, that is, $P(a_0,\ldots,a_{n-1}|x_0,\ldots,x_{n-1}) \geq 0$ must hold for all choices of $a_0,\ldots,a_{n-1}$ and $x_0,\ldots,x_{n-1}$, and normalized, meaning that
\begin{equation}
\sum_{a_0,\ldots,a_{n-1}}P(a_0,\ldots,a_{n-1}|x_0,\ldots,x_{n-1})=1
\end{equation}
is satisfied for any sequence $x_0,\ldots,x_{n-1}$. The latter condition implies that only $m^n(d^n-1)$ of the components of $\vec{p}$ are independent. Let us also notice that all the above constraints define a region in space, which is a polytope, denoted $\mathbbm{P}_{S}$,
%\added[id=JT]{(To me this is not a good name, as it also contains NS-correlations)},
to which any mathematically consistent $\vec{p}$ belongs. Recall that a polytope is a compact convex set with a finite number of extreme points.

Now, depending on the nature of the $n$-partite resource that the parties share, different types of correlations can be obtained. We are mainly interested in three scenarios: the case in which $\vec{p}$ admits a local hidden variable (LHV) model, the case in which $\vec{p}$ can be obtained through performing local measurements on quantum states (Q) and the case in which $\vec{p}$ fulfils the no-signalling principle (NS).

\textit{LHV correlations.} Those conditional probability distributions that $\vec{p}$ have a local hidden variable model have the following form:
\begin{equation}
 \label{eq:LHVmodel}
 P(a_0,\ldots,a_{n-1}|x_0,\ldots,x_{n-1})=\int_{\Lambda}p(\lambda)\prod_{i=0}^{n-1}P(a_i|x_i,\lambda)\mathrm{d}\lambda,
\end{equation}
where $\lambda\in \Lambda$ is some hidden variable, distributed according to a probability density $p(\lambda)$, and $\Lambda$ is the space of hidden variables. Operationally, the interpretation of (\ref{eq:LHVmodel}) is that the outcome $a_i$ of the $x_i$-th measurement is produced locally at site $i$ as a function of the input $x_i$ and the value of this hidden variable $\lambda$. The resulting  probability distribution is the one corresponding to our  understanding of the \textit{classical correlations between local independent measurements}: In order to simulate correlations $\vec{p}$ of this form, parties could gather beforehand and agree on which $\lambda$ to pick at each run of the experiment,  distributed according to $p(\lambda)$. This information, called also shared randomness is precisely the resource needed---given $\lambda$ and $x_i$, each party would just simulate $P(a_i|x_i,\lambda)$ on their own. The feasibility region for LHV model probability distributions is a polytope and its vertices are given by those $\vec{p}$ in which all parties follow a local deterministic strategy. We shall denote this polytope $\mathbbm{P}$.

The vertices of $\mathbbm{P}$ are straightforward to construct as these are those
$\vec{p}$'s for which all probabilities (\ref{eq:ProbabilitatCondicionada}) factorize:
\begin{equation}
 \label{eq:LHVVertex}
 P(a_0,\ldots,a_{n-1}|x_0,\ldots,x_{n-1})=\prod_{i=0}^{n-1}P(a_i|x_i),
\end{equation}
and each factor is deterministic, i.e., $P(a_i|x_i)=\delta(a_i=\alpha_i^{x_i})$, where $\delta$ is the Kronecker delta function and $0\leq \alpha_i^{x_i} < d$ is a different integer for every vertex. One immediately finds that in the $(n,m,d)$ scenario the number of vertices is $d^{mn}$.

\textit{Quantum correlations.} Imagine now that the resource the parties share is some $n$-partite quantum state
$\rho$ acting on a product Hilbert space $\mathcal{H}_n=(\mathbbm{C}^{D})^{\ot n}$ of local dimension $D$.
In such case, each local measurement $\mathcal{M}_{x_i}^{(i)}$ can be thought of as a collection of positive semi-definite operators $\Pi_{a_i}^{(x_i)}$ corresponding to outcomes $a_i=0,\ldots,d-1$ and obeying
\begin{equation}
\sum_{a_i=0}^{d-1}\Pi_{a_i}^{(x_i)}=\mathbbm{1}_D
\end{equation}
for every $x_i$ and $i$. Then, the probabilities (\ref{eq:ProbabilitatCondicionada}) can be represented \textit{via} Born's rule:
\begin{equation}
 \label{eq:BornsRule}
 P(a_0,\ldots,a_{n-1}|x_0,\ldots,x_{n-1})=\Tr\left(\rho \otimes_{i=0}^{n-1}\Pi_{a_i}^{(x_i)}\right).
\end{equation}
The feasibility region for probability distributions of the form (\ref{eq:BornsRule}) is a convex set, denoted $\mathbf{Q}$. Unlike the local set, $\mathbf{Q}$ is not a polytope and hence its boundary is more difficult to characterize. Still, it can be well approximated \cite{NPA}.

\textit{Non-signalling correlations.} For completeness of this work it is worth recalling the notion of non-signalling correlations. These are those $\vec{p}$'s that fulfil the no-signalling principle which says that information cannot be transmitted instantaneously. This means that in the above Bell scenario the choice of measurement made by one party cannot influence the statistics seen by the remaining parties. In terms of the conditional probabilities (\ref{eq:ProbabilitatCondicionada}) this is equivalent to a set of linear constraints of the form
\begin{equation}
 \label{eq:NSPrinciple}
 \sum_{a_i} P(a_0,\ldots,a_i,\ldots a_{n-1}|x_0,\ldots,x_i,\ldots,x_{n-1})=\sum_{a_i} P(a_0,\ldots,a_i,\ldots a_{n-1}|x_0,\ldots,x_i',\ldots,x_{n-1}).
\end{equation}
for all $x_i\neq x_i'$ and $a_0,\ldots,a_{i-1},a_{i+1},\ldots, a_{d-1}$ and $x_0,\ldots,x_{i-1},x_{i+1},\ldots,x_{d-1}$ and all $i$.
%
%In this case, equation (\ref{eq:NSPrinciple}) ensures that the marginal probability distribution does not actually %depend on $x_i$ and it can be written as $P(a_0,\ldots,a_{i-1}, a_{i+1},\ldots a_{n-1}|x_0,\ldots,x_{i-1},x_{i+1},%\ldots,x_{n-1})$. The no-signalling condition (\ref{eq:NSPrinciple}) is applied recursively to the marginal probabilities; hence to any subset of parties.
%
The feasibility region corresponding to distributions $\vec{p}$ satisfying the no-signalling principle, which we denote $\mathbbm{P}_{NS}$, is again a polytope, as it is an intersection of the polytope $\mathbbm{P}_S$ with the linear subspace defined by (\ref{eq:NSPrinciple}). However, unlike in the case of the local polytope $\mathbbm{P}$, the vertices of $\mathbbm{P}_{NS}$ are in general difficult to characterize, and except for the simplest scenarios
$(2,2,d)$ \cite{Barrett2005} and $(2,m,2)$ \cite{MasanesJones} and they are unknown (see also Ref. \cite{Tobias}).

When the no-signalling constraints and the normalization conditions are taken into consideration, the number of independent variables of any $\vec{p} \in \mathbbm{P}_{NS}$ is reduced to $[m(d-1)+1]^n-1$ which is the dimension of the polytope $\mathbbm{P}_{NS}$ and so is the dimension of the remaining two sets $\mathbbm{P}$ and $\mathbf{Q}$. Importantly, these three sets are different. In fact, one has the following chain of inclusions:
\begin{equation}
 \label{eq:inclusions}
 \mathbbm{P} \subsetneq \mathbf{Q} \subsetneq \mathbbm{P}_{NS} \subsetneq \mathbbm{P}_S.
\end{equation}
First, $\mathbbm{P}\subset\mathbf{Q}$ because all correlations belonging to $\mathbbm{P}$ can be realized with
separable states. Yet, due to the celebrated Bell's theorem $\mathbbm{P}\neq \mathbf{Q}$ \cite{Bell64,Bellissue}. Second, $\mathbf{Q}\subset\mathbbm{P}_{NS}$ because all correlations given by Eq. (\ref{eq:BornsRule})
trivially satisfy the no-signalling principle. On the other hand, there exist correlations obeying
the no-signalling principle which are not quantum, with the most paradigmatic example being the so-called
Popescu-Rohrlich box \cite{PR}. Hence, $\mathbf{Q}\neq \mathbbm{P}_{NS}$. Finally, it is easy to verify that
not all elements of $\mathbbm{P}_S$ satisfy the no-singalling principle, and therefore $\mathbbm{P}_{NS}\subsetneq\mathbbm{P}_S$.

Due to the structure of the local polytope $\mathbbm{P}$, the natural way of
checking whether $\vec{p}\in \mathbf{Q}\setminus\mathbbm{P}$ is to use Bell inequalities.
These are linear inequalities formulated in terms of the probabilities resulting
from the local measurements performed by the observers, that is,
\begin{equation}
I:=\sum_{\boldsymbol{a},\boldsymbol{x}}T_{\boldsymbol{a},\boldsymbol{x}}P(\boldsymbol{a}|\boldsymbol{x})\geq -\beta_C
\end{equation}
where $T_{\boldsymbol{a},\boldsymbol{x}}\in {\mathbbm{R}}$ are some coefficients and $\beta_C=-\min_{\vec{p}\in\mathbbm{P}}I$ is the so-called classical bound of the Bell inequality. Violation of such inequalities signals nonlocality.

All polytopes admit a dual description: They can be completely characterized either by listing all their vertices (\textit{i.e.}, a point $\vec{p}$ belongs to a polytope $\mathbbm{P}_S$ if and only if it is a convex combination of its vertices) or as the intersection of a number of half-spaces. Such intersection can be taken to be minimal and in this case every of these half-spaces intersected with $\mathbbm{P}_S$ defines a facet of $\mathbbm{P}_S$. Nevertheless, going from one description to the other is not a simple task in general. The best known algorithm \cite{Chazelle} has complexity $O(v^{\lfloor D/2 \rfloor})$, where $v$ is the number of vertices (facets) of the polytope and $D$ is the dimension of the space in which the polytope lives. In the case of $\mathbbm{P}$, its vertices, although there are many, are straightforward to find, but not its facets, which give rise to the so-called \textit{tight Bell inequalities}. Only few scenarios have been completely characterized, none of them with more than $3$ parties \cite{Sliwa,CollinsGisin}. In the case of $\mathbbm{P}_{NS}$ the situation is the opposite: its facets are easy to construct, since they are the $(md)^n$ positivity constraints on every probability (\ref{eq:ProbabilitatCondicionada}), but its vertices, known as extreme non-signalling boxes, are hard to find. In this case only few scenarios have been solved \cite{Barrett,MasanesJones,Tobias}.

For further benefits let us eventually mention that in the case $d=2$, which is the case we mostly consider in this work, it is convenient
to re-label outcomes from $0,1$ to $\pm 1$ (\textit{i.e.,} $a \mapsto (-1)^a$), and re-express all probabilities (\ref{eq:ProbabilitatCondicionada}) in terms of expectation values
\begin{equation}
 \label{eq:probs2expectations}
 \langle {\cal M}_{j_1}^{(i_1)}\cdots {\cal M}_{j_k}^{(i_k)}\rangle = \sum_{a_{i_1}, \ldots, a_{i_k}} (-1)^{\sum_{l=1}^{k}a_{i_l}}P(a_{i_1}\ldots a_{i_k}|x_{i_1}\ldots x_{i_k}),
\end{equation}
for $0\leq i_1 < \ldots < i_k < n$, $j_l\in \{0,1\}$ and $1 \leq k \leq n$. Below, the number of parties
appearing in such a correlators will also be referred to as order of the correlator; for instance,
correlators $ \langle {\cal M}_{j_1}^{(i_1)}{\cal M}_{j_2}^{(i_2)}\rangle$ are of order two.
As a direct consequence of (\ref{eq:LHVVertex}), the vertices of $\mathbbm{P}$ are in this case given by
\begin{equation}
\label{eq:correlationVertices}
\langle {\cal M}_{j_1}^{(i_1)}\cdots {\cal M}_{j_k}^{(i_k)}\rangle=\langle {\cal M}_{j_1}^{(i_1)}\rangle \cdots \langle{\cal M}_{j_k}^{(i_k)}\rangle,
\end{equation}
where each local expectation value $\langle \mathcal{M}_{x_i}^{(i)}\rangle$ equals either $-1$ or $1$.

\section{Projecting the local polytope}
\label{sec:projecting}
In the present work we are interested in the study of non-locality in many-body quantum systems consisting of
a large number of parties $n$. Finding all the facets of the corresponding local polytope $\mathbbm{P}$ (that give rise to the optimal Bell inequalities) would completely characterize it and would be the first step to answer the above question. Unfortunately, even in the simplest case of $m=d=2$ but for large $n$, $\mathbbm{P}$ is an object of tremendous complexity, and attempting to characterize it following this approach becomes an intractable task. Since both its dimension and its number of vertices are exponential in $n$, we aim at projecting $\mathbbm{P}$ onto a simpler object, which has less vertices and which is embedded in a lower-dimensional space.

This simplification does not come for free and the obtained Bell inequalities would, in general, be weaker. If one applies a projection $\pi$ to $\mathbbm{P}$, then $\pi(\mathbbm{P})$ can only give sufficient conditions for nonlocality \cite{JD1}. If some correlations $\vec{p}$ are nonlocal, \textit{i.e.}, $\vec{p}\notin \mathbbm{P}$, it may happen that $\pi(\vec{p}) \notin \pi(\mathbbm{P})$ or $\pi(\vec{p}) \in \pi(\mathbbm{P})$, depending on the projection $\pi$. However, if $\pi(\vec{p}) \notin \pi(\mathbbm{P})$, then $\vec{p} \notin \mathbbm{P}$. Hence, violating a Bell inequality corresponding to $\pi(\mathbbm{P})$ certifies nonlocality, while proving that correlations that lie inside $\pi(\mathbbm{P})$ is inconclusive. In this section we then discuss how to choose $\pi$ in such a way that, on one hand, the Bell inequalities constraining $\pi(\mathbbm{P})$ are powerful enough to reveal nonlocality for arbitrary $n$, and, on the other hand, they are accessible with current experimental technology. More precisely, we want to find Bell inequalities containing only one- and two-body correlators. To further simplify the problem we also demand that additionally they obey certain symmetries. From a more fundamental point of view, we also want to answer the question of what is the minimal amount of information that it is required in order to detect non-locality.

We shall simplify $\mathbbm{P}$ in two steps: one is to reduce the order of the correlators (\ref{eq:ProbabilitatCondicionada}), and the other is the application of a symmetry group $G$ such that the obtained Bell inequalities remain invariant under the permutation by any element of $G$. By reducing the amount of correlators in $(\ref{eq:ProbabilitatCondicionada})$ to those that involve at most $K$ parties, the dimension of $\mathbbm{P}$ is reduced from $(m(d-1)+1)^n-1$ to
\begin{equation}
 \label{eq:numberofKbodycorrelators}
 \sum_{k=1}^K{n \choose k}m^k(d-1)^k,
\end{equation}
and the number of vertices is kept the same. Each term in the above sum corresponds to correlators of exactly $k$ parties -- there are ${n \choose k}$ ways to choose them. When the $k$ parties are chosen, each of the members can choose one of $m$ measurements, with $d-1$ independent probabilities (the $d$-th one is determined from the fact that the sum of probabilities of all outcomes is equal to 1); this amounts to the factor $m^k(d-1)^k$.    We denote by $\mathbbm{P}_K$ the polytope obtained by \textit{not including} the correlators of order higher than $K$.

On the other hand, one can look for Bell inequalities that are invariant under the action of some symmetry group $G$ which permutes the parties. The reason to consider this approach that imposes extra constraints is twofold: firstly, even if $\mathbbm{P}_K$ is simpler than the whole polytope, its number of vertices is still exponential and its dimension goes like $O(n^K)$, which can still grow fast with $n$. Adding the permutation symmetry produces a simpler polytope. Secondly, as we will see below, imposing permutation symmetry allows for deriving Bell inequalities that can be measured with  techniques readily accessible in experiments.

\subsection{The symmetrized polytope}
Let us consider a group $G$, which is a subgroup of the group of permutations of $n$ elements, denoted ${\mathfrak S}_n$. Let us denote by ${\cal P}_K$ the set of all $K$-body conditional probabilities:
\begin{equation}
 \label{eq:setofKbodycorrelators}
 {\cal P}_K=\bigcup_{k=1}^K\left\{P_{i_1\ldots i_k}(a_{i_1}\ldots a_{i_k}|x_{i_1}\ldots x_{i_k}), \qquad 0 \leq i_1<\ldots < i_k< n, \quad 0 \leq a_j < d-1, \quad 0 \leq x_j < m \right\}.
\end{equation}
$G$ acts on ${\cal P}_K$ by permuting the indices of the parties $i_1\ldots i_k$. Formally, we have defined an action $g$ by:
\begin{equation}
 \label{eq:accio}
 \begin{array}{ccclcl}
  g:& G &\times& {\cal P}_K &\longrightarrow &{\cal P}_K\\
  &(\sigma&,& P_{i_1\ldots i_k}(a_{i_1}\ldots a_{i_k}|x_{i_1}\ldots x_{i_k}))&\mapsto& P_{\sigma(i_1)\ldots \sigma(i_k)}(a_{i_1}\ldots a_{i_k}|x_{i_1}\ldots x_{i_k}).
 \end{array}
\end{equation}
Every action of a group on a set induces a partition of the set into orbits. Given a probability $P \in {\cal P}_K$, we denote its orbit by $[P]=\{g(\sigma, P),\ \sigma \in G\}$ and the set of orbits by ${\cal P}_K/G$. Note that a  given probability $P$ can only belong to one orbit. Thus, we can express ${\cal P}_K$ as the disjoint union of all its orbits.
\begin{equation}
 \label{eq:partition}
 {\cal P}_K=\bigsqcup_{[P]\in {\cal P}_K/G}[P].
\end{equation}
The sum of all elements in $[P]$ is, by construction, invariant under the action of any element $\sigma \in G$. Let us then define \textit{$G$-invariant probabilities} as sums of probabilities that are invariant under the action of $G$, that is,
\begin{equation}
 \label{eq:G-invariant-correlator}
 S_{[P]}\propto \sum_{P\in [P]}P
\end{equation}
for all $[P]\in {\cal P}_K/G$ (there are clearly as many $G$-invariant probabilities as elements in ${\cal P}_K/G$).
Accordingly, we define the symmetric polytope of $K$-body probabilities, denoted ${\mathbbm P}_K^G$, as the image of ${\mathbbm P}$ under the projection (\ref{eq:G-invariant-correlator}). Its dimension is the number of $G$-invariant probabilities obtained from ${\cal P}_K/G$. Note that in (\ref{eq:G-invariant-correlator}) we have defined the $G$-invariant probabilities up to a proportionality constant.
This, however, does not change the shape of the polytope (\textit{i.e.} its number of vertices and/or facets and their arrangement).
Note that ${\mathbbm P}_K^G$ is indeed a polytope, because Eq. (\ref{eq:G-invariant-correlator}) defines a linear projection. So, ${\mathbbm P}_K^G$ can be completely characterized by listing its vertices. However, we are interested in finding the vertices of ${\mathbbm P}_K^G$ directly, without first constructing the vertices of ${\mathbbm P}_K$ (of which there is an exponential amount) and then projecting each of them by means of (\ref{eq:G-invariant-correlator}).

As the projection (\ref{eq:G-invariant-correlator}) from ${\mathbbm P}_K$ to ${\mathbbm P}_K^G$ is a linear map,  every vertex of ${\mathbbm P}_K^G$ has to be the image of a vertex of ${\mathbbm P}_K$ under this projection. The converse is not true in general, because a vertex of ${\mathbbm P}_K$ can be mapped by (\ref{eq:G-invariant-correlator}) to the interior of ${\mathbbm P}_K^G$. Since every set of deterministic local strategies that give the same value in each (\ref{eq:G-invariant-correlator}) is a candidate to a vertex, counting the number of them will give us an useful upper bound to the number of vertices in ${\mathbbm P}_K^G$. To this end, we will make use of P\'olya's enumeration theorem \cite{Polya}.

Every deterministic local strategy assigns to each party a list of the predetermined outcomes for all its measurements. So, it can be thought of as a function $f:X\longrightarrow Y$, where $X$ is the set $\{0,1,\ldots,n-1\}$, which indexes the parties and $Y$ is the set of $m$-tuples $(y_0,\ldots,y_{m-1})$ with $0\leq y_i < d$, in which $y_i$ is the outcome assigned to the $i$-th measurement. The set of all $f$'s is denoted $Y^X$. Permutation of the parties defines, in a similar fashion as in (\ref{eq:partition}), a partition $Y^X/G$ of the set of local deterministic strategies: two deterministic local strategies, $f_1$ and $f_2$, belong to the same class in $Y^X/G$ if and only if there exists a permutation $\sigma \in G$, such that when applied to the parties in $f_1$, we obtain $f_2$. P\'olya's enumeration theorem states that the number of such classes is given by
\begin{equation}
\label{eq:Polyatheorem}
|Y^X/G|=\frac{1}{|G|}\sum_{\sigma \in G}|Y|^{c(\sigma)},
\end{equation}
where $c(\sigma)$ counts the number of disjoint cycles of the permutation $\sigma$. Because every permutation decomposes into a product of disjoint cycles uniquely (up to permuting the cycles), $c(\sigma)$ is well defined.

Both the dimension and the number of vertices depend on the group $G$  considered,  and there is clearly a trade-off between the number of elements in $G$, the number of vertices, and the dimension of ${\mathbbm P}_K^G$: the bigger the symmetry group $G$, the smaller the dimension of $\mathbbm{P}_K^G$ and the smaller the number of vertices.\\

\noindent \textbf{Example} Let us illustrate this result with an example, in which $G$ is the group generated by the shift to the right: $0\mapsto 1 \mapsto \ldots \mapsto n-1 \mapsto 0$. In this case, one obtains Bell inequalities which are translationally invariant \cite{TIpaper}. One can show that
\begin{equation}
\label{eq:exempletransinv}
|Y^X/G| = \frac{1}{n}\sum_{\nu | n} \varphi(\nu) d^{mn/\nu},
\end{equation}
where the sum runs over all divisors $\nu$ of $n$, denoted $\nu | n$, and $\varphi(\nu)$ is Euler's totient function, which counts how many integers $l$ with $1 \leq l \leq \nu$ are coprime with $\nu$, (\textit{i.e.}, their greatest common divisor is $\gcd(\nu,l)=1$). Interestingly, when considering ${\mathbbm P}_2^G$ in some scenarios, the bound (\ref{eq:exempletransinv}) is tight [\textit{e.g.}, $(3,2,2)$ and $(5,2,2)$], whereas in some others it is not [\textit{e.g.} $(4,2,2)$].

\subsection{The permutationally invariant Bell polytope}
For the scope of this article, however, we shall restrict ourselves to the scenario where the obtained Bell
inequalities are invariant under any permutation of the parties; \textit{i.e.}, the case where the symmetry group $G$ is ${\mathfrak S}_n$. Let us first upper bound the number of vertices of ${\mathbbm P}_K^{{\mathfrak S}_n}$ by explicitly computing (\ref{eq:Polyatheorem}). In this case, recall that the number of permutations of $n$ elements with $k$ disjoint cycles is given by the unsigned Stirling number of the first kind $s(n,k)$, where $s(n,k)$ is defined through the following formula
%
%usually denoted $\left[\begin{array}{c}n\\k\end{array}\right]$, which has the following property
%
\cite{CombinatoricsHandbook}:
\begin{equation}
 \label{eq:Stirling1stkind}
 \sum_{k=1}^ns(n,k)x^k=x(x+1)\cdots(x+n-1)=\frac{(x+n-1)!}{(x-1)!}.
\end{equation}
where $x$ is a natural number. It then follows that
\begin{equation}
 \label{eq:exemple_Sym}
 |Y^X/{{\mathfrak S}_n}|=\frac{1}{n!}\sum_{\sigma \in {\mathfrak S}_n}(d^m)^{c(\sigma)}=\frac{1}{n!}\sum_{k=1}^ns(n,k)(d^m)^{k}=\frac{1}{n!}\frac{(n+d^m-1)!}{(d^m-1)!}={n+d^m-1\choose d^m-1},
\end{equation}
where the second equality is a consequence of the definition of $|s(n,k)|$, and the third equality stems from
Eq. (\ref{eq:Stirling1stkind}). As we shall see in section \ref{sec:vertices}, the bound (\ref{eq:exemple_Sym}) is not tight in general and it can be refined.

There is a simpler argument showing that the number of vertices of $\mathbbm{P}^{{\mathfrak S}_n}_K$ is upper bounded by (\ref{eq:exemple_Sym}), which in addition allows us to parametrize the vertices of it with $d^m-1$ integer parameters. The key idea is to observe that a local deterministic strategy $f:X \longrightarrow Y$ is just a coloring of a hypergraph whose nodes are the elements of $X$ and with the colors being the elements of $Y$.
This graph has multiple hyperedges, each one corresponding to a $k$-body correlator [\textit{cf.} Fig. \ref{fig:graph}]. If its nodes are permuted according to an element of $G$, then all the correlators (\ref{eq:G-invariant-correlator}) will have the same values. Hence, they will correspond to the same vertex of $\mathbbm{P}^{G}_K$.

\begin{center}
\begin{figure}[h!]
\includegraphics[scale=1]{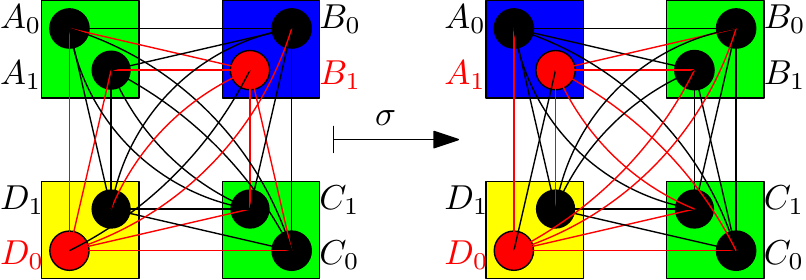}
\caption{(color online) \textit{The action of a permutation $\sigma = (A\mapsto C \mapsto B \mapsto A)(D\mapsto D)$ on a deterministic local strategy}. We pick as an example the $(4,2,2)$ scenario and consider the local deterministic strategy $f$ given by $f(A)=(+,+)$, $f(B)=(+,-)$, $f(C)=(+,+)$ and $f(D)=(-,+)$.  Each square represents a party and each circle corresponds to a measurement. We take the convention that if the circle is filled in black, its predetermined outcome is $+1$, whereas if it is filled in red, then its predetermined outcome is $-1$. All parties having the same deterministic local strategy have squares with the same color. Each line corresponds to a two-body correlator; if it is a black line, the results are correlated and if it is a red line, they are anti-correlated. Although the value of the single two-body correlators can be changed by $\sigma$ (\textit{e.g.}, $A_0B_1$), observe that the symmetrized correlators do not change their respective value. Both before and after the action
of $\sigma$, one has $S_{[A_0]}=A_0+B_0+C_0+D_0 = 2$, $S_{[A_1]}=2$, $S_{[A_0B_0]}=0$, $S_{[A_0B_1]}=4$, $S_{[A_1B_1]}=0$. In terms of $\mathbbm{P}$, the two graphs above correspond to two different vertices, but they are projected to the same point in $\mathbbm{P}^{\mathfrak{S}_n}_2$.}
\label{fig:graph}
\end{figure}
\end{center}

However, in the case that $G={\mathfrak S}_n$, the $G$-invariant probabilities (\ref{eq:G-invariant-correlator}) remain the same for any permutation of the parties. Thus, their value depends only on the amount of parties assigned to each color or, equivalently, the amount of parties with the same deterministic strategy. There will be as many equivalent strategies as partitions of $n$ in $d^m$ (some possibly empty) groups. Combinatorically, a simple way to count them is to imagine a set of $n$ elements in a line, then extend it to a set of $n+d^m-1$ elements and pick $d^m-1$ of them that act as a separator. This will produce a partition of it into $d^m$ groups, some of which may be empty. The number of ways to perform this choice is precisely (\ref{eq:exemple_Sym}).

From now on, we focus on the case when all the observables are dichotomic (\textit{i.e.}, $d=2$) and we switch for convenience to the expectation-value framework (\ref{eq:probs2expectations}), instead of working with probabilities (\ref{eq:ProbabilitatCondicionada}). Recall that for $d=2$ these two frameworks are equivalent. Let us then calculate the dimension of ${\mathbbm P}^{{\mathfrak S}_n}_K$, which is the number of elements in $\mathbbm{P}_K/{\mathfrak S}_n$. Equation (\ref{eq:G-invariant-correlator}) directly shows that $S_{[{\cal M}]}$, where ${\cal M}$ is a $k$-body expectation value, depends only on the choice of $k$ measurements out of the $m$, which are available. For instance, $S_{[A_2B_2C_0]} = S_{[B_2D_0E_2]}$. So, one can take as a canonical representative the lexicographically lowest element, which in this case would be $S_{[A_0B_2C_2]}$, and simply denote the corresponding correlator as ${\cal S}_{022}$. It is then clear that the amount of elements in $\mathbbm{P}_K/{\mathfrak S}_n$ is equal to the number of ordered sequences $0 \leq x_1 \leq x_2 \leq \cdots \leq x_k< m$ for $k=1,\ldots,K$, that is,
\begin{equation}
 \label{eq:dimension_PKSym}
|\mathbbm{P}_K/{\mathfrak S}_n|=\sum_{k=1}^{K}{k+m-1 \choose m-1}= {m+K \choose K}-1.
\end{equation}

In order to have the simplest symmetric polytope, we will further assume that $K=2$,
in which case the dimension of $\mathbbm{P}_2^{{\mathfrak S}_n}$ is just $5$.
Let us then define the one- and two-body ${\mathfrak S}_n$-invariant correlators (called simply \textit{permutationally invariant correlators}) we shall be working with:
\begin{equation}
\label{eq:def1bodysym}
{\cal S}_k := \sum_{i=0}^{n-1}\langle{\cal M}_k^{(i)}\rangle, \qquad 0\leq k \leq 1,
\end{equation}
\begin{equation}
\label{eq:def2bodysym}
{\cal S}_{kl}:=\sum_{\substack{i,j=0\\i\neq j}}^{n-1}\langle{\cal M}_k^{(i)}{\cal M}_{l}^{(j)}\rangle, \qquad 0 \leq k \leq l \leq 1,
\end{equation}
and observe that they are indeed proportional to their counterparts defined in (\ref{eq:G-invariant-correlator}).

Let us now calculate the value of (\ref{eq:def1bodysym}) and (\ref{eq:def2bodysym}) at every local deterministic strategy. In the case of single-body mean values, one can see from Fig. \ref{fig:graph} that their value corresponds to the black circles minus the red circles, whereas in the case of the two-body correlators one would take the corresponding black edges minus the red ones. It is then natural to work with the following variables:
\begin{equation}
\label{eq:defabcd}
a_f:=|f^{-1}(+,+)|, \quad b_f:=|f^{-1}(+,-)|, \quad c_f:=|f^{-1}(-,+)|, \quad d_f:=|f^{-1}(-,-)|,
\end{equation}
where $f^{-1}$ is the preimage of $f$. Notice that $a_f$ counts how many of the parties have predetermined outcomes ${\cal M}_0 = +, {\cal M}_1 = +$ for the deterministic local strategy $f$, etc. In the example of Fig. \ref{fig:graph}, one has $a_f=2$, $b_f=1$, $c_f=1$ and $d_f=0$. Observe that, for all $f$, one always has that $a_f+b_f+c_f+d_f=n$ because there are no more possibilities. Interestingly, this is all the information needed to compute the value of (\ref{eq:def1bodysym}) and (\ref{eq:def2bodysym}) at any deterministic local strategy.
Indeed, one directly sees that, for a given $f$, ${\cal S}_0 = a_f+b_f-c_f-d_f$ and ${\cal S}_1 = a_f-b_f+c_f-d_f$.
Moreover, it follows from Eq. (\ref{eq:correlationVertices}) that the value of the two-body permutationally invariant correlators (\ref{eq:def2bodysym}) for any $f$ is
\begin{equation}
\label{eq:factorization}
{\cal S}_{kl}={\cal S}_k {\cal S}_l - \sum_{i=0}^{n-1}\langle{\cal M}_k^{(i)}{\cal M}_l^{(i)}\rangle,
\end{equation}
where for $k=l$ the subtracted sum amounts to $n$, while for $k\neq l$ it is given by
\begin{equation}
\label{eq:defZ}
{\cal Z}:=\sum_{i=0}^{n-1}\langle{\cal M}_0^{(i)}{\cal M}_1^{(i)}\rangle=a_f-b_f-c_f+d_f.
\end{equation}
Finally, we have that ${\cal S}_{00}=({\cal S}_0)^2-n$, ${\cal S}_{01}={\cal S}_0{\cal S}_1-{\cal Z}$ and ${\cal S}_{11}=({\cal S}_1)^2-n$.

\subsection{Permutationally invariant Bell inequalities with two-body correlators}
The general form of Bell inequalities constraining the polytope $\mathbbm{P}_2$ in the scenario $(n,2,2)$
is
\begin{equation}
 \label{eq:2bodygeneral}
 I_2:=\sum_{i=0}^{n-1}\left(\alpha_i \langle {\cal M}_0^{(i)}\rangle+\beta_i \langle {\cal M}_1^{(i)}\rangle\right) + \sum_{0\leq i < j < n} \gamma_{ij} \langle {\cal M}_0^{(i)}{\cal M}_0^{(j)}\rangle + \sum_{0\leq i \neq j < n} \delta_{ij} \langle {\cal M}_0^{(i)}{\cal M}_1^{(j)}\rangle+\sum_{0\leq i < j < n} \varepsilon_{ij} \langle {\cal M}_1^{(i)}{\cal M}_1^{(j)}\rangle\geq-\beta_C,
\end{equation}
for some $\alpha_i, \beta_i, \gamma_{ij}, \delta_{ij}, \varepsilon_{ij} \in \mathbbm{R}$ and $\beta_c \in \mathbbm{R}$, which is the so-called classical bound defined as $\beta_C=-\min_{\vec{p}\in \mathbbm{P}_2}I_2$. There are $2n^2$ degrees of freedom, which can be read from expression (\ref{eq:numberofKbodycorrelators}) for $m=K=d=2$.

Then, the general form of Bell inequalities bounding the permutationally invariant polytope  ${\mathbbm{P}_2^{{\mathfrak S}_n}}$ simplifies to
\begin{equation}
 \label{eq:2bodySymgeneral}
 \beta_c + \alpha {\cal S}_0 + \beta {\cal S}_1 + \frac{\gamma}{2} {\cal S}_{00} + \delta {\cal S}_{01} + \frac{\varepsilon}{2} {\cal S}_{11} \geq 0,
\end{equation}
where $\alpha, \beta, \gamma, \delta, \varepsilon \in \mathbbm{R}$ and $\beta_c \in {\mathbbm R}$ is the corresponding classical bound.
In other words, by symmetrizing over ${\mathfrak S}_n$, we have found a way to obtain Bell inequalities of the form (\ref{eq:2bodygeneral}) with all the coefficients being equal: $\alpha_i = \alpha$, $\beta_i=\beta$, $\gamma_{ij}=\gamma$, $\delta_{ij}=\delta$, and $\varepsilon_{ij}=\varepsilon$.

\subsection{Characterizing the vertices of the symmetric polytope}
\label{sec:vertices}
The symmetrized $2$-body polytope $\mathbbm{P}_2^{{\mathfrak S}_n}$ in the $(n,2,2)$ scenario is embedded in a $5$-dimensional space with coordinates ${\cal S}_0, {\cal S}_1$, ${\cal S}_{00}$, ${\cal S}_{01}$ and ${\cal S}_{11}$. However, if a given set of correlations corresponds to a local deterministic strategy $f$, only $4$ parameters are needed in order to completely describe it: either $(a, b, c, d)$ or $(n, {\cal S}_1, {\cal S}_0, {\cal Z})$. To simplify the  notation, we have omitted the subindex $f$. Interestingly, these two sets of variables are related via the following (unnormalized) orthogonal transformation
\begin{equation}
 \label{eq:Hadamard}
 \left(
 \begin{array}{c}
  n\\{\cal S}_1\\{\cal S}_0\\{\cal Z}
 \end{array}
\right)=
\left(
\begin{array}{rrrr}
  1&1&1&1\\
  1&-1&1&-1\\
  1&1&-1&-1\\
  1&-1&-1&1
 \end{array}
\right)
\left(
 \begin{array}{c}
  a\\b\\c\\d
 \end{array}
\right) = 2H^{\otimes 2}\left(
 \begin{array}{c}
  a\\b\\c\\d
 \end{array}
\right).
\end{equation}
where $H$ stands for the $2\times 2$ Hadamard matrix
\begin{equation}\label{Hadamard}
H:=\frac{1}{\sqrt{2}}
\left(
\begin{array}{cc}
1 & 1\nonumber\\
1 & -1
\end{array}
\right).
\end{equation}
This relation can be generalized to the $(n,m,d)$ scenario with $K$-body correlators, both in the framework of expectation values (\ref{eq:probs2expectations}) and in  the framework of probabilities (\ref{eq:ProbabilitatCondicionada}), but this will be studied elsewhere.

In order to characterize all the vertices of $\mathbbm{P}_2^{{\mathfrak S}_n}$, let us introduce the following set of all the four-tuples in the set of non-negative integers that sum up exactly to $n$:
\begin{equation}
 \label{eq:defthetahedron}
 \mathbbm{T}_n:=\{(a,b,c,d) \in \mathbbm{Z}^4: \ a,b,c,d \geq 0,\ a+b+c+d=n\}.
\end{equation}
Geometrically, $\mathbbm{T}_n$ is the set of points of a simplex with integer coordinates. In the $(n,2,2)$ case, it can be viewed as the set of points of a tetrahedron with integer coordinates. Then, the polytope $\mathbbm{P}_2^{{\mathfrak S}_n}$ can be defined as the convex hull of $\varphi({\mathbbm T}_n)$,
where $\varphi$ is the following map:
\begin{equation}
 \label{eq:parametrization}
 \begin{array}{cccc}
  \varphi:&{\mathbbm T}_n&\longrightarrow&{\mathbbm P}_2^{{\mathfrak S}_n}\\
  &(a,b,c,d)&\mapsto&({\cal S}_0, {\cal S}_1, {\cal S}_{00},{\cal S}_{01}, {\cal S}_{11})
 \end{array}.
\end{equation}
Importantly, one can discard most of the points in $\mathbbm{T}_n$, as they get mapped onto the interior of ${\mathbbm P}_2^{{\mathfrak S}_n}$. In fact, as we prove in the following theorem, the vertices of ${\mathbbm P}_2^{{\mathfrak S}_n}$ are completely characterized by those four-tuples that belong to the \textit{boundary} of $\mathbbm{T}_n$.
Let us denote by $\mathrm{Ext}(\mathbbm{P}_2^{{\mathfrak S}_n})$ the set of vertices of $\mathbbm{P}_2^{{\mathfrak S}_n}$ and let us define
\begin{equation}
\partial \mathbbm{T}_n:=\{(a,b,c,d) \in \mathbbm{Z}^4: \ a,b,c,d \geq 0,\ a+b+c+d=n,\ abcd=0\},
\end{equation}
which is the set of points of $\mathbbm{T}_n$ with at least one coordinate equal to $0$.

\begin{thm}
 \label{thm:vertexs}
 For all $p=(a,b,c,d)\in {\mathbbm T}_n$, the following equivalence holds:
 \begin{equation}
  p \in \partial {\mathbbm{T}_n} \iff \varphi(p) \in \mathrm{Ext}(\mathbbm{P}_2^{{\mathfrak S}_n}).
 \end{equation}
\end{thm}
A proof of this can be found in \cite{Science}. For completeness, we have included it also in Appendix \ref{AppA}. Theorem \ref{thm:vertexs} enables us to construct $\mathbbm{P}_2^{{\mathfrak S}_n}$ by taking the convex hull of $\varphi(\partial {\mathbbm T}_n)$ instead of the convex hull of $\varphi({\mathbbm T}_n)$.
Hence, the total number of vertices is
\begin{equation}
\label{eq:actualnumberofvertices}
 |\mathrm{Ext}(\mathbbm{P}_2^{{\mathfrak S}_n})|=|\partial \mathbbm{T}_n| = |{\mathbbm T}_n| - |{\mathbbm T}_n \setminus \partial {\mathbbm T}_n| = |{\mathbbm T}_n| - |{\mathbbm T}_{n-4} | = {n+3 \choose 3} - {n-1 \choose 3}=2(n^2+1).
\end{equation}
Hence, we have improved the bound from $O(n^3)$ in (\ref{eq:exemple_Sym}) to $O(n^2)$ in (\ref{eq:actualnumberofvertices}). The bound (\ref{eq:actualnumberofvertices}) is now exact due to Theorem \ref{thm:vertexs} and the fact that $\varphi$ is invertible, so any vertex of $\mathbbm{P}_2^{{\mathfrak S}_n}$ can be generated from a unique tuple $(a,b,c,d) \in \partial {\mathbbm T_n}$.
We may interpret this  in terms of the graph introduced in Figure \ref{fig:graph}; it means that a coloring in which all the different colors appear cannot correspond to a vertex of $\mathbbm{P}_2^{{\mathfrak S}_n}$.
Thus, the coordinates of all the vertices of ${\mathbbm{P}_2^{{\mathfrak S}_n}}$ are of the form
\begin{eqnarray}
 \label{eq:coordinatesS0}
 {\cal S}_0 = a+b-c-d\\
 \label{eq:coordinatesS1}
 {\cal S}_1 = a-b+c-d\\
 \label{eq:coordinatesS00}
 {\cal S}_{00}=({\cal S}_0)^2-n\\
 \label{eq:coordinatesS01}
 {\cal S}_{01}={\cal S}_0 {\cal S}_1 - {\cal Z}\\
 \label{eq:coordinatesS11}
 {\cal S}_{11}=({\cal S}_1)^2-n,
\end{eqnarray}
for $(a,b,c,d)\in \partial{\mathbb T}_n$.

\section{Classes of Bell Inequalities}
\label{sec:classes}
Here we show in detail how the Bell inequalities presented in \cite{Science} were found. Later we will also
characterize its tightness.

In order to search for particular classes of Bell inequalities, we will use the parametrization we have derived in (\ref{eq:coordinatesS0}-{\ref{eq:coordinatesS11}}). It follows from Theorem \ref{thm:vertexs} that ${\mathbbm{P}_2^{{\mathfrak S}_n}}=\mathrm{conv}(\varphi(\partial {\mathbbm T}_n))$, where $\mathrm{conv}$ stands for convex hull. However, if one relaxes the condition on the numbers $(a,b,c,d)$ from being integer to being real, one obtains a new convex object $-$ not a polytope, but easier to characterize. Let us consider the following sets:
\begin{eqnarray}
 \label{eq:continuousT}
 \mathbf{T}_n=\{(a,b,c,d)\in {\mathbbm R}^4:\ a,b,c,d \geq 0,\ a+b+c+d=n\}
 \end{eqnarray}
 and
 \begin{eqnarray}
 \label{eq:continuousPartialT}
 \partial\mathbf{T}_n=\{(a,b,c,d)\in {\mathbbm R}^4:\ a,b,c,d \geq 0,\ a+b+c+d=n,\ abcd=0\}.
\end{eqnarray}
We shall define ${\mathbf{P}_2^{{\mathfrak S}_n}}=\mathrm{conv}(\varphi({\mathbf{T}_n}))$. Since the proof of Theorem \ref{thm:vertexs} also applies to the continuous case, $\mathrm{conv}(\varphi({\mathbf{T}_n}))=\mathrm{conv}(\varphi(\partial{\mathbf{T}_n}))$.  Because $\partial {\mathbbm T}_n \subset \partial {\mathbf T}_n$ (or ${\mathbbm T}_n \subset {\mathbf T}_n$) one has then ${\mathbbm{P}_2^{{\mathfrak S}_n}}\subseteq {\mathbf{P}_2^{{\mathfrak S}_n}}$. Consequently, a Bell inequality valid for ${\mathbf{P}_2^{{\mathfrak S}_n}}$ is also a Bell inequality valid for ${\mathbbm{P}_2^{{\mathfrak S}_n}}$. We expect that finding Bell inequalities of the set ${\mathbf{P}_2^{{\mathfrak S}_n}}$ is an easier task than finding the ones corresponding to ${\mathbbm{P}_2^{{\mathfrak S}_n}}$, at the price of having inequalities not as optimal as the ones of $\partial {\mathbbm T}_n \subset \partial {\mathbf T}_n$.

In passing let us notice that the sets $\varphi(\partial {\mathbf T}_n)$ and $\varphi({\mathbf T}_n)$ are semialgebraic, since they are defined through polynomial equalities and inequalities. The characterization of convex hulls of semialgebraic sets is a well studied subject \cite{Parrilo}. Although the exact characterization is an NP-hard problem \cite{Parrilo}, there exist efficient approximations with semi-definite programming techniques in terms of the so-called theta bodies. These techniques can be applied to any $(n,m,d)$ scenario with $K$-body correlators when the symmetry group $G$ is ${\mathfrak S}_n$ and they will be studied elsewhere.
These approximations are in the spirit of the Navascues-Pironio-Ac\'in (NPA) hierarchy \cite{NPA}, which bounds the set of quantum correlations ${\mathbf{Q}}$ appearing in (\ref{eq:inclusions}) from above by providing a hierarchy of sets ${\mathbf Q}\subseteq \ldots \subseteq {\mathbf{Q}_2}\subseteq {\mathbf {Q}_{1+AB}}\subseteq {\mathbf Q}_1$. Each level of the hierarchy better approximates ${\mathbf Q}$ and ${\mathbf Q}_k$ asymptotically converges to $\mathbf{Q}$ as $k\rightarrow \infty$. When the NPA hierarchy certifies that a set of correlations is outside ${\mathbf Q}_k$ for some $k$, then these correlations cannot be realized with quantum resources. In our case, a characterization of ${\mathbf{P}_2^{{\mathfrak S}_n}}$ through the so-called theta bodies produces also a hierarchy of sets ${\mathbf{P}_2^{{\mathfrak S}_n}}\subseteq \ldots \subseteq \Theta_2 \subseteq \Theta_1$ that in this case certifies that a set of correlations which is outside of $\Theta_k$ for some $k$ cannot be simulated through shared randomness and it must be necessarily nonlocal.

The approach that we shall follow here in order to derive a class of Bell inequalities is based on the following: Eqs. (\ref{eq:coordinatesS0}-{\ref{eq:coordinatesS11}}) are polynomials of at most degree $2$ in the variables $a,b,c$ and $d$. They define a $3$-dimensional manifold (because of the constraint $a+b+c+d=n$) in a $5$-dimensional space. By finding a tangent hyperplane to it, i.e., a plane of the following form:
\begin{equation}
 \label{eq:tangent}
 \alpha {\cal S}_0 + \beta {\cal S}_1 + \frac{\gamma}{2}{\cal S}_{00} + \delta{\cal S}_{01}+ \frac{\varepsilon}{2}{\cal S}_{11}  + \beta_c = 0,
\end{equation}
and then proving that the LHS of (\ref{eq:tangent}) is positive on ${\mathbf T}_n$, we shall obtain the coefficients $\alpha, \beta, \gamma, \delta, \varepsilon$ and the classical bound $\beta_c$ that define a valid Bell inequality.
Theorem \ref{thm:vertexs} also hints us to look for facets of ${\mathbf T}_n$, that is, points $(a,b,c,d)$ with coordinates satisfying $abcd=0$. We aim at obtaining a Bell inequality which is tangent on vertices of $\mathbf{P}_2^{\mathfrak{S}_n}$ represented by four-tuples belonging to a single facet of ${\mathbf T}_n$,  for example the one defined by $a=0$. Let us, however, remark that not all Bell inequalities need to come from the same facet, as $a=0$ implies $abcd=0$ but the converse does not hold, so the class of inequalities we shall derive will display this feature -- this is discussed in detail in the next sections.

To be more rigorous, we consider the following Lagrangian function
\begin{equation}
 \label{eq:lagrangian}
 {\cal L}=\alpha {\cal S}_0 + \beta {\cal S}_1 + \frac{\gamma}{2}{\cal S}_{00} + \delta{\cal S}_{01}+ \frac{\varepsilon}{2}{\cal S}_{11} + \lambda(a+b+c+d-n) + \mu a,
\end{equation}
and we wish to find its minimum $-\beta_c = \min_{a,b,c,d, \lambda, \mu \in {\mathbbm R}}{\cal L}$.
The necessary condition for an extremal value of ${\cal L}$ to exist in the point
$\boldsymbol{x}^*=(a^*,b^*,c^*,d^*,\lambda^*,\mu^*)$ is that the following
system of differential equations
\begin{equation}\label{OrchestradeVic}
\frac{\partial \mathcal{L}}{\partial x}=0\qquad (x=a,b,c,d,\lambda,\mu)
%
%\frac{\partial \mathcal{L}}{\partial \lambda}=0,\qquad
%\frac{\partial \mathcal{L}}{\partial \mu}=0
\end{equation}
is satisfied in $\boldsymbol{x}^*$. Since all the permutationally invariant mean values appearing in Eq. (\ref{eq:lagrangian}) are at most quadratic in the variables $a,b,c,d$, one finds that (\ref{OrchestradeVic}) can be rewritten as a simple system of linear equations
\begin{equation}
 \label{eq:conditionforextremum}
 \left(
 \begin{array}{cccccc}
  \xi_+&\zeta&-\zeta&-\xi_+&1&1\\
  \zeta&\xi_-&-\xi_-&-\zeta&1&0\\
  -\zeta&-\xi_-&\xi_-&\zeta&1&0\\
  -\xi_+&-\zeta&\zeta&\xi_+&1&0\\
  1&1&1&1&0&0\\
  1&0&0&0&0&0
 \end{array}
\right)
\left(
\begin{array}{c}
 a^*\\b^*\\c^*\\d^*\\ \lambda^*\\ \mu^*
\end{array}
\right) =
\left(
\begin{array}{c}
 -\alpha - \beta + \xi_+/2\\
 -\alpha + \beta + \xi_-/2\\
 \ \ \ \ \! \alpha - \beta + \xi_-/2\\
 \ \ \ \ \! \alpha + \beta + \xi_+/2\\
 0\\0
\end{array}
\right),
\end{equation}
where $\xi_{\pm}:=\gamma \pm 2 \delta + \varepsilon$ and $\zeta := \gamma - \varepsilon$.

For a generic choice of $\xi_{\pm}$ and ${\zeta}$, the matrix from (\ref{eq:conditionforextremum}) is non-singular, leading to a unique solution $(a^*, b^*, c^*, d^*)$, i.e., a Bell inequality which is tangent to $\mathbf{P}_2^{\mathfrak{S}_n}$ only at a single point. This is what one would expect from a generic choice of the coefficients $\alpha, \beta, \gamma, \delta$ and $\varepsilon$. And, although such solutions would lead us to proper Bell inequalities, we would like to find Bell inequalities that are tangent to $\mathbf{P}_2^{\mathfrak{S}_n}$ at more points, and therefore stronger in detecting nonlocality. For this purpose, we can impose some constraints on the
coefficients $\gamma$, $\delta$ and $\varepsilon$ appearing in Eq. (\ref{eq:conditionforextremum}). Precisely, the number of solutions of Eq. (\ref{eq:conditionforextremum}) is given by Rouch\'e-Frobenius theorem, which says that a linear system of equations $A\vec{b}=\vec{c}$ is compatible if, and only if, the rank of $A$ is the same as the rank of the extended matrix $A|{\vec c}$; the solution is unique if, and only if, $\det A \neq 0$. As we do not want a unique solution, the first condition is that $\det A=0$, which when applied to Eq. (\ref{eq:conditionforextremum}) is
\begin{equation}
 \label{eq:condition1}
 \delta^2-\gamma \varepsilon = 0.
\end{equation}
With this condition the rank of $A$ is not greater than $5$, and it is exactly $5$ if we further impose $\gamma\delta\varepsilon \neq 0$, which we can assume to be always fulfilled. Otherwise, Eq. (\ref{eq:condition1}) would just produce trivial Bell inequalities.
In order to guarantee a solution of (\ref{eq:conditionforextremum}), we have to ensure that also $\mathrm{rank}(A|{\vec c})=5$. This is done by imposing that the $5\times 5$ minors of $A|\vec{c}$ vanish, which is equivalent to the following condition
\begin{equation}
 \label{eq:condition2}
 \delta(\beta+\delta)=\varepsilon(\alpha + \delta).
\end{equation}
In this case, the classical bound reads
\begin{equation}
 \label{eq:classicalbound}
 \beta_c = -{\cal L}(a^*,b^*,c^*,d^*,\lambda^*,\mu^*)=\frac{(\beta+ \delta)^2 + n(\delta - \varepsilon)^2}{2\varepsilon}.
\end{equation}
If the minimum is global, we have obtained a Bell inequality which is tangent to ${\mathbf P}_2^{{\mathfrak S}_n}$ and valid for ${\mathbbm P}_2^{{\mathfrak S}_n}$, although in general not tangent to the latter. By repeating the same argument, but over the facets $b=0$, $c=0$ and $d=0$ we would get similar expressions as in (\ref{eq:condition2}) and (\ref{eq:classicalbound}), with swaps of parameters and changes of signs. In the following theorem \cite{Science}, we generalize all of them and we do a small correction to (\ref{eq:classicalbound}) ensuring that the Bell inequality is tangent to ${\mathbbm P}_2^{{\mathfrak S}_n}$, and not to ${\mathbf P}_2^{{\mathfrak S}_n}$. We also prove that the minimum is global, i.e., the Bell inequality is valid.

\begin{thm}
 \label{thm:class}
 For any $\sigma \in \{-1,1\}$ and $x,y, \mu \in \mathbbm{N}$ such that $\mu$ has opposite parity to $x$ ($y$) if $n$ is even (odd), define the parameters of (\ref{eq:2bodySymgeneral}) as
 \begin{equation}
  \label{eq:parametersclass}
  \alpha_{\pm}=x[\sigma \mu \pm (x+y)],\quad \beta = \mu y,\quad \gamma = x^2,\quad \delta = \sigma x y,\quad \varepsilon = y^2.
 \end{equation}
Then the classical bound of the resulting Bell inequality, for which it is tangent to ${\mathbbm P}_2^{{\mathfrak S}_n}$, is
\begin{equation}
 \label{eq:classicalboundclass}
 \beta_c=\frac{1}{2}[n(x+y)^2+(\sigma \mu \pm x)^2-1].
\end{equation}
\end{thm}
Theorem \ref{thm:class} shows that the few-parameter class of Bell inequalities defined by Eqs. (\ref{eq:parametersclass}, \ref{eq:classicalboundclass}), which can be explicitly written as
\begin{equation}\label{FuetdeVic}
I_{x,y,\mu,\sigma}^{\pm}:=x[\sigma\mu\pm(x+y)]\mathcal{S}_0+\mu y\mathcal{S}_1+\frac{x^2}{2}\mathcal{S}_{00}+
\sigma xy\mathcal{S}_{01}+\frac{y^2}{2}\mathcal{S}_{11}\geq -\frac{1}{2}[n(x+y)^2+(\sigma\mu\pm x)-1],
\end{equation}
is tangent to ${\mathbbm P}_2^{{\mathfrak S}_n}$ in at least one vertex, but it does not guarantee that it is tight (that is, defines a facet of ${\mathbbm P}_2^{{\mathfrak S}_n}$) and as such optimal. Notice that in this case
proving that a Bell inequality defines a facet of ${\mathbbm P}_2^{{\mathfrak S}_n}$ requires showing that
the set of vertices of ${\mathbbm P}_2^{{\mathfrak S}_n}$ saturating it spans a four-dimensional affine space.
In the next theorem we are going to characterize the vertices of ${\mathbbm P}_2^{{\mathfrak S}_n}$ which saturate the Bell inequality (\ref{FuetdeVic}). To this end, it will be useful to rename the faces of the tetrahedron ${\mathbbm T}_n$ according to the choice of $\sigma$ and $\pm$ in order to obtain a more general result:
\begin{equation}
\begin{array}{c|cccc}
 \sigma&+&+&-&-\\
 \pm&+&-&+&-\\
 \hline
 r&b&c&a&d\\
 s&a&d&b&c\\
 t&d&a&c&b\\
 u&c&b&d&a
\end{array}.
  \label{eq:renaming}
\end{equation}
Now, Eq. (\ref{eq:Hadamard}) is generalized to
 \begin{equation}
  \left(
  \begin{array}{c}
   n\\
   {\cal S}_1\\
   {\cal S}_0\\
   {\cal Z}
  \end{array}
  \right)=2\left(
  \begin{array}{cccc}
   1&&&\\
   &\mp \sigma&&\\
   &&\pm1&\\
   &&&-\sigma
  \end{array}
\right)\cdot H^{\otimes 2}\cdot\left(
\begin{array}{c}
r\\s\\t\\u
\end{array}
\right)=
\left(
  \begin{array}{cccc}
   1&1&1&1\\
   \mp \sigma&\pm \sigma&\mp \sigma&\pm \sigma\\
   \pm1&\pm1&\mp1&\mp1\\
   -\sigma&\sigma&\sigma&-\sigma
  \end{array}
\right)\left(
\begin{array}{c}
r\\s\\t\\u
\end{array}
\right),
\label{eq:Hadamard2}
 \end{equation}
 where $\sigma$ and $\pm$ are independent signs and $H$ is the $2\times 2$ Hadamard matrix defined in Eq. (\ref{Hadamard}). Having this, we can now state the following theorem.
 %
 %Denoting by
 %$I$ the Bell inequality defined in Theorem \ref{thm:class}, let us now state the following theorem.
%
\begin{thm}
\label{thm:numberofvertices}
 Let us assume that $x$ and $y$ in the Bell inequality (\ref{FuetdeVic}) are coprimes and let us define the following quantities:
 \begin{equation}
 \begin{array}{cll}
  K_r(\tau)&:=&[\pm n(y-x) + \sigma \mu \pm x + \tau]/2,\\
  t_0(\tau)&:=&\pm y^{-1}K_r(\tau) \mod x,\\
  u_0(\tau)&:=&[n x \pm K_r(\tau)-(x+y)t_0(\tau)]/x,\\
  s_0(\tau)&:=&[\mp K_r(\tau) + y t_0(\tau)]/x,
 \end{array}
 \end{equation}
where $\tau \in \{-1,1\}$.
If a vertex of ${\mathbbm P}_2^{{\mathfrak S}_n}$ saturates (\ref{FuetdeVic}), then it is of the form
$\varphi([r,s,t,u])$, where
 \begin{equation}
  [r,s,t,u]=[0,s_0(\tau), t_0(\tau), u_0(\tau)]+ k[0,x,y,-(x+y)], \qquad k\in \mathbbm{Z}.
  \label{eq:reparametrization}
 \end{equation}
 Furthermore, the number of vertices of  ${\mathbbm P}_2^{{\mathfrak S}_n}$ saturating (\ref{FuetdeVic}) is given by
 \begin{equation}
  N_S:=\sum_{\tau =\pm 1}\max\left\{0,\left\lfloor\frac{u_0(\tau)}{x+y}\right\rfloor-\max\left\{0,\left\lceil\frac{-s_0(\tau)}{y}\right\rceil\right\}+1\right\}. %Check if +1 is in the right parenthesis. OK!
  \label{eq:numberoftightvertices}
 \end{equation}
\end{thm}

Clearly, a necessary condition for the inequality (\ref{FuetdeVic}) to be tight is that $N_S \geq 5$ as $5$ is the minimal number of points to generate a $4$-dimensional affine subspace. We have checked numerically that this condition is also sufficient for all the instances of $n$ for which we could compute all facets of ${\mathbbm P}_2^{{\mathfrak S}_n}$, which is $n \leq 33$. To prove sufficiency, let us make the following observations: The proof of Theorem \ref{thm:numberofvertices} indicates that the vertices saturating (\ref{FuetdeVic}) form a \textit{zig-zag} pattern that alternates between $\tau = 1$ and $\tau = -1$. For a given $\tau$, more than three different points are always linearly dependent. Hence, one must pick three points for $\tau$ and two for $-\tau$. Short algebra shows then that the matrix formed by the coordinates of such points (in particular, if chosen with consecutive $k$) has maximal rank.

\section{Quantum violation}
\label{sec:quantum}

In this section we analyze the quantum violation of permutationally invariant 2-body Bell inequalities for the $(n,2,2)$ scenario; \textit{i.e.}, Bell inequalities of the form (\ref{eq:2bodySymgeneral}). If the left-hand side of (\ref{eq:2bodySymgeneral}) is smaller than $0$ for some $\vec{p}$, then the inequality is violated and the nonlocality of this $\vec{p}$ is detected. To show that there is a quantum violation of such an inequality, one must find a quantum state (taken from a Hilbert space of some dimension $D$) and POVM elements such that the probabilities obtained through (\ref{eq:BornsRule}) violate (\ref{eq:2bodySymgeneral}).

In order to find out what is the best performance of a Bell inequality, we are interested in finding the states and measurements giving the maximal quantum violation of such a Bell inequality, however, the dimension $D$ of the corresponding state depends in general on the scenario under consideration. Luckily, in the $(n,2,2)$ scenario, $D=2$ is enough, because in order to find the maximal quantum violation of a Bell inequality with two dichotomic measurements per party, it suffices to consider qubits and traceless real observables \cite{Masanes, TonerVerstraete}.

Let us define by $\vec{\sigma}=[\sigma_{x},\sigma_{y}, \sigma_{z}]$ the usual vector of the Pauli matrices and let us denote by $\vec{\sigma}^{(i)}$ the same vector acting on the $i$th site of an $n$-partite system. Any traceless qubit observable with eigenvalues $\pm 1$ can be written as $\hat{\mathbf{n}}\cdot\vec{\sigma}$
with $\hat{\mathbf{n}}\in\mathbbm{R}^3$ such that $|\hat{\mathbf{n}}|=1$. Then, any real traceless single-qubit observable ${\cal M}^{(i)}_{x_i}$ with eigenvalues $\pm 1$ can be parametrized as
\begin{equation}
 \label{eq:tracelessobs}
 {\cal M}^{(i)}_{x_i}= \cos (\theta_{x_i}^{(i)}) \sigma_z^{(i)} + \sin (\theta_{x_i}^{(i)}) \sigma_x^{(i)},
\end{equation}
where $\theta_{x_i}^{(i)} \in [0,2\pi)$. To simplify the notation, we shall denote $\theta_{x_i}^{(i)}$ as $\varphi_i$ when $x_i=0$ or as $\theta_i$ when $x_i=1$.

The so-called Bell operator is the quantum observable corresponding to the left hand side of (\ref{eq:2bodySymgeneral}), where the observables (\ref{eq:tracelessobs}) are used in the definitions of (\ref{eq:def1bodysym}, \ref{eq:def2bodysym}). We shall denote it ${\cal B}(\{\varphi_i,\theta_i\})$. The maximal quantum violation is then given by the minimal eigenvalue value of ${\cal B}(\{\varphi_i,\theta_i\})$, when varying the angles $\varphi_i$ and  $\theta_i$, and the corresponding eigenvector is the quantum state maximally violating the corresponding Bell inequality (\ref{eq:2bodySymgeneral}).

If a Bell inequality displays some symmetry with respect to an exchange of parties (i.e., it consists only of $G$-invariant correlators as defined in (\ref{eq:G-invariant-correlator}) for some group $G$), the pure state $\ket{\psi}$ for which the maximal quantum violation is achieved does not need to display the same symmetry (see  \cite{TIpaper}), nor do the observables be the same at each site (\textit{i.e.}, ${\cal M}_{x_i}^{(i)}$ be independent of $i$). However, one can always consider a mixed state $\rho$ and observables (at the expense of increasing the local dimension $D$) such that $\rho$ is invariant with respect to any permutation in $G$ and the extended observables are the same at each site. However, in the case we are considering, $G={\mathfrak S}_{n}$, numerical results for a small number of parties suggest that the maximal quantum violation is already achieved with a permutationally invariant pure qubit state and the same pair of measurements at each site. This is very convenient for our analysis because in the latter case the Bell operator ${\cal B}(\varphi, \theta)$ depends on only two angles $\varphi=\varphi_i$ and $\theta=\theta_i$, and it is permutationally invariant and therefore can be block-diagonalized.
Finally, as the Bell inequality (\ref{eq:2bodySymgeneral}) contains only $2$-body correlators, we are able to prove that in a certain basis, ${\cal B}(\varphi, \theta)$ is pentadiagonal. All this implies that its minimal eigenvalue over $\varphi,\theta$ can be found numerically for a large number of parties in an efficient manner due to sparsity of the matrix and the fact that we have reduced the optimization to a single variable problem.

\subsection{Block-diagonalization of the permutationally invariant Bell opeator}
When the same pair of measurements is taken at each site, the Bell operator is invariant under any permutation of the parties. It is a well-known result that under these circumstances it can be block diagonalized. The reason for that lies in the result known as Schur-Weyl duality \cite{ChristandlThesis}. Consider the permutation operator $V_{\sigma}$ acting on an $n$-qubit Hilbert space:
\begin{equation}
V_\sigma : \ket{i_0\cdots i_{n-1}} \mapsto \ket{i_{\sigma^{-1}(0)}\cdots i_{\sigma^{-1}(n-1)}},
\end{equation}
where $\sigma \in {\mathfrak S}_n$. $V_{\sigma}$ is a unitary representation of $\sigma\in{\mathfrak S}_n$ acting on the $n$ qubit Hilbert space. Consider also the tensor product representation $W^U=U^{\otimes n}$ of the special unitary group. The Schur-Weyl duality states that the Hilbert space splits into blocks, on which these two representations commute (in representation theory terms, the action of these two groups is dual). This block splitting is the following:
\begin{equation}
 \label{eq:Schur-Weyl}
 \left({\mathbbm C}^2\right)^{\otimes n} \cong \bigoplus_{J=J_0}^{n/2}{\cal H}_J\otimes {\cal K}_J,
\end{equation}
where $2 J_0\equiv n \ \mathrm{mod}\ 2$, the Hilbert spaces ${\cal H}_J$ are of dimension $2J+1$ and the dimension of ${\cal K}_J$ is $1$ if $J=n/2$ and
\begin{equation}
 \label{eq:multiplicity}
 \dim{\cal K}_J={n\choose n/2-J}-{n\choose n/2-J-1}
\end{equation}
otherwise.

A permutationally invariant state $\rho$ is such that for all $\sigma \in {\mathfrak S}_n$, $\rho = V_{\sigma} \rho V_{\sigma}^\dagger$. In the basis given by (\ref{eq:Schur-Weyl}), $\rho$ has the following form:
\begin{equation}
 \label{eq:blocksstate}
 \rho=\bigoplus_{J=J_0}^{n/2} \frac{p_J}{\dim_{{\cal K}_J}} \rho_J \otimes {\mathbbm 1}_{\dim_{{\cal K}_J}},
\end{equation}
where $p_J$ is a probability distribution and $\rho_J$ are density matrices.
Hence, $\rho$, which is a $2^n\times 2^n$ density matrix, is completely characterized by the $\left\lfloor n/2\right\rfloor$ blocks $\rho_J$, the biggest one having size $(n+1)\times (n+1)$, corresponding to the symmetric subspace
$(\mathbbm{C}^2)^{\ot n}$.

The blocks $\rho_J$ can be obtained by projecting the Bell operator onto the following basis \cite{MoroderPI}:
\begin{equation}
\{\ket{D^{k}_{2J}}\otimes\ket{\psi^-}^{\otimes m/2}\}_{k=0}^{2J},
\label{eq:basisblocks}
\end{equation}
where $m=(n-2J)$ is an even number, $\ket{D^k_n}$ are the so-called symmetric Dicke states (of $n$ qubits and $k$ excitations) which are symmetric superpositions of $k$ qubits in the state $\ket 1$ and $n-k$ qubits in the state $\ket 0$:
\begin{equation}
\label{eq:defDicke}
\ket{D^k_n}:={n\choose k}^{-1/2}\sum_{\sigma\in\mathfrak{S}_n}V_{\sigma}\left(\ket{0}^{\otimes n-k}\ket{1}^{\otimes k}\right),
\end{equation}
and $\ket{\psi^-}$ is the singlet state given by
\begin{equation}
\label{eq:defSinglet}
\ket{\psi^-}:=\frac{\ket{01}-\ket{10}}{\sqrt{2}}.
\end{equation}

In the following theorem, we give the analytic form of each block, when the same set of real traceless observables is taken at each site:
\begin{equation}
{\cal M}_0^{(i)}=\cos \varphi \sigma_z^{(i)}+ \sin \varphi \sigma_x^{(i)}, \qquad {\cal M}_1^{(i)}=\cos \theta \sigma_z^{(i)}+ \sin \theta \sigma_x^{(i)}.
\label{eq:same_measurements}
\end{equation}

\begin{thm}
\label{thm:Blocks}
The $J$-th Block ${\cal B}_J(\varphi, \theta)$ of the Bell operator
corresponding to a $2$-body symmetric Bell inequality (\ref{eq:2bodySymgeneral}), with measurements given by (\ref{eq:same_measurements}), has elements $({\cal B}_J(\varphi, \theta))^k_l$ ($k$-th row, $l$-th column)
given by
 \begin{equation}
\label{eq:penta-block}
  ({\cal B}_J(\varphi, \theta))^k_l = d_k \delta_{k,l} + u_k \delta_{k,l-1}+u_l
\delta_{k-1,l} + v_k \delta_{k,l-2} + v_l\delta_{k-2,l},
 \end{equation}
where $0\leq k,l,\leq 2J$, $d_k$ ($0\leq k \leq 2J$) are the diagonal elements, $u_k$ ($0\leq k \leq 2J-1$) correspond to the elements of the upper (lower) diagonal
and $v_k$ ($0 \leq k \leq 2J-2$) stand for the elements of the second upper
(lower) diagonal and $\delta_{k,l}$ is the Kronecker delta function
($\delta_{k,l}=1 \iff k=l$, otherwise $\delta_{k,l}=0$).
The coefficients $d_k, u_k, v_k$ are explicitly given by
\begin{eqnarray}
 \label{dk}d_k &:=& \beta_c + (2J-2k)A +  [(2J-2k)^2-n]B/2 + [2k(2J-k)-m]C/2,\\
 u_k &:=& [A' + (2J-1-2k)D]\sqrt{(2J-k)(k+1)},\\
 v_k&:=&C\sqrt{(2J-k)(2J-k-1)(k+1)(k+2)}/2,
\end{eqnarray}
where the parameters $A,A',B,C,D$ depend only on the Bell Inequality coefficients and the
measurements' angles:
\begin{eqnarray}
 A&:=&\alpha \cos \varphi + \beta \cos \theta,\\
 A'&:=&\alpha \sin \varphi + \beta \sin \theta,\\
 %B&=&\frac{\gamma}{2} \cos^2 \varphi + \delta \cos \varphi \cos \theta + \frac{\varepsilon}{2}\cos^2\theta,\\
 B&:=&\gamma \cos^2 \varphi + 2\delta \cos \varphi \cos \theta + \varepsilon\cos^2\theta,\\
\label{eq:def-C}
 %C&=&\frac{\gamma}{2} \sin^2 \varphi + \delta \sin \varphi \sin \theta + \frac{\varepsilon}{2}\sin^2\theta,\\
 C&:=&\gamma \sin^2 \varphi + 2\delta \sin \varphi \sin \theta + \varepsilon\sin^2\theta,\\
 D&:=&\gamma \cos \varphi \sin \varphi + \delta \cos \varphi \sin \theta + \delta
\cos \theta \sin \varphi + \varepsilon \cos \theta \sin \theta.
\end{eqnarray}
\end{thm}

Observe that (\ref{eq:penta-block}) ensures that every block of the Bell operator is pentadiagonal, and this comes from the fact that we are considering a symmetric Bell inequality with \textit{at most} $2$-body correlators. From the proof of Theorem \ref{thm:bigcalculation} in Appendix B, it is clear that if the Bell inequality involved $3$-body correlators, it would be heptadiagonal; if it involved $4$-body correlators, enneadiagonal, and so on.
This argument also shows that nonlocality of some of the highly entangled multipartite states cannot be revealed through Bell inequalities involving only few-body correlations: if we consider the GHZ state $(1/\sqrt{2})(\ket{0}^{\ot n}+\ket{1}^{\ot n})$ which is supported on the last block of (\ref{eq:blocksstate}) corresponding to $J=n/2$, the density matrix of the state has only four terms, two of them lie in the diagonal and the other two are the coherences corresponding to $\ket{D^0_n}\!\bra{D^n_n}$, $\ket{D^n_n}\!\bra{D^0_n}$, which can only be reached through a full-body correlator. Hence, for a symmetric $2$-input $2$-output inequality, which is not full-body, the GHZ state is indistinguishable from the classical mixture $(1/2)(\ket{0}\!\bra{0}^{\otimes n}+\ket{1}\!\bra{1}^{\otimes n})$. However, there are other states, which show nonlocality in a robust way, some of them even experimentally realizable, as we shall discuss it in further sections.

The form (\ref{eq:penta-block}) also has important numerical implications, as the Bell operator can be easily stored in a sparse matrix. This allowed us to show nonlocality of genuinely multipartite entangled states of more than $10^4$ qubits \cite{Science}. Recall that, in order to find the largest quantum violation with these settings (\textit{i.e.}, quantum state and measurements) of a particular Bell inequality, we have to vary the angles $\varphi, \theta$ in (\ref{eq:penta-block}) and look for the smallest eigenvalue of the $\lfloor n/2\rfloor$ possible blocks. Actually, as it is proven in the following fact, the spectrum of the Bell operator depends only on the difference between $\varphi$ and $\theta$.
\begin{thm}\label{thm:unitaryrotation}
Let $\mathcal{B}(\varphi,\theta)$ be the permutationally invariant Bell operator corresponding to the measurements given in Eq. (\ref{eq:same_measurements}). Then, for any $c\in\mathbbm{R}$, $\mathcal{B}(\varphi+c,\theta+c)=[U(c)]^{\ot n}\mathcal{B}(\varphi,\theta)[U(c)^{\dagger}]^{\ot n}$ with $U(c)$ being a unitary operation given by
 \begin{equation}
 \label{eq:unitary}
  U(c):=\left(
  \begin{array}{rr}
  \cos(c/2)&-\sin(c/2)\\
  \sin(c/2)& \cos(c/2)
  \end{array}
  \right).
 \end{equation}
\end{thm}

Theorem \ref{thm:unitaryrotation} implies that the spectrum of $\mathcal{B}(\varphi,\theta)$ depends only on $\varphi-\theta$. This fact will later be used to simplify our analysis of the maximal quantum violation of the corresponding Bell inequalities. Another simplification comes from (\ref{eq:penta-block}) and the form of the constant $C$ given in Eq. (\ref{eq:def-C}). Concretely, there is a convenient choice of $\varphi$ and $\theta$ for which $C=0$, which makes the Bell operator $\mathcal{B}(\varphi,\theta)$ tridiagonal. Such form is given in the following corollary:
%
%\textbf{EITHER IT IS A COROLLARY OR A THEOREM. ACTUALLY, I WAS WONDERING IF ALL WHAT IS SHOWN ARE THEOREMS, OR SOME %SHOULD BE CALLED LEMMAS, OR COROLLARIES. ANY THOUGHT?}

\begin{thm}
\label{thm:corollary}
Let $\kappa = \theta - \varphi$. Then the Bell operator ${\cal B}(\varphi,
\theta)$ is tridiagonal when $$\theta_{\pm} = \arctan\left(\frac{\gamma \sin
\kappa}{\gamma \cos \kappa +  \delta \pm \sqrt{\delta^2-\gamma
\varepsilon}}\right).$$
\end{thm}

\subsection{Analytical class of states}

Here we present an analytical class of states for which the Bell inequalities from the family  (\ref{FuetdeVic}) are violated. We also show that this analytical class of states converges to the
 maximal quantum violation when $n$ becomes large, providing strong numerical evidence for that. Such a state is a Gaussian superposition of Dicke states with variance growing as $O(\sqrt{n})$ and centered at the minimum of the function $d_k$ defined in Eq. (\ref{dk}). We have found that for this class the maximal quantum violation is achieved with a state from the symmetric block ($J=n/2$), although other blocks also lead to quantum violation. In other scenarios, such as the ($n,3,2$) scenario, we have also found some Bell inequalities maximally violated with states supported on the lowest $J$ blocks.

\begin{thm}
\label{thm:qv}

 Let $\ket{\psi_n}=\sum_{k=0}^n \psi_k^{(n)} \ket{D^k_n}$ be a symmetric $n$-qubit
state with the following coefficients in the Dicke basis:
$$\psi_k^{(n)}:=\frac{e^{-(k-\mu)^2/4\sigma}}{\sqrt[4]{2\pi \sigma}},$$ where
$\mu :=n/2+A/(2B-C)$ and $e^{-1}/2\pi\ll \sigma \ll n$.
Moreover, let $\theta(\varphi)$ be the solution of Eq. (\ref{thm:corollary}) guaranteeing that $C=0$.
Then,
%
%the value of $\bra{\psi_n}{\cal B}_{n/2}(\varphi(\theta),
%\theta)\ket{\psi_n}$ is given by
%
 $$\bra{\psi_n}{\cal B}_{n/2}(\varphi,
\theta(\varphi))\ket{\psi_n}=\left[\frac{\beta_c}{n}-\frac{B}{2} +
e^{-1/8\sigma}\left(A'-\frac{AD}{B}\right)\right] n + \left[2B \sigma
-\frac{A^2}{2B} +e^{-1/8\sigma}\left(A'-\frac{AD}{B}\right)\right] + o(\sigma
n^{-1}).$$

\end{thm}
Note that the choice of $\mu$ is motivated by the following fact that $d_k$ is a
quadratic function in $k$ with minimum at $k=\mu$. Hence, if ${\cal B}_{n/2}(\varphi, \theta)$ were diagonal, $d_{[\mu]}$, where $[\cdot]$ is the rounding function, would be its minimal eigenvalue. However, since this value is always positive for the family of inequalities described in Theorem \ref{thm:class}, we must make the off-diagonal elements play some role in order to achieve quantum violation.

\ \\
\noindent{\bf Example.}
Let us illustrate Theorem \ref{thm:qv} with an example. Consider the Bell inequality presented in Ref. \cite{Science}, which is a particular case of the class (\ref{FuetdeVic}) with the choice of parameters $x=y=1$, $\sigma = \pm = -1$ and $\mu = 0$, which reads
\begin{equation}
-2{\cal S}_0 + \frac{1}{2}{\cal S}_{00} - {\cal S}_{01} + \frac{1}{2}{\cal
S}_{11} + 2n\geq 0.
 \label{eq:exampleineq}
\end{equation}
We have that $C=(\sin \varphi - \sin \theta)^2.$
Hence, the solution of $C=0$ which makes ${\cal M}_0 \neq \pm {\cal M}_1$ is to
take $\varphi = \pi - \theta.$ Then,
the values of $A,A',B,C,D$ for $\alpha = -2, \beta = 0, \gamma = 1, \delta = -1,
\varepsilon = 1$ and $\varphi = \pi - \theta$ are

\begin{equation}
\begin{array}{c}
\begin{array}{ccc}
A=2\cos \theta, & A'=-2\sin \theta, & B=4\cos^2 \theta,
\end{array}\\
\begin{array}{cc}
C=0&D=0.
\end{array}
\end{array}
\nonumber
\end{equation}
Therefore, the form of ${\cal B}_{n/2}(\pi - \theta, \theta)$
is then given by
\begin{eqnarray}
d_k&=& 2n(1+\cos \theta + (n-1)\cos^2\theta)-[4\cos\theta (1+2n\cos\theta)]k + 8k^2 \cos^2\theta ,\nonumber\\
u_k&=& -2\sin \theta \sqrt{(n-k)(k+1)},\nonumber\\
v_k&=&0.\nonumber
\end{eqnarray}
Hence, the class of states proposed has parameters $$\mu_n =
\frac{n}{2}+\frac{1}{4\cos \theta},$$
and $\sigma_n \in \Theta(\sqrt{n})$ such that it fulfils the requirements for the validity of the
approximations of Theorem \ref{thm:qv}.
The asymptotic expectation value is then given by Theorem \ref{thm:qv}:
\begin{eqnarray}
\bra{\varphi_n}{\cal B}_{n/2}(\pi - \theta, \theta)\ket{\varphi_n}
\simeq \left(\frac{\beta_c}{n}-\frac{B}{2} + e^{-1/8\sigma}A'\right) n +
\left(2B \sigma -\frac{A^2}{2B} +e^{-1/8\sigma}A'\right) + o(\sigma n^{-1}).
\nonumber
\end{eqnarray}
To find the maximal value of the above expression let us approximate its leading term by
$$\frac{\beta_c}{n} - \frac{B}{2} + e^{-1/8\sigma}A' \simeq
\frac{\beta_c}{n}-\frac{B}{2}+A' = 2(1-\cos^2 \theta - \sin \theta),$$
and notice that the latter attains its maximum for $\theta^* \in \{\pi/6, 5\pi/6\}$.

As a result, the quantum violation of inequality (\ref{eq:exampleineq}) by our analytical class of states
$\ket{\psi_n}$ divided by its classical bound reads (let us take $\sigma = \sqrt n$ for simplicity)
\begin{eqnarray}
\frac{1}{\beta_c}\bra{\varphi_n}{\cal B}_{n/2}(\pi/6, 5\pi/6)\ket{\varphi_n}
\simeq -\frac{1}{4} + 3n^{-1/2} -\frac{3}{4}n^{-1} + o(n^{-3/2}). \nonumber
\end{eqnarray}
Hence, the quantum violation exists for any $n$ and it tends to $-1/4$ when
compared to the classical bound as $n\to \infty$.

\subsection{Accuracy of analytical results}
Let us now compare the analytical results provided above to the maximal
quantum violation of the Bell inequality (\ref{eq:exampleineq}) obtained numerically.
For this purpose, let us denote as $\ket{\psi_n^{\min}}$ the state maximally violating (\ref{eq:exampleineq}). As we have discussed, it is a superposition of Dicke states, and thus can be written as
\begin{equation}
 \label{eq:gaussiansuperposnumeric}
 \ket{\psi_n^{\min}}=\sum_{k=0}^nc_k^{(n)}\ket{D_n^k},
\end{equation}
where $c_k^{(n)}$ are real coefficients, because the Bell operator
is taken to be a real symmetric matrix. Figure \ref{fig:osc} shows
that $c_k^{(n)}$ display an interesting regularity when plotted
versus $k$; this corresponds to the case ${\cal B}_{n/2}(0,\theta
- \varphi)$, where $\theta- \varphi$ has been optimized
numerically. However, if $\varphi$ is such that $C=0$, then the
form of the $c_k^{(n)}$, now denoted $\alpha_k^{(n)}$, is
displayed in Figure \ref{fig:gauss} and compared with the form
given in Theorem \ref{thm:qv}, where we can see a good agreement.
The relative difference between the numerical and the analytical
state giving maximal quantum violation is compared in Figure
\ref{fig:inset}. Observe that Theorem \ref{thm:unitaryrotation}
ensures that both $c_k^{(n)}$ and $\alpha_k^{(n)}$ lead to the
same quantum violation; one just has to rotate the measurements
accordingly.

\begin{figure}[]
\centering\includegraphics[width=0.4\columnwidth]{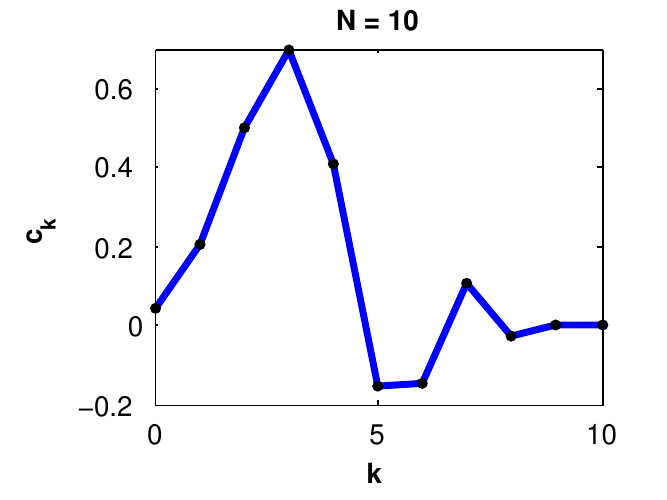}
\includegraphics[width=0.4\columnwidth]{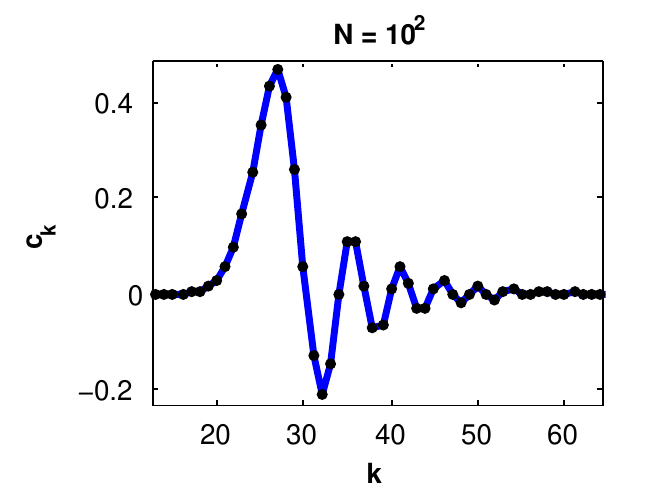}\\
\includegraphics[width=0.4\columnwidth]{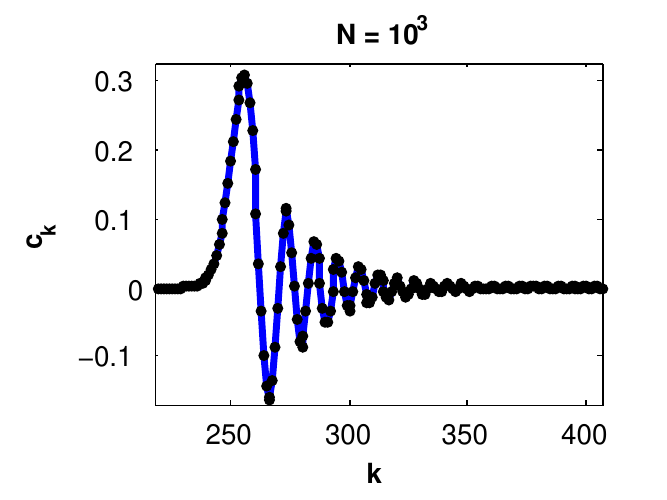}
\includegraphics[width=0.4\columnwidth]{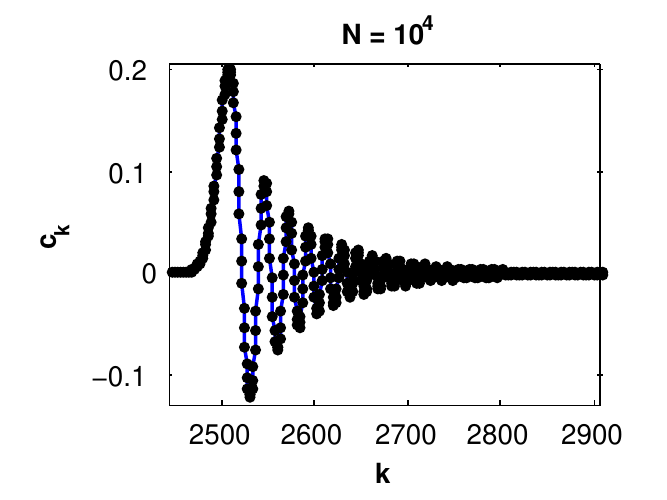}
\caption{(Color online) Distributions of $c_k^{(n)}$ versus $k$ for different
values of $n=10,10^2,10^3,10^4$. $c_k^{(n)}$ are the coefficients of the eigenstate corresponding to the minimal eigenvalue of ${\cal B}_{n/2}(0,\theta-\varphi)$, for the optimal $\theta- \varphi$. The black dots correspond to the values of
$c_k^{(n)}$ and the blue line connecting them is added to better show
how they behave with $k$.}\label{fig:osc}
\end{figure}

\begin{figure}[!h]
\centering
\includegraphics[width=0.4\columnwidth]{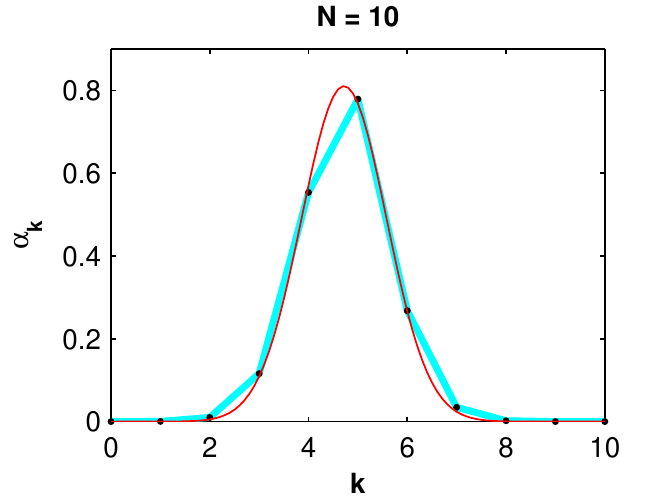}
\includegraphics[width=0.4\columnwidth]{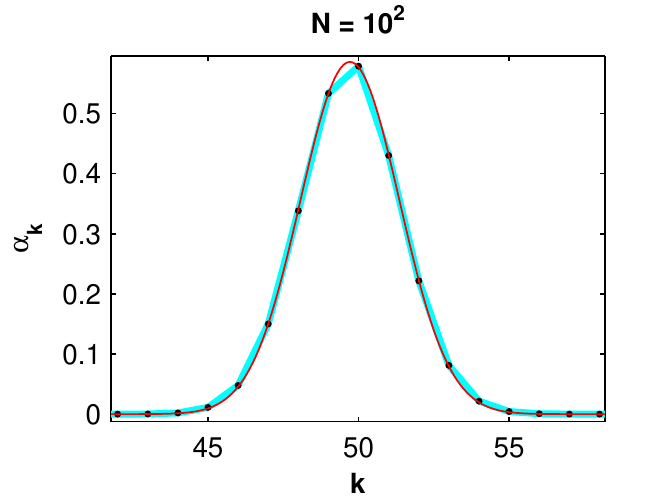}\\
\includegraphics[width=0.4\columnwidth]{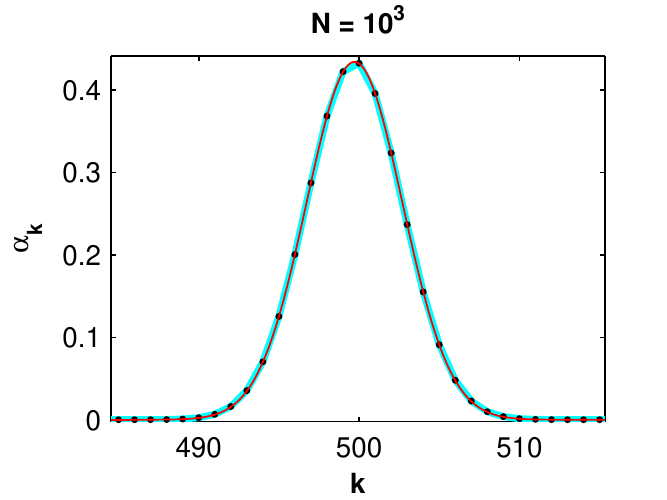}
\includegraphics[width=0.4\columnwidth]{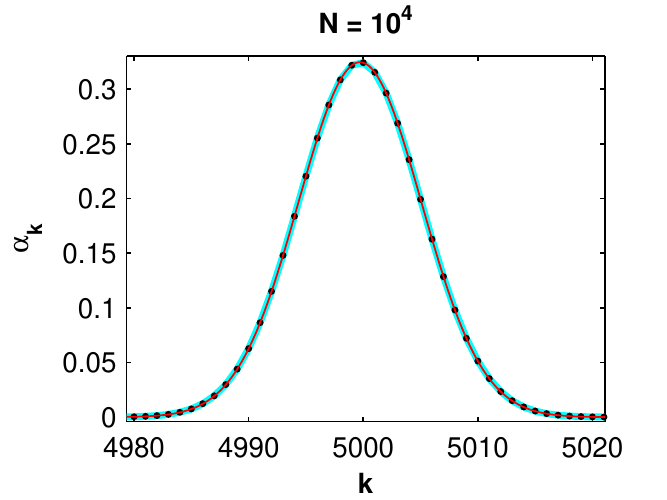}
\caption{(Color online) The black points represent values of $\alpha_k^{(n)}$
for $n=10,10^2,10^3,10^4$. The cyan line is added
to mark better their behavior with $k$. The red lines present square roots of
the Gaussian distributions of Theorem \ref{thm:qv} with
means $\mu_n$ and variances $\sigma_n$ chosen so that
the distributions fit the points. Their explicit values are
$\mu_n=n/2+(1/4\cos\theta_n)$, where $\theta_{10}=2.6672$,
$\theta_{10^2}=2.6334$, $\theta_{10^3}=2.6231$,
$\theta_{10^4}=2.6180$, and $\sigma_{10}=0.4049$, $\sigma_{10^2}=1.3935$, $\sigma_{10^3}=4.5109$, and
$\sigma_{10^4}=14.379$, respectively.
Noticeably, already for $n=10$ the red line matches the points very
well.}\label{fig:gauss}
\end{figure}

\begin{figure}[]
\centering
\includegraphics[width=1\columnwidth]{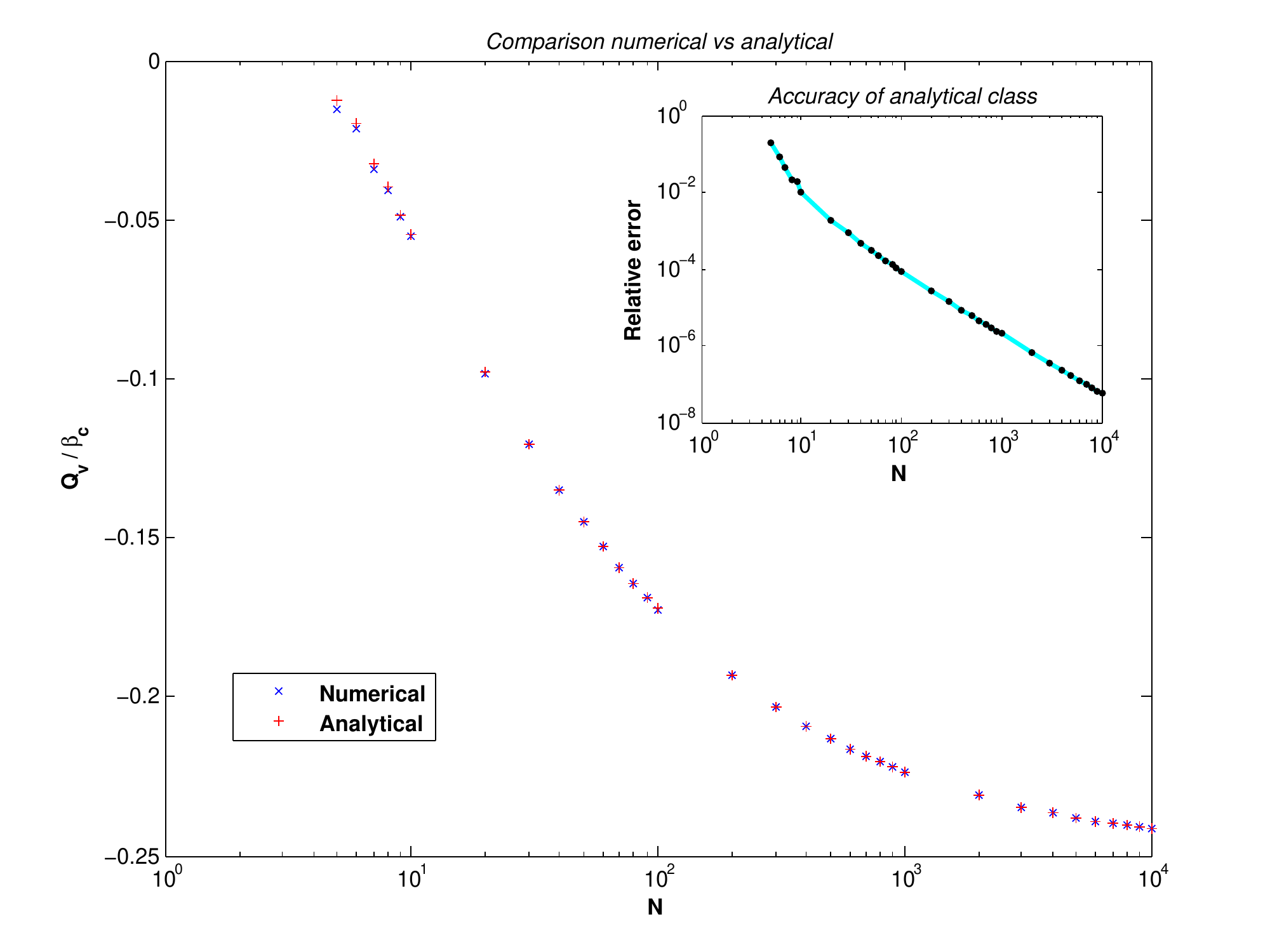}
\caption{(Color online) Maximal quantum violations of inequality (\ref{eq:exampleineq}) and obtained from diagonalization of the symmetric block ${\cal B}_{n/2}$ (blue crosses)
are compared to violations by the states defined in Theorem \ref{thm:qv} (red crosses). All means
and variances in the Gaussian distributions were chosen so that
$\ket{\psi_n^{\min}}$ maximally violate (\ref{eq:exampleineq}).
The inset contains the relative difference between numerical and analytical violations.
The points in the inset are the actual values of $n$ for which we performed the calculations and the cyan line is added just to help visualize the tendency. Clearly, already for $n=10$, the violations are almost the same. }\label{fig:inset}
\end{figure}

\subsection{The $2$-body reduced state}
Using the fact that, when we trace out a number of subsystems of a quantum state supported on the symmetric space, the reduced density matrix of the state is the same regardless of which systems are traced out, we are able to give an equivalent way of computing the maximal quantum violation of a permutationally invariant Bell inequality (\ref{eq:2bodySymgeneral}) provided all parties perform the same pair of measurements.
Indeed, if we denote by $\rho_2$ the two-body reduced state of $\ket{\psi}$, where $\ket{\psi}$ maximally violates inequality (\ref{eq:2bodySymgeneral}) with measurements ${\cal M}_0$ and ${\cal M}_1$, then we only need to multiply the expectation values $\langle{\cal M}_k\otimes {\cal M}_l \rangle_{\rho_2}$ or $\langle{\cal M}_k\otimes {\mathbbm 1}_2 \rangle_{\rho_2}$ by the times they appear in ${\cal S}_{kl}$ or in ${\cal S}_k$. By doing so, the quantum expectation value of (\ref{eq:2bodySymgeneral}) simplifies to
\begin{equation}
\langle\psi|\mathcal{B}|\psi\rangle=
\beta_c + n\left(\alpha \langle{\cal M}_0\otimes {\mathbbm 1}_2\rangle_{\rho_2}+\beta\langle{\cal M}_1\otimes {\mathbbm 1}_2\rangle_{\rho_2}\right)+n(n-1)\left(\frac{\gamma}{2}\langle{\cal M}_0\otimes {\cal M}_0 \rangle_{\rho_2}+\delta\langle{\cal M}_0\otimes {\cal M}_1 \rangle_{\rho_2}+\frac{\varepsilon}{2}\langle{\cal M}_1\otimes {\cal M}_1 \rangle_{\rho_2}\right)
 \label{eq:quantumexpectationfromreducedstate}
\end{equation}
And if the value of (\ref{eq:quantumexpectationfromreducedstate}) is negative, nonlocality is revealed.

Let us first show how to compute this reduced state:
\begin{thm}
Let $\rho$ be the density matrix of an $n$-qubit state supported on the symmetric space. Let $\rho_{\cal S}$ be the representation of $\rho$ in the symmetric space. Let us denote by $\mathrm{Tr}_{n-d}(\rho)$ the density matrix of $\rho$ after tracing out $n-d$ subsystems. This reduced density matrix is given by
\begin{equation}
\left(\mathrm{Tr}_{n-d}(\rho)\right)^{\mathbf i'}_{\mathbf j'}= \sum_{k=0}^{n-d}\frac{{n-d \choose k}(\rho_{\cal S})^{k+|\mathbf{i}'|}_{k+|\mathbf{j}'|}}{\sqrt{{n \choose k + |\mathbf{i}'|}{n \choose k+ |\mathbf{j}'|}}},
\label{eq:reduced-state}
\end{equation}
where $0 \leq {\mathbf{i'}}, {\mathbf{j'}}<2^d$, $\mathbf{i}'=i_0\ldots i_{d-1}$ and $\mathbf{j}'=j_0\ldots j_{d-1}$ are represented in binary and $|\mathbf{i}'|, |\mathbf{j}'|$ is their number of ones in this binary representation. Columns are indexed by $\mathbf{i'}$ and rows are indexed by $\mathbf{j'}$.
\label{thm:partialtrace}
\end{thm}

Let us observe that the factor appearing in Eq. (\ref{eq:reduced-state}), i.e.,
\begin{equation}
f(n,k,d,|\mathbf{i}'|,|\mathbf{j}'|)={n-d \choose k}\left[{n\choose k + |\mathbf{i'}|}{n\choose k + |\mathbf{j'}|}\right]^{-1/2}
\end{equation}
%
%Observe that in Eq. (\ref{eq:reduced-state}) the factor ${n-d \choose k}\left[{n\choose k + |\mathbf{i'}|}{n\choose k %+ |\mathbf{j'}|}\right]^{-1/2}$ appears. Let us denote it by $f(n,k,d,|\mathbf{i'}|,|\mathbf{j'}|)$. It is %\replaced{straightforward}{immediate} to see that $f$
%
is symmetric in $|\mathbf{i'}|$ and $|\mathbf{j'}|$ and it fulfills the geometric mean property
\begin{equation}
 f(n,k,d,|\mathbf{i'}|,|\mathbf{j'}|) = \sqrt{f(n,k,d,|\mathbf{i'}|,|\mathbf{i'}|)f(n,k,d,|\mathbf{j'}|,|\mathbf{j'}|)}.
 \label{eq:geometricmean}
\end{equation}
When exploring large values of $n$, from the point of view of numerics, it is important to have in mind that the combinatorial numbers appearing in $f$ grow very rapidly, affecting the formula given in Eq. (\ref{eq:reduced-state}) if they are not handled carefully. In particular, for $n\gtrsim 650$ they can take values greater than $10^{308}$, which is the storage limit for $64$-bits floating point arithmetic. It is therefore important to work with a simplified version of $f$ and here we present the cases for $d=1,2$, which are the ones we are going to use in the present context.
\begin{equation}
 f(n,k,1,0,0)=\frac{n-k}{n}, \qquad f(n,k,1,1,1)=\frac{k+1}{n}.
 \label{eq:f_for_d=1}
\end{equation}
\begin{equation}
 f(n,k,2,0,0)=\frac{n-k}{n}\frac{n-k-1}{n-1}, \quad f(n,k,2,1,1)=\frac{n-k-1}{n}\frac{k+1}{n-1}, \quad f(n,k,2,1,1)=\frac{k+1}{n}\frac{k+2}{n-1}.
 \label{eq:f_for_d=2}
\end{equation}
Using Eq. (\ref{eq:geometricmean}) we can find the rest of the values.

We are now ready to show what is the form of $\rho_2$ for the state given in Theorem \ref{thm:qv}.
\begin{thm}
 Let $\ket{\psi}$ be the state introduced in Theorem \ref{thm:qv}. Its two-body reduced state $\rho_2$ can be well approximated for large $n$ as
 \begin{equation}
  \rho_2 = \frac{1}{n(n-1)}\left(n^2\proj{+}^{\ot 2}
  +\frac{n}{2}
  \left(
  \begin{array}{cccc}
   -(1+2c)&-c&-c&1\\
   -c&0&0&c\\
   -c&0&0&c\\
   1&c&c&(2c-1)
  \end{array}
  \right)
  +o(n)
  \right),
  \label{eq:2bodyanalyticalreducedstate}
 \end{equation}
 where $c:=\mu-n/2$ and $\ket{+}:=(\ket{0}+\ket{1})/\sqrt{2}$.
 \label{thm:2bodyreduced}
\end{thm}

It is worth noting that, by monogamy of entanglement, $\rho_2$ tends to a separable state, as one can see directly from (\ref{eq:2bodyanalyticalreducedstate}). When this state is used in the class of Bell inequalities presented in Theorem \ref{thm:class}, we observe that the quantum violation relative to the classical bound increases \cite{Science}, \textit{i.e.}, its quantum violation increases linearly with $n$, whereas Theorem \ref{thm:2bodyreduced} ensures that the two-body reduced states are almost separable. The reason for that is that the coefficients of the inequality in Theorem \ref{thm:class} are picked such that the state $(\ket{0}+\ket{1})/\sqrt{2}$ has expectation value exactly zero on them, so the quantum violation is completely defined by the second term in Eq. (\ref{eq:2bodyanalyticalreducedstate}).

\subsection{Experimental realization of permutationally invariant Bell inequalities with $2$-body correlators}
\label{sec:Experimental}

Let us finally notice that when the two-body Bell inequality is permutationally invariant and the observers perform the same pair of measurements at each site (\ref{eq:same_measurements}), the one-body and the two-body permutationally invariant expectation values $\mathcal{S}_k$ (\ref{eq:def1bodysym}) and $\mathcal{S}_{kl}$ (\ref{eq:def2bodysym}) can be determined by measuring collective spin operators. Such operators, corresponding to spin-$1/2$ particles, are given by
\begin{equation}
 S_{\alpha}:=\frac{1}{2}\sum_{i=0}^{n-1}\sigma_{\alpha}^{(i)},
\end{equation}
where $\alpha \in \{x,y,z\}$. It is convenient to consider the vector $\vec{S} := [S_x, S_y, S_z]$. The value of the permutationally invariant correlators (hence, the value of any Bell inequality built from them), can be readily measured by taking measurements in any direction $\mathbf{\hat n}$ in the Bloch sphere, where $\mathbf{\hat n}$ is a unit vector, by means of $\mathbf{\hat n}\cdot \vec{S}$. Observe that, if ${\cal M}_k^{(i)}$ is aligned along the direction $\mathbf{\hat{m}}_k$ for all $i$, meaning that ${\cal M}_k^{(i)}=\mathbf{\hat{m}}_k\cdot\vec{\sigma}$, then ${\cal S}_k$ is given by
\begin{equation}
 {\cal S}_k = 2 \langle \mathbf{\hat{m}}_k \cdot \vec{S} \rangle.
\end{equation}
Similarly, the value of $\mathcal{S}_{kl}$ can be expressed in terms of the different directions $\mathbf{\hat m}_0, \mathbf{\hat m}_1$ of the measurements as
\begin{equation}
 {\cal S}_{kl} = \left\langle \left[(\mathbf{{\hat m}}_k + \mathbf{{\hat m}}_l)\cdot \vec{S}\right]^2 \right\rangle - \left\langle \left[(\mathbf{{\hat m}}_k - \mathbf{{\hat m}}_l)\cdot \vec{S}\right]^2 \right\rangle - n (\mathbf{{\hat m}}_k \cdot \mathbf{{\hat m}}_l).
\end{equation}
In particular, when $k=l$, we have
\begin{equation}
 {\cal S}_{kk} = 4 \left\langle \left(\mathbf{{\hat m}}_k\cdot  \vec{S}\right)^2\right\rangle - n.
\end{equation}
The way to estimate such collective measurements was already proposed in \cite{spin-sq} and it was implemented \cite{Hammerer} using spin polarization spectroscopy \cite{Isart}. In section VII we discuss its implementation with other technologies. Moreover, motivated by the recent experiment aiming at detection of entanglement in the Dicke states \cite{Klempt}, in appendix \ref{AppC} we study how certain experimental errors affect the quantum violation of our Bell inequalities.

\section{Inequalities detecting Dicke states}
\label{sec:Dicke}

Symmetric Dicke states $\ket{D_n^k}$ \cite{Dicke} are of great importance in many-body physics. They are the ground states of the isotropic Lipkin-Meshkov-Glick Hamiltonian \cite{LMG}
\begin{equation}
H=-\frac{\lambda}{n}\sum_{i<j}\left(\sigma_x^{(i)}\sigma_x^{(j)}+\sigma_y^{(i)}\sigma_y^{(j)}\right) - h \sum_{i=1}^n \sigma_z^{(i)},
 \label{eq:LMG}
\end{equation}
which describes $n$ spins interacting through the two-body ferromagnetic coupling ($\lambda > 0$) and $h$ is the intensity of the magnetic field along the $z$ direction. By tuning the value of $h$, the ground state of (\ref{eq:LMG})
varies along all $\ket{D_n^k}$ for every $k$.
Observe that the interactions in (\ref{eq:LMG}) are over all possible pairs, so it is worth mentioning that Dicke states of $n$ qubits can also be well approximated for low $n$ as ground states of a one-dimensional ferromagnetic
spin-$1/2$ chain with nearest-neighbor XXZ interactions \cite{DickeChain}. The states $\ket{D_n^k}$ are genuinely multipartite entangled for any $0<k<n$. Recall that any $n$-qubit symmetric state (pure or mixed) is entangled if,
 and only if, it is genuinely multipartite entangled \cite{Eckert2001,OurPRA}). These states are also generically genuinely multipartite nonlocal \cite{GMN}.
Dicke states can be  generated in experiments with current technologies (see, e.g., Refs. \cite{ExperimentZeilinger, ExperimentWeinfurter,Klempt}).
%
%For instance, recently, the half-filled Dicke state for $6$ qubits $\ket{D_6^3}$ has been generated with photons

In Ref. \cite{Science} we showed that the two-body Bell inequalities are capable of detecting nonlocality in some of the Dicke states, namely, $\ket{D_N^k}$ with $k=\lfloor N/2\rfloor$ and $k=\lceil N/2\rceil$. Here, we show that nonlocality of any of the entangled Dicke states can be revealed with such Bell inequalities. In this section we present two classes of two-body Bell inequalities that are violated by any Dicke state $\ket{D_n^k}$ for $0 < k < n$. Note that $\ket{D_n^0}=\ket{0}^{\otimes n}$ and $\ket{D_n^n}=\ket{1}^{\otimes n}$, which cannot lead to nonlocal correlations, as they are clearly separable states. The first class, introduced in Theorem \ref{thm:classDickeLowk} detects all Dicke states $\ket{D^k_n}$ for which $0 < k \lesssim 3n/10$ (or $7n/10 \lesssim k < n$, by an appropriate symmetry of the Bell inequalities, consisting on renaming the outcomes of all the observables); the second one is introduced in Theorem \ref{thm:classDickeMidk} and it detects the rest, where $3n/10 \lesssim k \lesssim 7n/10$. As it can be seen in Figure \ref{fig:dickeviol}, the limits $\lesssim$ have been taken very conservatively, and there is a clear overlap between the two classes, both detecting Dicke states with a number of excitations $k$ around $3n/10$ (or $7n/10$). For the states $\ket{D_n^k}$, the expectation value that a Dicke state takes with a symmetric two-body Bell inequality is straightforward to calculate via Theorem \ref{thm:Blocks} and it is simply $\bra{D_n^k}{\cal B}\ket{D_n^k}= d_k$.

\begin{thm}
 The inequality (\ref{eq:2bodySymgeneral}) defined with the following parameters
  \begin{equation}
  \beta_c = (1+2k)[(n-2k-1)^2+n-1],\qquad \alpha = \beta = (1+2k)(n-1-2k), \qquad \gamma = \varepsilon = k, \quad \delta = k+1,
  \label{class:lowk}
 \end{equation}
 where $0\leq k \leq (n-1)/2$, is a valid Bell inequality for $\mathbbm{P}_2^{{\mathfrak S}_n}$ for any $n$.
 \label{thm:classDickeLowk}
\end{thm}

\begin{thm}
Let $\nu := \lfloor n/2 \rfloor-k$. The inequality (\ref{eq:2bodySymgeneral}) defined with the following parameters
 \begin{itemize}
 \item If $n\equiv 0 \mathrm{\ mod\ } 2$:
  \begin{equation}
   \beta_c = {n \choose 2}[n+2(2\nu^2+1)], \qquad \alpha = 2\nu n(n-1),\ \beta = \alpha/n, \qquad \gamma = n(n-1),\ \delta = n,\ \varepsilon = -2.
   \label{class:midkeven}
  \end{equation}
  \item If $n\equiv 1 \mathrm{\ mod\ } 2$:
  \begin{equation}
   \beta_c = {n \choose 2}[n+3+4\nu (\nu+1)], \quad \alpha = (1+2\nu)n(n-1),\ \beta = \alpha/n, \quad \gamma = n(n-1),\ \delta = n,\ \varepsilon = -2.
   \label{class:midkodd}
  \end{equation}
 \end{itemize}
 where $0\leq k \leq \lfloor n/2 \rfloor-1$, is a valid Bell inequality for $\mathbbm{P}_2^{{\mathfrak S}_n}$ for any $n$.
 \label{thm:classDickeMidk}
\end{thm}

\subsection{Quantum violation of the Dicke states}
Here we analyze which are the optimal measurements to be performed on a Dicke state $\ket{D^k_n}$. By Theorem \ref{thm:Blocks}, the computation of the expectation value $\bra{D^k_n}{\cal B}(\varphi, \theta)\ket{D^k_n}$ is immediate because it is precisely $d_k$, as defined in Theorem \ref{thm:Blocks}. Note, however, that the state is fixed, so we cannot suppose that $\bra{D^k_n}{\cal B}(\varphi, \theta)\ket{D^k_n}$ depends only on $\theta - \varphi$, since the Theorem \ref{thm:unitaryrotation} does not apply. Using the results stated in Theorem \ref{thm:bigcalculation}, one can easily generalize Theorem \ref{thm:Blocks} to the use of any pair of measurements ${\cal M}_0 = \hat{\mathbf n}_0\cdot \vec{\sigma}$, ${\cal M}_1 = \hat{\mathbf n}_1\cdot \vec{\sigma}$. However, we have seen that $\bra{D^k_n}{\cal B}(\hat{\mathbf n}_0, \hat{\mathbf n}_1)\ket{D^k_n}$ depends only on the difference between the azimuthal angles that define $\hat{\mathbf n}_0, \hat{\mathbf n}_1$ in spherical coordinates, and that this expectation value is minimal when such difference is zero. That implies that there is no loss of generality in considering real observables, and the present form of Theorem \ref{thm:Blocks} applies.
For the case of inequality (\ref{class:lowk}), \textit{i.e.} with Dicke states with low number of excitations, we have that the expectation value $\bra{D^k_n}{\cal B}(\varphi, \theta)\ket{D^k_n}$ is a symmetric function of $\varphi$ and $\theta$. Numerically we observe that its minimum is obtained at $\varphi = - \theta$, a fact which significantly simplifies our analysis, and it allows us to obtain an analytically closed form for its maximal quantum violation for a fixed state $\ket{D^k_n}$. After some algebra, by solving
\begin{equation}
\frac{\mathrm{d}}{\mathrm{d} \theta}d_k(-\theta,\theta)=0
 \label{eq:lowk}
\end{equation}
we obtain the following non-trivial solutions:
\begin{equation}
 \theta^* = \pm 2 \arccos \left(\pm \sqrt{\frac{C_1}{C_2}}\right),
\end{equation}
where $C_1 = k[n-(1+3k)]$ and $C_2 = n^2 (2k+1) - n (8k^2+4k+1) + 2k^2(1+4k)$ and the two $\pm$ signs are independent. Using this optimal angle $\theta^*$ we obtain the maximal quantum violation, given by
\begin{equation}
 d_k(-\theta^*, \theta^*)= -\frac{4C_1^2}{C_2}\in O((n^2)^2/n^3) = O(n).
\end{equation}

For the case of inequality (\ref{class:midkeven}) or (\ref{class:midkodd}), a proof that quantum violation exists for any $n$ was already given in \cite{Science} for $k=\lceil n/2 \rceil$. However, such violation is not the maximal one, as it was assumed that ${\cal M}_0 = \sigma_z$ in order to simplify the calculations.\\

We have calculated numerically the maximal quantum violation of (\ref{class:lowk})-(\ref{class:midkodd}), and the results are reported in Figure \ref{fig:dickeviol} for $n= 2^{10}$.
\begin{figure}[!h]
\centering
\includegraphics[width=0.4\columnwidth]{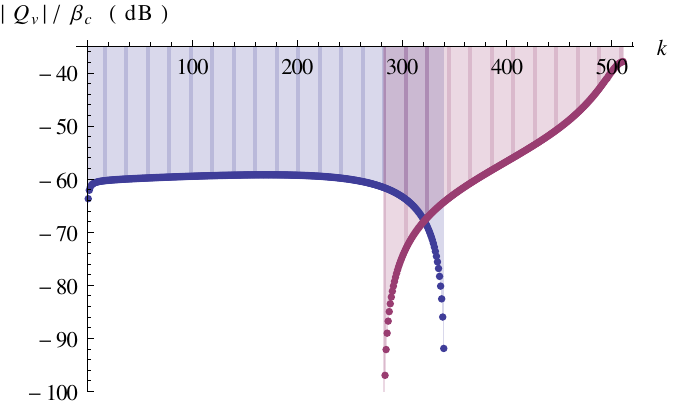}
\includegraphics[width=0.4\columnwidth]{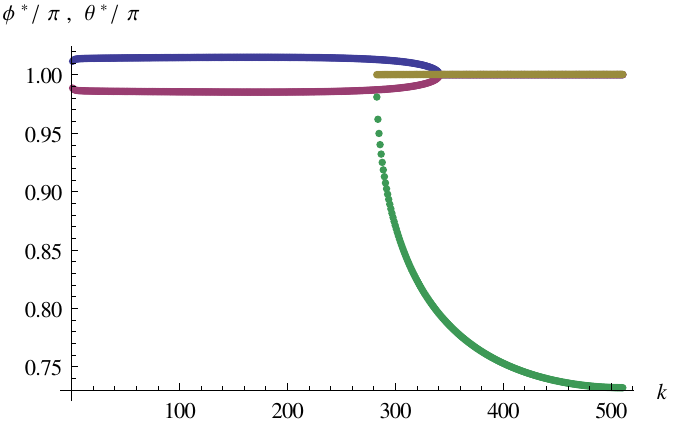}
\caption{(Color online) On the left, the values of the maximal quantum violation, denoted $|Q_v|$, of (\ref{class:lowk} - \ref{class:midkeven}) for the Dicke state $\ket{D_n^k}$ with $n= 2^{10}=1024$ qubits. The violation is normalized to the corresponding classical bound $\beta_c$.  The blue curve corresponds to (\ref{class:lowk}) and the red curve corresponds to (\ref{class:midkeven}). The maximal quantum violation has been plotted in logarithmic scale (in dB) in order to better compare the two bounds; \textit{i.e.} the actual plotted value is $10 \log_{10}(-\min_{\varphi, \theta} d_k/\beta_c)$, where $d_k$ is the quantum expecation value $\bra{D^k_n}{\cal B}(\varphi, \theta)\ket{D^k_n}$. The horizontal axis is cut at $k=n/2$ since the plot is symmetric with respect to $k=n/2$. The reason for that is that the same violation can be achieved by an inequality that results from renaming the outcomes of the measurements or, equivalently, by switching $\ket{0} \leftrightarrow \ket{1}$ for every qubit. It is clear from the plot that the two classes overlap, so these two classes cover all possible entangled Dicke states $\ket{D^k_n}$.\\
On the right, the optimal measurement angles $\varphi^*/\pi, \theta^*/\pi$ for every $k$ that lead to the violation shown on the left. For the class of inequalities (\ref{class:lowk}) $\phi^*$ and $\theta^*$ are plotted in blue and red, respectively. For the class of inequalities (\ref{class:midkeven}), the optimal measurement angles $\phi^*$ and $\theta^*$ are plotted in yellow and green, respectively. Interestingly, although the plot suggests that $\varphi^* = \pi$ for the inequality (\ref{class:midkeven}), it is not the case. If it were exactly $\pi$, one could not achieve the optimal quantum violation. For this particular $n$, $\phi^*|_{k=280}\approx\pi$ and $\phi^*|_{k=511}\approx 1.00023\pi$.
Intuitively, the reason for this slight discrepancy is that $d_k$ can be expressed as $d_k = 2n^4\cos(\varphi/2)^4 -2n^3\cos(\varphi/2)^2(k-\cos \theta + (1+3k)\cos \varphi) + 2n^2k^2\cos^2(\varphi/2)(1+3\cos \varphi) + o(n^2k^2)$, and this expression can be made $o(n^2k^2)$ just by picking $\varphi = \pi$. However, one can do better by taking advantage of the fact that a small enough value of $\cos(\varphi/2)$ makes the whole expression effectively $o(n^2k^2)$ (this value of $\varphi^*$ will actually depend on $n$ and $k$), and the leading terms actually help in the minimization in this situation.}
\label{fig:dickeviol}
\end{figure}

\section{Conclusions}
\label{sec:conclusions}

In this paper we have constructed  and classified Bell
inequalities involving 1-body and 2-body correlators
of many-body systems. We studied maximal quantum violations of these
inequalities, as well as  robustness of these violations and
dependence with respect to various parameters involved. We paid
particular attention to Bell inequalities that are violated by
Dicke states that have physical relevance: on one hand they can be
prepared in experiments, on the other hand they appear as ground
states of physically relevant models, such as the Lipkin-Meshkov-Glick
model.

Let us conclude by listing several experimental setups in which
nonlocality in many body systems may be tested using our
inequalities:

\begin{itemize}

\item {\bf Ultracold trapped atoms.} Dicke states have been recently created in spinor Bose-Einstein condensates
(BEC) of Rubidium $F=1$ atoms, via the parametric process of spin
changing collisions, in which two $m_F=0$ collide to produce a
$m_F=\pm 1$ pair  \cite{Klempt}.   These recent experiments
demonstrate the production of many thousands of neutral atoms
entangled in their spin degrees of freedom. Dicke-like states can
be produced in this way, with at least 28-particle genuine
multiparty entanglement and  a generalized squeezing parameter of
$−11.4(5)$ dB. Similarly, Rubidium atoms of pseudo-spin 1/2 in BEC may be
employed to generate scalable squeezed states (for the early theory proposal see \cite{spin-sq}, for experiments see \cite{Muessel}). Very recently non-squeezed (non-Gaussian) entangled states of many atoms \cite{Strobel} were generated in this way. The number of atoms used in these experiments are of order of thousands
and larger. So far, these numbers and experimental errors and
imperfections are too large, while the corresponding fidelities
too small to detect many body nonlocality. In principle, however,
it is possible to perform these experiments with mesoscopic atom
numbers (say $\lesssim 100$, controlling the atom number to a single atom level (see
Ref. \cite{Hume} for the resonant fluorescence detection of
Rb$^{87}$ atoms in a MOT, and Refs. \cite{Wenz,Zuern} for
optically trapped spin 1/2 fermions).

\item{\bf Ultracold trapped ions.} Ultracold trapped ions with internal
pseudo-spin "talk" to each other via phonon excitations, and under
some conditions behave as spin chains with long range interaction.
This was originally proposed in Ref. \cite{Wunderlich}, using
inhomogeneous magnetic fields, and in Ref. \cite{Porras},
employing appropriately designed laser-ion interactions. The
pioneering experiments were reported in Refs.
\cite{Schaetz,Monroe2010}. While in the early theory studies
\cite{Porras,Porras2,HaukeNJP,Maik}  spin interactions decaying
with the third power of the distance were considered, it was experimentally
demonstrated  that  management of phonon dispersion allows to
achieve powers between 0.1 and 3  in the 2D arrays of
traps \cite{Bollinger}. Recent state of art represents the  work
on experimental realization of a quantum integer-spin chain with
controllable interactions \cite{Monroe}. We have studied trapped
ion systems in relation to long range $SU(3)$ spin chains and
quantum chaos \cite{Grass}, and trapped-ion quantum simulation of
tunable-range Heisenberg chains \cite{Grass1}. In the latter paper
we demonstrated that significant violation of the Bell
inequalities, discussed in the present paper, is possible even for
the ground states of the models with large, but finite interaction
range.

\item {\bf Ultracold atoms in nanostructures.} Yet another possibility concerns systems of
ultracold atoms traps in a vicinity of tapered fibers and optical
crystals (band gap materials). The experimental progress in
coupling of ultracold atomic gases to nanophotonic waveguides,
initiated by Refs. \cite{Nayak,Vetsch,Goban}, is very rapid (cf.
\cite{Atomlight}). Early theoretical studies concerned remarkable
optical properties of these systems (cf.
\cite{Kien,Zoubi,Gong,Gorshkov}).  Ideas  and proposals concerning
realization of long range spin models were developed more
recently, and mainly in Refs. see \cite{Chang, Chang1, Chang2}.

\item{\bf Cold and ultracold atomic ensembles.} Last, but not least, one should consider cold and ultracold ensembles (for an excellent  review
 see \cite{Hammerer}), in which, by employing quantum Faraday effect, one can reach unprecendented
degeees of squeezing of the total atomic spin (cf. \cite{Napolitano,Sewell}, and unprecendented degrees of  precision of quantum magnetometry (cf. \cite{magneto}).
Note that in many concrete realisations the many body Bell inequalities derived in this paper require precise measurments of the total spin components, and their quantum fluctuations.
Quantum Faraday effect, or in other words spin polarization spectroscopy, seems to be  a perfect method to achieve this goal; note that in principle
 it allows also to reach spatial resolution, and/or  to measure spatial Fourier components of the total spin \cite{Isart}.

\end{itemize}

\acknowledgements We thank Joe Eberly,
Markus Oberthaler, Julia Stasi\'nska, and Tam\'as V\'ertesi for
enlightening discussions. This work was supported by Spanish MINECO
projects FOQUS and a AP2009-1174 FPU PhD grant, the EU IP SIQS, ERC
AdG OSYRIS and CoG QITBOX, the Generalitat de Catalunya
and the John Templeton Foundation. C. K. and B. L. 
acknowledge support from the Centre for Quantum Engineering and
Space-Time Research (QUEST) and from the Deutsche Forschungsgemeinschaft
(Research Training Group 1729).

\appendix

\section{Proofs}\label{AppA}
\subsection{Proof of Theorem \ref{thm:vertexs}}
\begin{proof}
 Let us start from the \textit{if} part and consider  $p=(a,b,c,d) \in {\mathbbm T}_n \setminus \partial {\mathbbm T}_n$. This implies that all the coordinates of $p$ are strictly positive, \textit{i.e.}, $a,b,c,d \geq 1$. We will show that $p$ is not a vertex of $\mathbbm{P}_2^{{\mathfrak S}_n}$ by explicitly finding a convex decomposition of $p$ into vertices of $\mathbbm{P}_2^{{\mathfrak S}_n}$. Let us consider a vector $v=(1,-1,-1,1)$. Note that $v$ is proportional to the fourth row of the Hadamard matrix in (\ref{eq:Hadamard}). Since the Hadamard matrix is orthogonal, all its rows are pairwise orthogonal. Consequently, for any $\lambda$, $p+\lambda v$ has the same value for $n$, ${\cal S}_0$ and ${\cal S}_1$, whereas ${\cal Z}(p+\lambda v)={\cal Z}(p)+4\lambda$. Hence, the only coordinate that can change is ${\cal S}_{01}$, where one obtains ${\cal S}_{01}(p+\lambda v)={\cal S}_{01}(p)-4\lambda$. Observe that for any $\mu_1, \mu_2$, we have
 \begin{equation}
  \mu_1 \varphi(p+\mu_2 v)+ \mu_2 \varphi(p-\mu_1 v) = (\mu_1 + \mu_2)\varphi(p).
 \end{equation}
If $\mu_1, \mu_2 > 0$, then a convex decomposition of $\varphi(p)$ is given:
\begin{equation}
\label{eq:cdecomp}
 \varphi(p)=\frac{\mu_1}{\mu_1+\mu_2}\varphi(p+\mu_2 v) + \frac{\mu_2}{\mu_1+\mu_2}\varphi(p-\mu_1 v)
\end{equation}
Notice that the choice $\mu_1 = \min \{a,d\}$ and $\mu_2 = \min \{b,c\}$ ensures that $p+\mu_2 v \in \partial {\mathbbm T}_n$ and that $p-\mu_1 v \in \partial {\mathbbm T}_n$. Since $\mu_1, \mu_2 > 0$, this ensures that (\ref{eq:cdecomp}) is a proper convex decomposition, so $\varphi(p) \notin \mathrm{Ext}(\mathbbm{P}_2^{{\mathfrak S}_n})$.

Conversely, for the \textit{only if} part, let us consider $p\in \partial {\mathbbm T}_n$. If $\varphi(p)$ were not a vertex of $\mathbbm{P}_2^{{\mathfrak S}_n}$, it could be decomposed as a convex combination of elements from $\mathrm{Ext}(\mathbbm{P}_2^{{\mathfrak S}_n})$, namely $\{q_i\}_{i=1}^k$ with some $k>1$. Since $\varphi$ is an invertible transformation in deterministic local strategies, for each $q_i$ there is a unique $p_i \in {\mathbbm T}_n$ such that $\varphi(p_i)=q_i$.
To see that $\varphi$ is indeed invertible in any local deterministic strategy, it suffices to observe that $({\cal S}_0, {\cal S}_1, {\cal S}_{00}, {\cal S}_{01}, {\cal S}_{11})=({\cal S}_0, {\cal S}_1, ({\cal S}_{0})^2-n, {\cal S}_{0}{\cal S}_{1} - {\cal Z}, ({\cal S}_{1})^2-n)$, from which one can trivially obtain $n,{\cal S}_1, {\cal S}_0, {\cal Z}$ and, from these, $a,b,c,d$ by virtue of Equation (\ref{eq:Hadamard}).
Denote then $p_i=(a_i,b_i,c_i,d_i)$; $\varphi(p)$ has then the following decomposition:
\begin{equation}
\label{eq:cdecomp2}
\varphi(p)=\sum_{i=1}^k \lambda_i \varphi(p_i),
\end{equation}
where $0<\lambda_i<1$ and the sum of the $\lambda_i$ is $1$.

By looking at the third coordinate of Equation (\ref{eq:cdecomp2}), the fact that ${\cal S}_{00}(r)=({\cal S}_0(r))^2-n$ for all $r\in {\mathbbm T}_n$ implies the following:
\begin{eqnarray}
\label{eq:tmp1}
({\cal S}_0(p))^2-n=\sum_{i=1}^k\lambda_i [({\cal S}_0(p_i))^2-n]\nonumber\\
 \left(\sum_{i=1}^k \lambda_i {\cal S}_0(p_i)\right)^2=\sum_{i=0}^k \lambda_i ({\cal S}_0(p_i))^2.
\end{eqnarray}
Treating Equation (\ref{eq:tmp1}) as a quadratic function in ${\cal S}_0(p_m)$, for some $1\leq m\leq k$ we obtain
\begin{equation}
\label{eq:tmp2}
 \lambda_m(\lambda_m-1)({\cal S}_0(p_m))^2 + 2\lambda_m {\cal S}_0(p_m) \sum_{i \neq m}\lambda_i{\cal S}_0(p_i) + \left(\sum_{i\neq m} \lambda_i {\cal S}_0(p_i)\right)^2-\sum_{i \neq m}\lambda_i ({\cal S}_0(p_i))^2=0.
\end{equation}
Equation (\ref{eq:tmp2}) has a real solution ${\cal S}_0(p_m) \in {\mathbbm R}$ if, and only if, its discriminant is non-negative, which is equivalent to
\begin{equation}
\label{eq:tmp3}
 -4\lambda_m \sum_{i<j:\ i,j\neq m} \lambda_i\lambda_j({\cal S}_0 (p_i) - {\cal S}_0(p_j))^2 \geq 0.
\end{equation}
Because all $\lambda_i >0$, (\ref{eq:tmp3}) is satisfied if, and only if, ${\cal S}_0(p_i) = {\cal S}_0(p_j)$ for all $i,j \neq m$. However, this is true for every $m$, which allows us to conclude that it must necessarily
be ${\cal S}_0(p_i)={\cal S}_0(p)$ for all $i$. By applying exactly the same reasoning to ${\cal S}_1$, we finally obtain
\begin{equation}
 \label{eq:tmp4}
 {\cal S}_x(p_i)={\cal S}_x(p), \qquad 0\leq x \leq 1, \quad 1 \leq i \leq k.
\end{equation}
The assumption that $\varphi(p)\notin \mathrm{Ext}(\mathbbm{P}_2^{{\mathfrak S}_n})$ leaves only the possibility that all ${\cal S}_{01}(p_i)$ must be different. Since in any local deterministic strategy one
has ${\cal S}_{01}(r)={\cal S}_{0}(r) {\cal S}_{1}(r) - {\cal Z}(r)$, Equation (\ref{eq:tmp4}) implies that all the ${\cal Z}(p_i)$ have to be different. The fourth coordinate or (\ref{eq:cdecomp2}) then reads
\begin{equation}
 \label{eq:tmp5}
 a-b-c+d = \sum_{i=1}^{k} \lambda_i (a_i-b_i-c_i+d_i).
\end{equation}
Equations (\ref{eq:tmp5}, \ref{eq:tmp4}) together with the fact that $a+b+c+d=a_i+b_i+c_i+d_i=n$ imply (by appropriately adding and subtracting them) the following:
\begin{equation}
 \label{eq:tmp6}
 a=\sum_{i=1}^k \lambda_i a_i, \quad b=\sum_{i=1}^k \lambda_i b_i, \quad c=\sum_{i=1}^k \lambda_i c_i, \quad d=\sum_{i=1}^k \lambda_i d_i,
\end{equation}
which is equivalent to $p=\sum_{i=1}^k \lambda_i p_i$; \textit{i.e.} $p$ is a convex combination of elements of ${\mathbbm{T}_n}$ with the same weights as the convex combination of $\varphi(p)$.
Finally, if $p\in \partial{\mathbbm T}_n$, then one of its coordinates must be $0$. However, if any of its coordinates is $0$, Equations (\ref{eq:tmp4}, \ref{eq:tmp5}) and $a+b+c+d=a_i+b_i+c_i+d_i$ imply that $p=p_i\ \forall i$, contradicting the fact that  (\ref{eq:cdecomp}) is a proper convex decomposition. This contradiction comes from the fact that we assumed that $\varphi(p)$ was not a vertex of $\mathbbm{P}_2^{{\mathfrak S}_n}$. Hence, $\varphi(p) \in \mathrm{Ext}(\mathbbm{P}_2^{{\mathfrak S}_n})$.
\end{proof}

\subsection{Proof of Theorem \ref{thm:class}}
\begin{proof}

 We consider the function $I$ defined as
 \begin{equation}
  \label{eq:tmp10}
  I:=\frac{1}{2}[n(x+y)^2+(\sigma \mu \pm x)^2-1] + x[\sigma \mu \pm (x+y)]{\cal S}_0 + \mu y {\cal S}_1 + \frac{x^2}{2} {\cal S}_{00} + \sigma x y {\cal S}_{01} + \frac{y^2}{2}{\cal S}_{11}
 \end{equation}
and we wish to show that $\min_{{\mathbbm P}_2^{{\mathfrak S}_n}} I = 0$.
For all local deterministic strategies, $I$ has the form
 \begin{eqnarray}
 \label{eq:tmp11}
  I&=&\frac{1}{2}[n(x+y)^2 + (\sigma \mu \pm x)^2 -1]+x[\sigma \mu \pm (x+y)]{\cal S}_0 + \mu y{\cal S}_1 + \frac{x^2}{2}\left({\cal S}_0^2 - n\right)+\sigma x y \left({\cal S}_0{\cal S}_1 - {\cal Z}\right)+\frac{y^2}{2}\left({\cal S}_1^2 - n\right)\nonumber\\
  &=&\frac{1}{2}\left(x^2{\cal S}_0^2+2\sigma x y {\cal S}_0{\cal S}_1 + y^2 {\cal S}_1^2\right) + xy n -\sigma xy{\cal Z} + \frac{1}{2}(\sigma \mu \pm x)^2 +(\sigma \mu \pm x)x{\cal S}_0 \pm xy{\cal S}_0+\mu y {\cal S}_1-\frac{1}{2}\nonumber\\
  &=&\frac{1}{2}\left(x{\cal S}_0+\sigma y {\cal S}_1\right)^2+\frac{1}{2}\left[(\sigma \mu \pm x)^2 + 2(\sigma \mu \pm x)x{\cal S}_0 +2(\sigma \mu \pm x)\sigma y {\cal S}_1\right]\mp \sigma x y {\cal S}_1 \pm xy {\cal S}_0 -\sigma xy {\cal Z} + xy n -\frac{1}{2}\nonumber\\
  &=&\frac{1}{2}\left(x{\cal S}_0 +\sigma y {\cal S}_1 +\sigma \mu \pm x\right)^2+ xy \left(\pm {\cal S}_0 \mp \sigma {\cal S}_1 -\sigma {\cal Z} + n\right)-\frac{1}{2}.\nonumber
 \end{eqnarray}
The condition $I\geq 0$ can be recast now as
\begin{equation}
 \label{eq:tmp12}
 \left(x{\cal S}_0 +\sigma y {\cal S}_1 +\sigma \mu \pm x\right)^2 + 8xy r \geq 1,
\end{equation}
where $r := (\pm {\cal S}_0 \mp \sigma {\cal S}_1 - \sigma {\cal Z} + n)/4$. We note that $r$ results into one of the variables $a,b,c,d$ which depends on the choice of the signs $\sigma, \pm$ we make.
Indeed, one has the following values for $r$:
\begin{equation}
 \label{eq:tmp13}
 \begin{array}{c|cc}
  \sigma\backslash \pm&+&-\\
  \hline
  +&b&c\\
  -&a&d
 \end{array}
\end{equation}
The left-hand side of (\ref{eq:tmp12}), denoted $\tilde I$, is always non-negative, because $x,y \in {\mathbbm{N}}$ and $r\geq 0$. We wish to show however that $0$ cannot be achieved, thus proving
that $I$ is indeed tangent to ${\mathbbm P}_2^{{\mathfrak S}_n}$. Note that if $r>0$ or $x{\cal S}_0 +\sigma y {\cal S}_1 +\sigma \mu \pm x\neq 0$, inequality (\ref{eq:tmp12}) is trivially satisfied.
Thus, the set of points in which $\tilde I=0$ could be possible needs to be a subset of $\tilde I_0$, where $\tilde I_0$ is defined as
\begin{equation}
 \label{eq:tmp14}
 \tilde I_0:= \{(a,b,c,d) \in {\mathbbm{T}}_n:\ r=0,\ x{\cal S}_0 +\sigma y {\cal S}_1 +\sigma \mu \pm x=0\}.
\end{equation}
Observe that when we take $(a,b,c,d)$ to be continuous, then $\tilde I_0$ is a line lying on a facet of $\partial {\mathbf T}_n$.
We shall now see how the conditions of Theorem \ref{thm:class} ensure that $\tilde I_0 = \emptyset$. Let us discuss the parity of $x{\cal S}_0 +\sigma y {\cal S}_1 +\sigma \mu \pm x$.
Note first that for all $m\in \mathbbm{Z}$, $m\equiv m^2\equiv -m \mod 2$. Then, it follows that ${\cal S}_0 \equiv {\cal S}_1 \equiv n \mod 2$. If $n$ is even, $x{\cal S}_0 +\sigma y {\cal S}_1 +\sigma \mu \pm x \equiv \mu + x \mod 2$. But $\mu + x \equiv 1 \mod 2$ by hypothesis. Otherwise, if $n$ is odd, then $x{\cal S}_0 +\sigma y {\cal S}_1 +\sigma \mu \pm x \equiv x+y+\mu + x\equiv \mu + y \mod 2$ and in this case $\mu + y \equiv 1 \mod 2$ again by hypothesis. Hence, $\tilde I_0 = \emptyset$, which implies (\ref{eq:tmp12}), which in turn proves that (\ref{eq:parametersclass}, \ref{eq:classicalboundclass}) is a valid Bell inequality, tangent to ${\mathbbm P}_2^{{\mathfrak S}_n}$ for any $n$.

\end{proof}

\subsection{Proof of Theorem \ref{thm:numberofvertices}}
\begin{proof}
 We will prove the theorem by explicitly solving the Diophantine equation
 \begin{equation}
  \frac{1}{2}[n(x+y)^2+(\sigma \mu \pm x)^2-1] + x[\sigma \mu \pm (x+y)]{\cal S}_0 + \mu y {\cal S}_1 + \frac{x^2}{2} {\cal S}_{00} + \sigma x y {\cal S}_{01} + \frac{y^2}{2}{\cal S}_{11}=0
  \label{eq:tmp20}
 \end{equation}
 over the integers. As discussed in the proof of Theorem \ref{thm:class}, (\ref{eq:tmp20}) is saturated on $r=0$ and $(x{\cal S}_0 +\sigma y {\cal S}_1 +\sigma \mu \pm x)^2=1$. Let us rewrite the last condition as $x{\cal S}_0 +\sigma y {\cal S}_1 +\sigma \mu \pm x + \tau = 0$, where $\tau \in \{-1,1\}$. In the notation introduced in (\ref{eq:reparametrization}) this condition reads $\tilde{K}_r(\tau) \pm 2(xs-yt)=0$, where $\tilde{K}_r(\tau)$ is defined as $\tilde{K}_r(\tau):=\pm n(y-x) + \sigma \mu \pm x + \tau$.
 Let us notice that the assumptions of Theorem \ref{thm:class} guarantee that $\tilde{K}_r(\tau)$ is even: Indeed, if $n\equiv 0 \mod 2$, then $\tilde{K}_r(\tau) \equiv \mu + x + 1\equiv 0 \mod 2$ (because the assumptions on Theorem \ref{thm:class} guarantee that $\mu$ has opposite parity to $x$ when $n$ is even); otherwise, if $n\equiv 1 \mod 2$ then $\tilde{K}_r(\tau) \equiv y+\mu+1\equiv 0 \mod 2$ (because $\mu$ has opposite parity to $y$ when $n$ is odd). Hence, we can define $K_r(\tau):=\tilde{K}_r(\tau)/2$ and we just showed that $K_r(\tau) \in {\mathbb Z}$.

 Now, the condition $\tilde{K}_r(\tau) \pm 2(xs-yt)=0$ is equivalent to $K_r(\tau) \pm (xs-yt)=0$, which we can solve for $t$: $yt=\pm K_r(\tau)+ xs$. This equation, when taken \textit{modulo} $x$, reads $yt\equiv \pm K_r(\tau) \mod x$. Since $\gcd(x,y)=1$ by hypothesis, this equation has a solution $t\equiv \pm y^{-1}K_r(\tau) \mod x$, where $y^{-1}$ is the inverse of $y$ in the group of integers \textit{modulo} $x$, typically denoted ${\mathbb Z}_x$, (such an inverse exists if, and only if, $\gcd(x,y)=1$). In practice, computation of $y^{-1}$ is done via the B\'ezout identity, which says that for any pair of integers $x,y \in \mathbb{Z}$ there exist (not unique) $p,q \in \mathbb{Z}$ such that $px+qy = \gcd(x,y)$. In the case that $\gcd(x,y)=1$ when taking B\'ezout's identity \textit{modulo} $x$ we obtain $qy\equiv 1 \mod x$; because of the assumption $\gcd(x,y)=1$ of the Theorem, the inverse of $y$ \textit{modulo} $x$ is well defined and we can write $q\equiv y^{-1}\mod x$.

 Let us now solve Equation (\ref{eq:tmp20}) for the other variables $s,u$, but we will keep the following form for $t$, in order to count the total number of solutions: $t(\tau) = t_0(\tau) + kx$, where $k\in {\mathbb Z}$ and $0\leq t_0(\tau) < x$. Then, from $K_r(\tau)\pm (xs-yt)=0$ we solve for $s$ and obtain $s(\tau)=s_0(\tau)+ky$, where $s_0(\tau):=(\mp K_r(\tau)+yt_0(\tau))/x$. Note that $s_0(\tau)$ is well defined integer number, a fact that can be shown by directly solving $xs(\tau) = (\mp K_r(\tau) + yt_0(\tau))+kxy$, because $y t(\tau) \equiv y t_0(\tau) \equiv \pm K_r(\tau) \mod x$.
 Now, the last variable $u$ is directly obtained from the normalization condition $n=r+s+t+u$. Then, writing $u=n-s-t$ as $u(\tau)=u_0(\tau)-k(x+y)$, where $u_0(\tau):=(nx+K_r(\tau) -(x+y)t_0(\tau))/x$, which is integer by the same argument as $s_0(\tau)$.

 To sum up, the family of points $(r,s,t,u)\in \mathbbm{Z}^4$ for which (\ref{eq:tmp20}) is fulfilled is
 \begin{equation}
  \bigcup_{\tau \in \{-1,1\}}\{[0,s_0(\tau),t_0(\tau),u_0(\tau)]+k[0,x,y,-(x+y)]: \ k\in \mathbb{Z}\}
  \label{eq:tmp21}
 \end{equation}
 Geometrically, this corresponds to alternating points in a \textit{zig-zag} pattern along two parallel lines which lie on the facet $r=0$ of $\mathbb{T}_n$.

 Let us now study which is the number of solutions that really belong to $\mathbb{T}_n$; \textit{i.e.} those for which $r,s,t,u \geq 0$. Since we have the family of solutions indexed by $k$, we just have to count how many $k$'s are available.
 The condition $t\geq 0$ leads to $t_0 + k x \geq 0$, which for $k\in \mathbb{Z}$ means $k\geq \lceil-t_0/x\rceil$. Now, taking into account that $t_0$ is chosen to be $0\leq t_0<x$, which is equivalent to $0\geq-t_0/x>-1$; equivalently $\lceil -t_0/x \rceil = 0$. Thus, $k \geq 0$.
 The condition $s \geq 0$ becomes $s_0 + k y \geq 0$, which for $k\in \mathbb{Z}$ is $k\geq \lceil -s_0/y\rceil$.
 Finally, the condition $u\geq 0$ is $u_0 - k (x+y) \geq 0$, which for $k \in \mathbb{Z}$ is $k \leq \lfloor u_0/(x+y) \rfloor$.

 Hence, for each $\tau \in \{-1,1\}$, the number of solutions belonging to $\mathbb{T}_n$ is given by $|\{k\in \mathbb{Z}: k\geq \max\{0,\lceil -s_0(\tau)/y\rceil\} , k \leq \lfloor u_0(\tau)/(x+y)\rfloor \}|$, which is $\max\{0,\lfloor u_0(\tau)/(x+y)\rfloor -\max \{0,\lceil -s_0(\tau)/y\rceil\} + 1\}$. Note that we put the first maximum just to ensure that there is always a non-negative number of solutions, denoted $N_S$.
 Finally, we obtain $N_S$ by summing this expression over the possible values for $\tau$:
 \begin{equation}
  N_S=\sum_{\tau =\pm 1}\max\left\{0,\left\lfloor\frac{u_0(\tau)}{x+y}\right\rfloor-\max\left\{0,\left\lceil\frac{-s_0(\tau)}{y}\right\rceil\right\}+1\right\}.
  \label{eq:tmp22}
 \end{equation}
\end{proof}

\subsection{Proof of Theorem \ref{thm:Blocks}}
\begin{proof}
 Recall that the Bell Operator can be decomposed into a sum as
 \begin{equation}
  \begin{array}{ccl}
  {\cal B}(\varphi, \theta)&=&\alpha {\cal S}_0 + \beta {\cal S}_1 + \frac{\gamma}{2} {\cal S}_{00} + \delta {\cal S}_{01} + \frac{\varepsilon}{2} {\cal S}_{11} + \beta_c {\mathbbm{1}_{2^n}}\\
  &=&A S_z + A' S_x + B S_{zz} + C S_{xx} + D S_{xz} + \beta_c \mathbbm{1}_{2^n},
 \end{array}
 \end{equation}
where $S_k = \sum_{v}\sigma_k^{(v)}$ and $S_{kl}=\sum_{v\neq w}\sigma_k^{(v)}\sigma_l^{(w)}$, with $k,l \in \{x,z\}$.
By projecting everything into the space ${\cal H}_J$ of (\ref{eq:Schur-Weyl}) we obtain the analytic expression of every block ${\cal B}_J(\varphi, \theta)$, as the projection of $S_k$ and $S_{k,l}$ is given by Theorem \ref{thm:bigcalculation}.
Simply by joining terms we arrive at the form (\ref{eq:penta-block}).
\end{proof}

\subsection{Proof of Theorem \ref{thm:unitaryrotation}}
\begin{proof}
By a direct check one finds that for any $c\in\mathbbm{R}$, $U(c)\mathcal{M}_0(\varphi)U(c)^{\dagger}=\mathcal{M}_0(\varphi+c)$, where $U(c)$ is given by Eq. (\ref{eq:unitary}).
Obviously, the same holds for $\mathcal{M}_1(\theta)$. Consequently,
\begin{equation}
\mathcal{B}(\varphi+c,\theta+c)=[U(c)]^{\ot n}\mathcal{B}(\varphi,\theta)[U(c)^{\dagger}]^{\ot n},
\end{equation}
which in particular means that the spectrum of $\mathcal{B}(\varphi,\theta)$ is invariant under
mutual translations of both its arguments, and therefore it depends only on the difference $\varphi-\theta$.
\end{proof}

\subsection{Proof of Theorem \ref{thm:corollary}}
\begin{proof}
We have to make $C=0$ with $C$ being defined in (\ref{eq:def-C}). Hence,
$\theta_{\pm}$ is the solution to $C=\gamma/2 \sin^2 \varphi + \delta \sin
\varphi \sin \theta + \varepsilon/2 \sin^2\theta=0$. This means that
\begin{eqnarray}
\sin\varphi& =&\frac{-2\delta \sin \theta \pm \sqrt{4\delta^2 \sin^2
\theta -4\gamma \varepsilon \sin^2\theta}}{2\gamma}=\frac{-\delta \sin \theta
\pm \sqrt{\delta^2-\gamma \varepsilon }|\sin \theta |}{\gamma}\nonumber\\
&=&\frac{-\delta \pm \sqrt{\delta^2-\gamma \varepsilon }}{\gamma}\sin
\theta.
\end{eqnarray}
Note that in the third equality the absolute value is omitted due to the freedom of choice in the sign $\pm$, which can be absorbed into the $\sin \theta$.

Let us denote $$\xi := \frac{-\delta \pm \sqrt{\Delta}}{\gamma}, \quad \Delta :=
\delta^2 - \gamma \varepsilon.$$ Then, we wish to solve for $\theta$ the
equation $\sin(\theta - \kappa)=\xi \sin(\theta)$ for $\xi, \kappa \in
\mathbb{R}$.
 We can assume that $\kappa \neq m \pi,\ m \in {\mathbb Z}$, otherwise the
equation is trivially satisfied by setting $\theta = 0$. But then ${\cal M}_0 =
\pm {\cal M}_1$, which is a degenerate case.
 Then, by expanding the $\sin(\theta - \kappa)$, one obtains
 $$\xi = \frac{\sin \theta \cos \kappa - \cos \theta \sin \kappa}{\sin
\theta}=\cos \kappa - \frac{\sin \kappa}{\tan \theta},$$ which leads to
 $$\theta = \arctan\left(\frac{\sin \kappa}{\cos \kappa - \xi}\right) = \arctan
\left(\frac{\gamma \sin \kappa}{\gamma \cos \kappa + \delta \pm
\sqrt{\Delta}}\right).$$
\end{proof}

\subsection{Proof of Theorem \ref{thm:qv}}
\begin{proof}

Let us first recall a known result from probability theory:
The non-central moments of the gaussian distribution are given by
\begin{equation}
\mathbb{E}_{\mu, \sigma}(x^k)= \int_{-\infty}^\infty x^k
\frac{e^{-(x-\mu)^2/2\sigma}}{\sqrt{2 \pi \sigma}}\mathrm{d}x =
\sigma^{k/2}H_k(\mu \sigma^{-1/2}) \qquad (k=0,1,\ldots),\nonumber
\end{equation}
where $H_k(x)$ are the Hermite polynomials, which can be defined as
\begin{equation}
 H_k(x)=e^{-x^2/2}\frac{d^k}{dx^k}e^{x^2/2}.
 \nonumber
\end{equation}
So, we have
\begin{equation}
\{\mathbb{E}_{\mu, \sigma}(x^k)\}_{k\geq
0}=\{1,\mu,\mu^2+\sigma, \mu^3 + 3 \mu \sigma, \mu^4 + 6\mu^2\sigma + 3
\sigma^2, \mu^5 + 10 \mu^3 \sigma+ 15 \mu \sigma^2, \ldots \}.
 \label{eq:lema1}
\end{equation}

Additionally, we shall make use of the two following facts in the proof:
\begin{itemize}
 \item Let $c_1, c_2 \in \mathbbm{R}$. Then
 \begin{equation}
 \int_{-\infty}^\infty x^k
\frac{e^{-(x-\mu+c_1)^2/4\sigma-(x-\mu+c_2)^2/4\sigma}}{\sqrt{2 \pi
\sigma}}\mathrm{d}x =e^{-(c_1-c_2)^2/8\sigma}\mathbb{E}_{\mu',\sigma}(x^k),
  \label{eq:lema2}
 \end{equation}
where $\mu' = \mu -(c_1+c_2)/2$.
This can easily be verified by using the following expression
\begin{equation}
\frac{(x-a)^2+(x-b)^2}{2}=\left[x-\left(\frac{a+b}{2}
\right)\right]^2+\left(\frac{a-b}{2}\right)^2,\nonumber
\end{equation}
 with $a=\mu-c_1$ and $b=\mu-c_2$.
 \item Let us consider the family of states $$\ket{\varphi_n} = \frac{1}{{\cal
N}_n}\sum_{k=0}^n \gamma_k^{(n)} \ket{D^k_n} = \frac{1}{{\cal
N}_n}\sum_{k=0}^n\frac{e^{-(k-\mu)^2/4\sigma}}{\sqrt[4]{2\pi
\sigma}}\ket{D^k_n},$$
where $\mu = n/2 + A/(2B - C)$ and $e^{-1}/2\pi\ll \sigma \ll n$.
Then, one finds that $\lim_{n\rightarrow \infty}{\cal N}_n=1$ because for large
$n$ the following holds:
\begin{equation}
 1=\langle \varphi_n | \varphi_n \rangle = \frac{1}{{\cal N}_n^2}
\sum_{k=0}^n \frac{e^{-(k-\mu)^2/2\sigma}}{\sqrt{2\pi\sigma}} \simeq
\frac{1}{{\cal N}_n^2}\int_{0}^n\frac{e^{-(k-\mu)^2/2\sigma}}{\sqrt{2\pi\sigma}}
\mathrm{d}k \simeq \frac{1}{{\cal N}_n^2}.\nonumber
\end{equation}
Observe that the choice of $\sigma$ cannot be too small: otherwise the first approximation (the sum by an integral) would be invalid; $\sigma$ cannot be too large either as
the second approximation (integration over $\mathbbm{R}$ instead of the interval $[0,n]$) would be incorrect. This is the reason why we added the constraints $e^{-1}/2\pi\ll \sigma \ll n$.
\end{itemize}

Let us then assume that $n$ is large enough so that we can omit the factor ${\cal N}_n$ for practical purposes.
To calculate $\bra{\psi_n}{\cal B}_{n/2}(\varphi(\theta), \theta)\ket{\psi_n}$,
let us first express $d_k$ and $u_k$ conveniently as polynomials in terms of $k$
(note that $C=0$):
 \begin{itemize}
  \item $d_k = 2Bk^2-2(A+nB)k+[\beta_c + n(A-B/2)+n^2B/2]$, for $0\leq k \leq
n$.
  \item $u_k=u_{k'+m-1}\equiv \tilde{u}_{k'} =
(A'-2k'D)m\sqrt{1-\frac{k'^2}{m^2}}$, for $-m+1\leq k' \leq m-1$, where
$m=(n+1)/2$ and $k'=k-(n-1)/2$. By means of the Taylor expansion of
$\sqrt{1-x^2}$ around $x=0$,
$$\sqrt{1-x^2}=\sum_{l=0}^\infty c_l
x^{2l}=1 - \frac{1}{2}x^2 -\frac{1}{8}x^4 - \frac{1}{16}x^6 - \frac{5}{128}x^8 +
o(x^{10}),$$
  we can express $u_k'$ as a polynomial in $k'$; this is a good approximation as long as
$|k'|\ll m$:
$$u_{k'}=(A'-2k'D)\left(m-\frac{k'^2}{2m}-\frac{k'^4}{8m^3}-\frac{k'^6}{16m^5}
-\frac{5k'^8}{128m^7}-\ldots\right).$$
  \item $v_k=0$ for $0\leq k \leq n-2$.
 \end{itemize}

 Then, the expectation value of the elements of the diagonal of
${\cal B}_{n/2}(\varphi(\theta), \theta)$ is given by
 \begin{eqnarray}
  \sum_{k=0}^n (\psi_k^{(n)})^2 d_k &\simeq& \int_{0}^n d_k
\frac{e^{-(k-\mu)^2/2\sigma}}{\sqrt{2\pi \sigma}} \mathrm{d}k \simeq
\int_{-\infty}^\infty d_k\frac{e^{-(k-\mu)^2/2\sigma}}{\sqrt{2\pi \sigma}}
\mathrm{d}k\nonumber \\
  &=& 2B\mathbb{E}_{\mu, \sigma}(x^2) -2(A+nB)\mathbb{E}_{\mu, \sigma}(x) +
[\beta_c + n(A-B/2)+n^2B/2] \mathbb{E}_{\mu, \sigma}(1)\nonumber \\
  &=& 2B(\mu^2+ \sigma)-2(A+nB)\mu + [\beta_c + n(A-B/2)+n^2B/2]\nonumber \\
  &=& 0\cdot n^2 + \left(\frac{\beta_c}{n} -\frac{B}{2} \right)n + 2B \sigma
-\frac{A^2}{2B},
 \end{eqnarray}
where we have applied (\ref{eq:lema1}).

The contribution to the expectation value for the elements that correspond
to the off-diagonal terms of ${\cal B}_{n/2}(\varphi(\theta),\theta)$ is given by (let us
write in short $c\equiv A/2B$)
 \begin{eqnarray}
  &&\hspace{-0.5cm}\sum_{k=0}^{n-1} \psi_k^{(n)} \psi_{k+1}^{(n)} u_k=
\sum_{{k'}=-(m-1)}^{m-1} \psi_{k'+m-1}^{(n)} \psi_{{k'}+m}^{(n)}
\tilde{u}_{k'}\nonumber\\
&&\simeq  \int_{-m+1}^{m-1} \tilde{u}_{k'}
\frac{e^{-(k'+m-1-\mu)^2/4\sigma-(k'+m-\mu)^2/4\sigma}}{\sqrt{2\pi
\sigma}}\mathrm{d}k'\nonumber\\
  &&= \int_{-m+1}^{m-1} \tilde{u}_{k'}
\frac{e^{-(k'-(c-1/2))^2/4\sigma-(k'-(c+1/2))^2/4\sigma}}{\sqrt{2\pi
\sigma}}\mathrm{d}k'\nonumber\\
&&\simeq \int_{-\infty}^{\infty} \tilde{u}_{k'}
\frac{e^{-(k'-(c-1/2))^2/4\sigma-(k'-(c+1/2))^2/4\sigma}}{\sqrt{2\pi\sigma}}
\mathrm{d}k'\nonumber \\
  &&= \int_{-\infty}^{\infty}
(A'-2k'D)\left(m-\frac{k'^2}{2m}-\frac{k'^4}{8m^3}-\frac{k'^6}{16m^5}-\frac{
5k'^8}{128m^7}-\ldots\right)\nonumber\\
&&\hspace{5cm}\times\frac{e^{-(k'-(c-1/2))^2/4\sigma-(k'-(c+1/2))^2/4\sigma}}{
\sqrt { 2\pi\sigma } } \mathrm{d}k'\nonumber \\
 &&= e^{-1/8\sigma}  \int_{-\infty}^{\infty}
(A'-2k'D)\left(m-\frac{k'^2}{2m}-\frac{k'^4}{8m^3}-\frac{k'^6}{16m^5}-\frac{
5k'^8}{128m^7}-\ldots\right)
\frac{e^{-(k'-c)^2/2\sigma}}{\sqrt{2\pi\sigma}}\mathrm{d}k'\nonumber \\
 &&= e^{-1/8\sigma} \left(A'm \mathbb{E}_{c,\sigma}(1) -2Dm
\mathbb{E}_{c,\sigma}(x) -
\frac{A'}{2m}\mathbb{E}_{c,\sigma}(x^2)+\frac{D}{m}\mathbb{E}_{c,\sigma}
(x^3)\right.\nonumber\\
&&\hspace{5cm}\left.-\frac{A'}{8m^3}\mathbb{E}_{c,\sigma}(x^4)+\frac{D}{4m^3}
\mathbb { E } _ { c , \sigma } (x^5) - \cdots\right)\nonumber \\
 &&=e^{-1/8\sigma}\left(\frac{A'-2Dc}{2}n + \frac{A'-2Dc}{2} -
\frac{(c^2+
\sigma)A'-2Dc(c^2+3\sigma)}{2}\frac{2}{n+1}+O(\sigma^2n^{-3})\right),\nonumber
 \end{eqnarray}
where we have applied (\ref{eq:lema1} and \ref{eq:lema2}).
Hence, the expectation value $\bra{\psi_n}{\cal B}_{n/2}(\varphi(\theta),
\theta)\ket{\psi_n}$ is given by
\begin{eqnarray}
\bra{\psi_n}B_{n}^S(\varphi(\theta),
\theta)\ket{\psi_n}&\nmsss=\nmsss&\sum_{k=0}^n \left(\psi_k^{(n)}\right)^2
d_k+ 2\sum_{k=0}^{n-1} \psi_k^{(n)}\psi_{k+1}^{(n)} u_k \nonumber\\
&\simeq& \left(\frac{\beta_c}{n}-\frac{B}{2} +
e^{-1/8\sigma}\left(A'-\frac{AD}{B}\right)\right) n + \left(2B \sigma -\frac{A^2}{2B}
+e^{-1/8\sigma}\left(A'-\frac{AD}{B}\right)\right) + o(\sigma n^{-1}).\nonumber
\end{eqnarray}
\end{proof}

\subsection{Proof of Theorem \ref{thm:partialtrace}}
\begin{proof}
Since the state acts on the symmetric space, without loss of generality we will compute its partial trace after forgetting about the last $n-d$ subsystems.
In general, when $\rho$ acts on the Hilbert Space of $n$ qubits $({\mathbb C}^2)^{\otimes n}$, its parcial trace after forgetting about the last $n-d$ subsystems is given by
$$\left(\mathrm{Tr}_{n-d}(\rho)\right)^{i_0\ldots i_{d-1}}_{j_0 \ldots j_{d-1}} = \sum_{0\leq i_d, \ldots, i_{n-1}\leq 1} \rho^{i_0,\ldots,i_{d-1},i_{d},\ldots, i_{n-1}}_{j_0,\ldots, j_{d-1},i_d,\ldots, i_{n-1}}.$$
Now, let us expand $\rho_{\cal S}$ to the full Hilbert Space:
\begin{eqnarray}
(\rho)^{\mathbf{i}}_{\mathbf{j}}=(p\rho_{\cal S}p^T)^{\mathbf{i}}_{\mathbf j} = \sum_{0\leq k,l\leq n} p^\mathbf{i}_k(\rho_{\cal S})^k_l p^{\mathbf{j}}_l = \sum_{0\leq k,l\leq n}\frac{(\rho_{\cal S})^{k}_{l}}{\sqrt{{n \choose k}{n \choose l}}}\delta(k-|\mathbf{i}|)\delta(l-|\mathbf{j}|),\nonumber
\end{eqnarray}
where $p$ is the same matrix for the change of basis as the one defined in the proof of Theorem \ref{thm:bigcalculation}, but in this case $m=0$.\\
Joining the last two expressions and introducing the notation $\overline{\mathbf{i}'}:=i_d\ldots i_{n-1}$ and $\overline{\mathbf{j}'}:=j_d\ldots j_{n-1}$, we have
\begin{eqnarray}
\left(\mathrm{Tr}_{n-d}(\rho)\right)^{\mathbf i'}_{\mathbf j'}
=\left(\mathrm{Tr}_{n-d}(p\rho_{\cal S}p^T)\right)^{\mathbf i'}_{\mathbf j'}=\sum_{\overline{\mathbf i'}}\sum_{0\leq k,l\leq n}\frac{(\rho_{\cal S})^{k}_{l}}{\sqrt{{n \choose k}{n \choose l}}}\delta(k-|\mathbf{i}'|-|\overline{\mathbf{i}'}|)
\delta(l-|\mathbf{j}'|-|\overline{\mathbf{i}'}|)\nonumber\\
= \sum_{\overline{\mathbf i'}}\sum_{0\leq k,l\leq n}\frac{(\rho_{\cal S})^{k}_{l}}{\sqrt{{n \choose k}{n \choose l}}}\delta(k-|\mathbf{i}'|-|\overline{\mathbf{i}'}|)
\delta(k-|\mathbf{i}'| -l+|\mathbf{j}'|)\nonumber\\
= \sum_{0\leq k,l\leq n}\frac{(\rho_{\cal S})^{k}_{l}}{\sqrt{{n \choose k}{n \choose l}}}
\delta(k-|\mathbf{i}'| -l+|\mathbf{j}'|)\sum_{\overline{\mathbf i'}}\delta(k-|\mathbf{i}'|-|\overline{\mathbf{i}'}|),\nonumber\\
\end{eqnarray}
where in the third equality we have used that Kronecker deltas fulfill $\delta(a-b)\delta(a-c)=\delta(a-b)\delta(b-c)$.
The expression $\sum_{\overline{\mathbf i'}}\delta(k-|\mathbf{i}'|-|\overline{\mathbf{i}'}|)$ is equal to ${n-d \choose k-|\mathbf{i}'|}I_{[|\mathbf{i}'|,|\mathbf{i}'|+n-d]}(k)$, where $I_{[a,b]}(k)$ is the indicator function: it outputs $1$ if $a\leq k \leq b$ and $0$ otherwise. The reason behind that equality is that we are counting how many indices $\mathbf{i}$ with the first $d$ bits already set and with weight $|\mathbf{i}'|$ have total weight $k$. Hence, we have to choose from the remaining $n-d$ bits, indexed by $\overline{\mathbf{i}'}$, that they have weight $k-|\mathbf{i}'|$. Of course, this is possible only if $|\mathbf{i}'|$ is not already greater than $k$ and it is not smaller than $n-d-k$; that is the reason why the indicator function appears, so that $|\mathbf{i}'|\leq k \leq |\mathbf{i}'|+n-d$.
We can then write
\begin{eqnarray}
\left(\mathrm{Tr}_{n-d}(\rho)\right)^{\mathbf i'}_{\mathbf j'} = \sum_{0\leq k,l\leq n}\frac{(\rho_{\cal S})^{k}_{l}}{\sqrt{{n \choose k}{n \choose l}}}
\delta(k-|\mathbf{i}'| -l+|\mathbf{j}'|){n-d \choose k-|\mathbf{i}'|}I_{[|\mathbf{i}'|,|\mathbf{i}'|+n-d]}(k)
\nonumber\\
=\sum_{0\leq k\leq n}\frac{(\rho_{\cal S})^{k}_{k-|\mathbf{i}'|+|\mathbf{j}'|}}{\sqrt{{n \choose k}{n \choose k-|\mathbf{i}'|+|\mathbf{j}'|}}}{n-d \choose k-|\mathbf{i}'|}I_{[|\mathbf{i}'|,|\mathbf{i}'|+n-d]}(k)\nonumber\\
=\sum_{|\mathbf{i}'|\leq k\leq |\mathbf{i}'|+n-d}\frac{(\rho_{\cal S})^{k}_{k-|\mathbf{i}'|+|\mathbf{j}'|}}{\sqrt{{n \choose k}{n \choose k-|\mathbf{i}'|+|\mathbf{j}'|}}}{n-d \choose k-|\mathbf{i}'|}.\nonumber
\end{eqnarray}
Finally we apply the change of variables $m:=k-|\mathbf{i}'|$ in order to obtain
\begin{eqnarray}
\left(\mathrm{Tr}_{n-d}(\rho)\right)^{\mathbf i'}_{\mathbf j'} = \sum_{0\leq m\leq n-d}\frac{{n-d \choose m}(\rho_{\cal S})^{m+|\mathbf{i}'|}_{m+|\mathbf{j}'|}}{\sqrt{{n \choose m+|\mathbf{i}'|}{n \choose m+|\mathbf{j}'|}}},\nonumber
\end{eqnarray}
which completes the proof. Observe that the expression only depends on the weights of the indices, so it is permutationally invariant, as expected.
\end{proof}

\subsection{Proof of Theorem \ref{thm:2bodyreduced}}
\begin{proof}
 We shall make use of the formula provided in Theorem \ref{thm:partialtrace} and use the same kind of approximations as in the proof of Theorem \ref{thm:qv}.
 We shall use the definition for the function $f$ given in Eq. (\ref{eq:f_for_d=2}) and its relation (\ref{eq:geometricmean}). We shall also adopt the notation (\ref{eq:lema1}).
 The simplest elements of $\rho_2$ to begin with are its diagonal elements:
 \begin{equation}
 \begin{array}{rcl}
(\rho_2)^{00}_{00} &\simeq& \int_{0}^{n-2}f(n,x,2,0,0) \psi_{x}^2 \mathrm{d}x = \frac{1}{n(n-1)}\int_{0}^{n-2}(n-x)(n-x-1)\frac{e^{-(x-\mu)^2/2\sigma}}{\sqrt{2\pi \sigma}}\mathrm{d}x\nonumber \\
&\simeq&\frac{1}{n(n-1)}\left(n(n-1) \mathbb{E}_{\mu, \sigma}(1) - (2n-1)\mathbb{E}_{\mu, \sigma}(x) + \mathbb{E}_{\mu, \sigma}(x^2)\right),\\
(\rho_2)^{01}_{01} &\simeq& \int_{0}^{n-2}f(n,x,2,1,1) \psi_{x+1}^2 \mathrm{d}x = \frac{1}{n(n-1)}\int_{0}^{n-2}(n-x-1)(x+1)\frac{e^{-(x-(\mu-1))^2/2\sigma}}{\sqrt{2\pi \sigma}}\mathrm{d}x\nonumber \\
&\simeq&\frac{1}{n(n-1)}\left((n-1) \mathbb{E}_{\mu-1, \sigma}(1) + (n-2)\mathbb{E}_{\mu-1, \sigma}(x) - \mathbb{E}_{\mu-1, \sigma}(x^2)\right),\\
(\rho_2)^{11}_{11} &\simeq& \int_{0}^{n-2}f(n,x,2,2,2) \psi_{x+2}^2 \mathrm{d}x = \frac{1}{n(n-1)}\int_{0}^{n-2}(x+1)(x+2)\frac{e^{-(x-(\mu-2))^2/2\sigma}}{\sqrt{2\pi \sigma}}\mathrm{d}x\nonumber \\
&\simeq&\frac{1}{n(n-1)}\left(2 \mathbb{E}_{\mu-2, \sigma}(1) +3 \mathbb{E}_{\mu-2, \sigma}(x) + \mathbb{E}_{\mu-2, \sigma}(x^2)\right).
\end{array}
\end{equation}
Note that $(\rho_2)^{01}_{01} = (\rho_2)^{10}_{10} = (\rho_2)^{10}_{01} = (\rho_2)^{01}_{10}$. Observe as well that $\Tr\rho_2$ does not depend neither on $\mu$ nor $\sigma$ and it is $1$.
Now we move to the off-diagonal elements, starting with $(\rho_2)^{00}_{11}$, which is equal to $(\rho_2)^{11}_{00}$.
\begin{eqnarray}
 \begin{array}{rcl}
(\rho_2)^{00}_{11}&\simeq& \int_{0}^{n-2} f(n,x,2,0,2) \psi_x \psi_{x+2} \mathrm{d}x\nonumber \\
&=&\frac{1}{n(n-1)}\int_{0}^{n-2} \sqrt{(n-x)(n-x-1)(x+1)(x+2)} \frac{e^{-\frac{(x-\mu)^2}{4\sigma}-\frac{-(x-\mu+2)^2}{4\sigma}}}{\sqrt{2\pi \sigma}} \mathrm{d}x \nonumber \\
&\simeq&\frac{e^{-1/2\sigma}}{n(n-1)}\int_{-\infty}^{\infty}(n-x-1/2)(x+3/2) \frac{e^{-\frac{(x-(\mu-1))^2}{2\sigma}}}{\sqrt{2\pi \sigma}} \mathrm{d}x\nonumber \\
&=&\frac{e^{-1/2\sigma}}{n(n-1)}\int_{-\infty}^{\infty} (n+1/2-y)(y+1/2) \frac{e^{-(y-\mu)^2/2\sigma}}{\sqrt{2\pi \sigma}}\mathrm{d}y\nonumber \\
&=&\frac{e^{-1/2\sigma}}{n(n-1)}\left(\frac{(2n+1)}{4}\mathbb{E}_{\mu, \sigma}(1) + n \mathbb{E}_{\mu, \sigma}(x) - \mathbb{E}_{\mu, \sigma}(x^2)\right).
\end{array}
\end{eqnarray}
Here we have used the approximations $\sqrt{(n-x)(n-x-1)}\simeq (n-x-1/2)$ and $\sqrt{(x+1)(x+2)}\simeq x+3/2$ and the change of variables $y\equiv x+1$.
Finally, we consider the cases $(\rho_2)^{00}_{01} = (\rho_2)^{00}_{10} = (\rho_2)^{01}_{00} = (\rho_2)^{10}_{00}$ and $(\rho_2)^{01}_{11} = (\rho_2)^{10}_{11} = (\rho_2)^{11}_{01} = (\rho_2)^{11}_{10}$.
To this end, let us recall that $\mu = n/2 + A/2B$. We shall denote $c\equiv A/2B$ and we shall also define $m\equiv (n+1)/2$ and  $\mu' \equiv \mu - 1/2$.
\begin{eqnarray}
 \begin{array}{rcl}
(\rho_2)^{00}_{01} &\simeq& \int_{0}^{n-2} f(n,x,2,0,1) \psi_x \psi_{x+1} \mathrm{d}x=\frac{1}{n(n-1)}\int_{0}^{n-2} (n-x-1)\sqrt{(n-x)(x+1)} \frac{e^{-\frac{(x-\mu)^2}{4\sigma}-\frac{-(x-\mu+1)^2}{4\sigma}}}{\sqrt{2\pi \sigma}} \mathrm{d}x  \nonumber \\
&\simeq&\frac{e^{-1/8\sigma}}{n(n-1)}\int_{-\infty}^{\infty}(n-x-1)\sqrt{(n-x)(x+1)} \frac{e^{-\frac{(x-\mu')^2}{2\sigma}}}{\sqrt{2\pi \sigma}} \mathrm{d}x=\frac{e^{-1/8\sigma}}{n(n-1)}\int_{-\infty}^{\infty} (m-1-y)m\sqrt{1-(y/m)^2} \frac{e^{-(y-c)^2/2\sigma}}{\sqrt{2\pi \sigma}}\mathrm{d}y\nonumber \\
&=&\frac{e^{-1/8\sigma}}{n(n-1)}\int_{-\infty}^{\infty} (m-1-y)m\sum_{l=0}^\infty c_l (y/m)^{2l} \frac{e^{-(y-c)^2/2\sigma}}{\sqrt{2\pi \sigma}}\mathrm{d}y  \nonumber \\
&=&\frac{e^{-1/8\sigma}}{n(n-1)}\int_{-\infty}^{\infty} (m-1-y)m\left(1 - \frac{y^2}{2m^2} - \frac{y^4}{8m^4}-O((y/m)^6)\right) \frac{e^{-(y-c)^2/2\sigma}}{\sqrt{2\pi \sigma}}\mathrm{d}y  \nonumber\\
&\simeq&\frac{e^{-1/8\sigma}}{n(n-1)}\left(m(m-1)\mathbb{E}_{c,\sigma}(1) - m \mathbb{E}_{c,\sigma}(x) - \frac{m-1}{2m} \mathbb{E}_{c,\sigma}(x^2) + \frac{1}{2m} \mathbb{E}_{c,\sigma}(x^3)- \frac{m-1}{8m^3} \mathbb{E}_{c,\sigma}(x^4) + \frac{1}{8m^3} \mathbb{E}_{c,\sigma}(x^5)\right)  \nonumber \\
&\simeq&\frac{e^{-1/8\sigma}}{n(n-1)}\left(\frac{(n+1)(n-1)}{4}-\frac{(n+1)}{2}c - \frac{n-1}{2(n+1)}(c^2+ \sigma)+\frac{1}{n+1}(c^3 + 3c \sigma)+ \cdots \right) \nonumber\\
&=&\frac{e^{-1/8\sigma}}{n(n-1)}\left(\frac{1}{4}n^2 - \frac{c}{2}n - \frac{2c^2+2c+1+2\sigma}{4} + \frac{c^2+\sigma + c^3+3c\sigma}{n} + \cdots \right), \nonumber\\
\end{array}
\end{eqnarray}
where $c_l$ are the coefficients of the Taylor expansion of $\sqrt{1-x^2}$ at $x=0$, we have used the change of variables $y\equiv x -(n-1)/2$ and the fact that
\begin{equation}
 \frac{1}{n-a}=\frac{1}{n}\sum_{k=0}^{\infty}(a/n)^k
\end{equation}
in order to obtain an expression purely in powers of $n$.

\begin{eqnarray}
 \begin{array}{rcl}
(\rho_2)^{01}_{11} &\simeq& \int_{0}^{n-2} f(n,x,2,1,2) \psi_{x+1} \psi_{x+2} \mathrm{d}x=\frac{1}{n(n-1)}\int_{0}^{n-2} (x+1)\sqrt{(n-x-1)(x+2)} \frac{e^{-\frac{(x-\mu+1)^2}{4\sigma}-\frac{-(x-\mu+2)^2}{4\sigma}}}{\sqrt{2\pi \sigma}} \mathrm{d}x  \nonumber \\
&\simeq&\frac{e^{-1/8\sigma}}{n(n-1)}\int_{-\infty}^{\infty}(x+1)\sqrt{(n-x-1)(x+2)} \frac{e^{-\frac{(x-(\mu'-1))^2}{2\sigma}}}{\sqrt{2\pi \sigma}} \mathrm{d}x\\
&=&\frac{e^{-1/8\sigma}}{n(n-1)}\int_{-\infty}^{\infty} (m-1+y)m\sqrt{1-(y/m)^2} \frac{e^{-(y-c)^2/2\sigma}}{\sqrt{2\pi \sigma}}\mathrm{d}y\nonumber \\
&=&\frac{e^{-1/8\sigma}}{n(n-1)}\int_{-\infty}^{\infty} (m-1+y)m\left(1 - \frac{y^2}{2m^2} - \frac{y^4}{8m^4}-O((y/m)^6)\right) \frac{e^{-(y-c)^2/2\sigma}}{\sqrt{2\pi \sigma}}\mathrm{d}y\nonumber\\
&\simeq&\frac{e^{-1/8\sigma}}{n(n-1)}\left(m(m-1)\mathbb{E}_{c,\sigma}(1) + m \mathbb{E}_{c,\sigma}(x) - \frac{m-1}{2m} \mathbb{E}_{c,\sigma}(x^2) - \frac{1}{2m} \mathbb{E}_{c,\sigma}(x^3)- \frac{m-1}{8m^3} \mathbb{E}_{c,\sigma}(x^4) - \frac{1}{8m^3} \mathbb{E}_{c,\sigma}(x^5)\right) \nonumber \\
&\simeq& \frac{e^{-1/8\sigma}}{n(n-1)}\left(\frac{(n+1)(n-1)}{4}+\frac{n+1}{2}c - \frac{n-1}{2(n+1)}(c^2+ \sigma)-\frac{1}{n+1}(c^3 + 3c \sigma)+ \cdots \right) \nonumber\\
&=&\frac{e^{-1/8\sigma}}{n(n-1)}\left(\frac{1}{4}n^2 + \frac{c}{2}n - \frac{2c^2+2c+1+2\sigma}{4} + \frac{c^2+\sigma -(c^3+3c\sigma)}{n} + \cdots \right), \nonumber\\
\end{array}
\end{eqnarray}
where we have used the change of variables $y\equiv x -(n-3)/2$.

Joining all terms, we obtain that $\rho_2$ can be written as
\begin{equation}
 \rho_2 = \frac{1}{n(n-1)}\left(
 \left(
 \begin{array}{cccc}
 1&e^{-1/8\sigma}&e^{-1/8\sigma}&e^{-1/2\sigma}\\
 e^{-1/8\sigma}&1&1&e^{-1/8\sigma}\\
 e^{-1/8\sigma}&1&1&e^{-1/8\sigma}\\
 e^{-1/2\sigma}&e^{-1/8\sigma}&e^{-1/8\sigma}&1\\
 \end{array}
\right)\frac{n^2}{4}+
\left(
 \begin{array}{cccc}
 -(2c+1)&-ce^{-1/8\sigma}&-ce^{-1/8\sigma}&e^{-1/2\sigma}\\
 -ce^{-1/8\sigma}&0&0&ce^{-1/8\sigma}\\
 -ce^{-1/8\sigma}&0&0&ce^{-1/8\sigma}\\
 e^{-1/2\sigma}&ce^{-1/8\sigma}&ce^{-1/8\sigma}&-1+2c\\
 \end{array}
\right)\frac{n}{2}+o(n)
 \right)
 \label{eq:tmp60}
\end{equation}
Since $\sigma \in O(n^{1/2})$, we can neglect the exponential terms and we obtain the form given in Eq. (\ref{eq:2bodyanalyticalreducedstate}).

\end{proof}

\subsection{Proof of Theorem \ref{thm:classDickeLowk}}
\begin{proof}
We begin by proving that the choice of coefficients (\ref{class:lowk}) defines a valid Bell inequality; \textit{i.e.}, (\ref{eq:2bodySymgeneral}) is non-negative on all points $\mathbbm{P}_2^{{\mathfrak S}_n}$.
 To this aim, we shall express the inequality (\ref{eq:2bodySymgeneral}) in terms of the variables $a,b,c,d$ and finding the minimum to its left hand side, while imposing the conditions $a, b, c, d \geq 0$. If we denote by $F(a,b,c,d)$ the left hand side of (\ref{eq:2bodySymgeneral}), then we find that it does not depend on the variable $d$, and, as a quadratic function of $a,b,c$, it can be written in the following form:
 \begin{equation}
  F(a,b,c,d)=(1,a,b,c,d)
\left(
\begin{array}{ccccc}
 4k(1+k)(1+2k)&2(1+2k)^2&-k(3+4k)&-k(3+4k)&0\\
 2(1+2k)^2&4(1+2k)&2(1+2k)&2(1+2k)&0\\
 -k(3+4k)&2(1+2k)&2k&2(1+k)&0\\
 -k(3+4k)&2(1+2k)&2(1+k)&2k&0\\
 0&0&0&0&0
\end{array}
\right)
\left(
\begin{array}{c}
 1\\a\\b\\c\\d
\end{array}
\right).
 \end{equation}
We aim to prove that $F\geq 0$ on all vertices. The condition for extrema is $\partial_a F = \partial_b F = \partial_c F = 0$, which reads
\begin{eqnarray}
 \partial_a F &=& 4 (1 + 2 k) [2 a + b + c - (1 + 2 k)]=0,\label{eq:tmp40}\\
 \partial_b F &=& 2 [ 2 b - (4 k + 3)] k + 4 c (1 + k) + 4a (1 + 2 k)=0,\label{eq:tmp41}\\
 \partial_c F &=& 2 [ 2 c - (4 k + 3)] k + 4 b (1 + k) + 4a (1 + 2 k)=0.\label{eq:tmp42}
\end{eqnarray}
We also have to take into account the boundary, so there are $8$ cases to consider, depending on the case if the inequalities $a,b,c\geq 0$ are saturated or not, which are summarized in the following table:

\begin{equation}
 \begin{array}{c|c|c}
 (a=0,b=0,c=0)&(a,b,c)&F(a,b,c)\\
 \hline
 (Y,Y,Y)&(0,0,0)&4k(1+k)(1+2k)\\
 (Y,Y,N)&(0,0,3/2+2k)&-k/2\\
 (Y,N,Y)&(0,3/2+2k,0)&-k/2\\
 (Y,N,N)&(0,\frac{k(3+4k)}{2(1+2k)},\frac{k(3+4k)}{2(1+2k)})&\frac{k(8k^2+11k+4)}{1+2k}\\
 (N,Y,Y)&(1/2+k,0,0)&-(1+2k)\\
 (N,Y,N)&(k/2,0,1+k)&k^2\\
 (N,N,Y)&(k/2,1+k,0)&k^2\\
 (N,N,N)&\mbox{The system is incompatible}&\mbox{N/A}
\end{array}
\label{eq:taula}
\end{equation}
Thus, in order to prove $F(a,b,c)\geq 0$, one only needs to look at the cases where $F(a,b,c)$ is negative. Since $k>0$, it is immediate to see from (\ref{eq:taula}) that this occurs when exactly two of the variables $a,b,c$ are equal to zero. Since in these cases $F$ is just a quadratic function of the remaining nonzero variable, we only need to check two cases, which are the neighboring integer values to the optimal (real number) given in (\ref{eq:taula}) for this nonzero variable. When exploring the case for which this nonzero variable is rounded below (cf. the case $a>0$ in table (\ref{eq:taula2})) we must also take into account the case when the value of the variables set to zero becomes $1$ or we might miss some solutions.
The following table summarizes the strategies for which $F=0$ (the rest are $F>0$ so they are not listed) thus proving inequality (\ref{class:lowk}). More precisely, from $c>0$ one finds the first two rows of (\ref{eq:taula2}), from $b>0$ one obtains the third and fourth row of (\ref{eq:taula2}) and the remaining ones are obtained from $a>0$.

Observe that the fact that (\ref{class:lowk}) is symmetric under the exchange of measurements translates in the following: if a vertex exists for which $F=0$, then by swapping $b$ and $c$ (recall that these variables corresponded to $({\cal M}_0, {\cal M}_1)=(+,-)$ and $({\cal M}_0, {\cal M}_1)=(-,+)$, respectively) we obtain another vertex for which $F=0$. Equivalently, by swapping the values of ${\cal S}_0$ and ${\cal S}_1$, the same happens.

\begin{equation}
 \begin{array}{cccc|cccc}
 a&b&c&d&{\cal S}_0&{\cal S}_1&{\cal Z}&n\\
 \hline
 0&0&1+2k&n-(2k+1)&-n&2(2k+1)-n&n-2(2k+1)&n\\
 0&0&2+2k&n-(2k+2)&-n&4(k+1)-n&n-4(k+1)&n\\
 \hline
 0&1+2k&0&n-(2k+1)&2(2k+1)-n&-n&n-2(2k+1)&n\\
 0&2+2k&0&n-(2k+2)&4(k+1)-n&-n&n-4(k+1)&n\\
 \hline
 k&0&0&n-(k+1)+1& 2k-n&2k-n&n&n\\
 k&0&1&n-(k+1)&2k-n&2(k+1)-n&n-2&n\\
 k&1&0&n-(k+1)&2(k+1)-n&2k-n&n-2&n\\
 k+1&0&0&n-(k+1)&2(k+1)-n&2(k+1)-n&n&n
\end{array}
\label{eq:taula2}
\end{equation}

\end{proof}

\subsection{Proof of Theorem \ref{thm:classDickeMidk}}
\begin{proof}
 We begin by making a general observation. For any vertex in ${\mathbbm{P}_2^{{\mathfrak S}_n}}$, the following inequalities hold:
 \begin{equation}
 -n \leq {\cal S}_{00}, {\cal S}_{11} \leq n(n-1), \qquad |{\cal S}_{01}|\leq n(n-1), \qquad  |{\cal S}_0|, |{\cal S}_1| \leq n.
  \label{eq:tmp50}
 \end{equation}
Hence, the dominating term in (\ref{class:midkeven} $-$ \ref{class:midkodd}) is the one corresponding to $\gamma {\cal S}_{00}/2$, which is $O(n^4)$, while the rest are $O(n^3)$. Hence, in order to minimize the value of the Bell inequality, a necessary condition is that the term containing ${\cal S}_{00}$ should be small, at least not bigger than $O(n^3)$. As ${\cal S}_{00}=({\cal S}_0)^2-n$ in the vertices of ${\mathbbm{P}_2^{{\mathfrak S}_n}}$, this suggests to use ${\cal S}_0$ as a parameter and minimize the left hand side of (\ref{eq:2bodySymgeneral}) with the coefficients given in (\ref{class:midkeven} $-$ \ref{class:midkodd}) for a fixed value of ${\cal S}_0$.
This leaves us with just two variables to consider in the optimization, because $n$ and ${\cal S}_0$ are fixed. We pick, for example, $b$ and $d$, so that we have
\begin{equation}
a\equiv a(b) = \frac{n+{\cal S}_0}{2}-b, \qquad c \equiv c(d) = \frac{n-{\cal S}_0}{2}-d.
 \label{eq:tmp51}
\end{equation}
Equation (\ref{eq:tmp51}) enables us to rewrite ${\cal S}_1$ and ${\cal Z}$ as
\begin{equation}
 {\cal S}_1 = n-2(b+d), \qquad {\cal Z}= {\cal S}_0 - 2(b-d).
 \label{eq:tmp52}
\end{equation}
And the constraints $a\geq 0, c\geq 0$ impose bounds on $b$ and $d$, which together with their non-negativity become
\begin{equation}
0 \leq b \leq \frac{n+{\cal S}_0}{2}, \qquad 0 \leq d \leq \frac{n-{\cal S}_0}{2}.
 \label{eq:tmp53}
\end{equation}

We shall proceed in the same spirit as in Theorem \ref{thm:classDickeLowk}, and we start by focusing
on the cases for even $n$ and odd $n$, which have a similar proof:
 \begin{itemize}
  \item Even $n$.\\
  We start by considering the Bell inequality (\ref{class:midkeven}) as a function of $b$ and $d$, with $n$ and ${\cal S}_0$ treated as parameters. We denote it $I_{n,{\cal S}_0}(b,d)$:\\
  \begin{itemize}
   \item Case $a=0$. Note that it is equivalent to $b=(n+{\cal S}_0)/2$. Then we have
   \begin{equation}
    I_{n,{\cal S}_0}\left(\frac{n+{\cal S}_0}{2},d\right)=-4d^2 + f_1(\nu,{\cal S}_0,n)d+f_0(\nu,{\cal S}_0,n),
    \label{eq:tmp54}
   \end{equation}
   where $f_i$ are expressions that do not depend on $d$. Equation (\ref{eq:tmp54}) is quadratic in $d$, and it fulfills that its second derivative is negative for all $d$, so Eq. (\ref{eq:tmp54}) has only one maximum. Therefore, its minimal value is attained either at $d=0$ or at $d=(n-{\cal S}_0)/2$. The first case does not lead to the optimal solution (See \cite{Science} for the case $\nu=0$, where the same argument follows), so we shall omit it here. The second case produces the following result:
   \begin{equation}
    I_{n,{\cal S}_0}\left(\frac{n+{\cal S}_0}{2},\frac{n-{\cal S}_0}{2}\right)={n \choose 2}({\cal S}_0+2\nu)({\cal S}_0+2\nu-2),
    \label{eq:tmp55}
   \end{equation}
   which is minimal for ${\cal S}_0 = 1-2\nu$. However, since we are in the case where $n$ is even, ${\cal S}_0$ must be even as well, which now is not the case. So we look for the closest even integers, which are $2-2\nu$ and $-2\nu$. It is immediate to see from (\ref{eq:tmp55}) that the value of $I$ is precisely $0$ on these points. Hence, we find the following values of $(a,b,c,d)$ that saturate $I$:
   \begin{equation}
    \left(0, \frac{n}{2}-\nu, 0, \frac{n}{2}+\nu\right), \qquad \left(0, \frac{n}{2}+1-\nu, 0, \frac{n}{2}-1+\nu\right).
    \label{eq:tmppunts01}
   \end{equation}
   \item Case $b=0$.
   In this case we need to consider
   \begin{equation}
    I_{n,{\cal S}_0}(0,d)=-4d^2+ g_1(\nu,{\cal S}_0,n)d+g_0(\nu,{\cal S}_0,n),
    \label{eq:tmp56}
   \end{equation}
   for some expressions $g_i$ that do not depend on $d$. Again, we see from (\ref{eq:tmp56}) that $I$ has only one maximum, so its minimum must be either at $d=0$ or at $d=(n-{\cal S}_0)/2$.
   If $d=0$, we have that
   \begin{equation}
    I_{n,{\cal S}_0}(0,0) = {n \choose 2}({\cal S}_0 + 2 \nu)({\cal S}_0 + 2\nu + 2),
    \label{eq:tmp57}
   \end{equation}
   which is minimal for ${\cal S}_0 = -1-2\nu$. Again, as ${\cal S}_0$ must be even, we consider two cases, ${\cal S}_0 = -2-2\nu$ and ${\cal S}_0 = -2\nu$, for which Equation (\ref{eq:tmp57}) immediately guarantees that $I=0$ on these points.
   We then obtain the following values of $(a,b,c,d)$:
   \begin{equation}
    \left(\frac{n}{2}-1-\nu, 0, \frac{n}{2}+1+\nu, 0\right), \qquad \left(\frac{n}{2}-\nu, 0, \frac{n}{2}+\nu,0\right).
    \label{eq:tmppunts23}
   \end{equation}
   If $d=(n-{\cal S}_0)/2$ then we obtain
   \begin{equation}
    I_{n,{\cal S}_0}\left(0,\frac{n-{\cal S}_0}{2}\right)=(n-1)({\cal S}_0 + 2 \nu)((n+2){\cal S}_0 +2n\nu)/2,
   \end{equation}
   which is minimal for ${\cal S}_0 = -2\nu \frac{n+1}{n+2}$, which is not an integer. The closest even integers to ${\cal S}_0$ are $-2\nu$ and $-2\nu+2$ and in this case, only $-2\nu$ leads to $I=0$, which corresponds to the following coordinates $(a,b,c,d)$:
   \begin{equation}
    \left(\frac{n}{2}-\nu, 0, 0, \frac{n}{2}+\nu\right).
    \label{eq:tmppunts4}
   \end{equation}
  \end{itemize}
  The points (\ref{eq:tmppunts01}, \ref{eq:tmppunts23} and \ref{eq:tmppunts4}) comprise all the vertices of ${\mathbbm{P}_2^{{\mathfrak S}_n}}$ in which $I$ is minimum, with corresponding value $0$, showing that (\ref{class:midkeven}) is a valid Bell inequality, as there are no more cases.
   %\mentalnote{These inequalities are tight, as one can check in Mathematica if really bored}.
  \item Odd $n$.\\
  We proceed in a similar way to the case of even $n$, so we also discriminate the cases $a=0$ and $b=0$. In this case, we denote by $\tilde I$ the inequality corresponding to (\ref{class:midkodd}).
  \begin{itemize}
   \item Case $a=0$. This is equivalent to $b=(n+{\cal S}_0)/2$. We have that
   \begin{equation}
    \tilde I_{n,{\cal S}_0}\left(\frac{n+{\cal S}_0}{2},d\right) = -4d^2+ \tilde f_1(\nu,{\cal S}_0,n)d+\tilde f_0(\nu,{\cal S}_0,n),
   \end{equation}
   as in (\ref{eq:tmp54}). So, its second derivative is negative and we must find the minimum at the boundary; \textit{i.e.,} either at $c=0$ or $d=0$. If we choose $d=0$ (See \cite{Science} for $\nu = 0$) then we obtain $\tilde I > 0$ for all $n$, so we shall omit its discussion here. If we consider $d=(n-{\cal S}_0)/2$, then
   \begin{equation}
    \tilde I_{n,{\cal S}_0}\left(\frac{n+{\cal S}_0}{2}, \frac{n-{\cal S}_0}{2}\right)={n \choose 2}({\cal S}_0 + 2\nu + 1)({\cal S}_0 + 2\nu - 1),
    \label{eq:tmp58}
   \end{equation}
   which is minimal for ${\cal S}_0 = -2\nu$. However, now $n$ is odd, which implies that ${\cal S}_0$ must be odd as well. Hence, the closest odd integer values to $-2\nu$ are $1-2\nu$ and $-1-2\nu$ which give the following coordinates $(a,b,c,d)$:
   \begin{equation}
   \left(0, \frac{n\pm 1}{2}-\nu, 0, \frac{n\mp 1}{2}+\nu\right).
    \label{eq:tmppunts56}
   \end{equation}
   \item Case $b=0$.\\
   Here we consider
   \begin{equation}
    \tilde I_{n,{\cal S}_0}\left(0,d\right) = -4d^2+ \tilde g_1(\nu,{\cal S}_0,n)d+\tilde g_0(\nu,{\cal S}_0,n),
   \end{equation}
   and we inspect the boundary of the domain of $d$.\\
   If $d=0$, we obtain
   \begin{equation}
    \tilde I_{n,{\cal S}_0}\left(0,0\right) = {n \choose 2}({\cal S}_0 + 2\nu +1)({\cal S}_0 + 2\nu + 3),
   \end{equation}
   an expression which is minimal for ${\cal S}_0 = -2(1+\nu)$, which is an even number. As it has to be odd, the candidates to be considered are $-3-2\nu$ and $-1-2\nu$, both of which produce a value of $\tilde I =0$. The corresponding coordinates are
   \begin{equation}
    \left(\frac{n- 1}{2}-\nu, 0, \frac{n+ 1}{2}+\nu,0\right), \qquad \left(\frac{n- 3}{2}-\nu, 0, \frac{n+ 3}{2}+\nu,0\right).
    \label{eq:tmppunts78}
   \end{equation}
   Finally, if $d=(n-{\cal S}_0)/2$, we have
   \begin{equation}
    \tilde I_{n,{\cal S}_0}\left(0,\frac{n-{\cal S}_0}{2}\right) = (n-1)({\cal S}_0 + 2\nu + 1)(2{\cal S}_0 + n({\cal S}_0 + 2\nu + 1))/2,
   \end{equation}
   which is minimal for ${\cal S}_0 = -(1+2\nu) \frac{n+1}{n+2}$, a value which is not integer. The closest odd integers in this case are $-(1+2\nu)$ and $1-2\nu$, and the smallest value of $\tilde I$ is given by the first one, for which we obtain $\tilde I = 0$ and the point
   \begin{equation}
    \left(\frac{n-1}{2}-\nu,0,0,\frac{n+1}{2}+\nu\right).
    \label{eq:tmppunts9}
   \end{equation}
  \end{itemize}
  Since the points (\ref{eq:tmppunts56}, \ref{eq:tmppunts78} and \ref{eq:tmppunts9}) comprise all the vertices of ${\mathbbm{P}_2^{{\mathfrak S}_n}}$ in which $\tilde I$ is minimum and it is $0$, we deduce that (\ref{class:midkeven}) is a valid Bell inequality, as we have checked all the cases.
 \end{itemize}
\end{proof}

\section{Blocks}\label{AppB}
In this section we discuss what is the form of $S_k$ and $S_{kl}$ in the block-decomposition given by the basis (\ref{eq:basisblocks}), where $k,l\in\{x,y,z\}$ and
\begin{equation}
\label{eq:def-sk,skl}
 S_{k} = \sum_{i=0}^{n-1} \sigma_k^{(i)}, \qquad S_{kl} = \sum_{\stackrel{i,j=0}{i\neq j}}^{n-1}\sigma_k^{(i)}\otimes \sigma_l^{(j)},
\end{equation}
where $\sigma_x, \sigma_y$ and $\sigma_z$ are the Pauli matrices.
\begin{thm}
\label{thm:bigcalculation}
 Let us define $m=n-2J$ and let $\ket{\xi_{k,J}}=\ket{D^k_{2J}}\otimes \ket{\psi^-}^{\otimes m/2}$, as in (\ref{eq:basisblocks}). Then, the $J$-th block $\rho_J$ of $S_k$ or $S_{kl}$, where $k,l\in \{x,y,z\}$, in the form (\ref{eq:blocksstate}) is given by
 \begin{equation}
 \begin{array}{ccl}
  \bra{\xi_{k,J}}S_x\ket{\xi_{l,J}} &=& \sqrt{(l+1)(2J-l)}\delta(k-l-1) + \sqrt{(k+1)(2J-k)}\delta(l-k-1)\\
  \bra{\xi_{k,J}}S_y\ket{\xi_{l,J}} &=& {\mathbbm{i}}\cdot \mathrm{sgn}(k-l)\bra{\xi_{k,J}}S_x\ket{\xi_{l,J}}\\
  \bra{\xi_{k,J}}S_z\ket{\xi_{l,J}} &=& (2J-2k)\delta(k-l)\\
  \bra{\xi_{k,J}}S_{xx}\ket{\xi_{l,J}} &=& [2k(2J-k)-m]\delta(k-l)+\sqrt{(l+1)(l+2)(2J-l)(2J-l-1)}\delta(k-l-2)\\
  &&+\sqrt{(k+1)(k+2)(2J-k)(2J-k-1)}\delta(l-k-2)\\
  \bra{\xi_{k,j}}S_{xy}\ket{\xi_{l,J}}&=& {\mathbbm{i}}\cdot \mathrm{sgn}(k-l)\sqrt{(l+1)(l+2)(2J-l)(2J-l-1)}\delta(k-l-2)\\
  &&+{\mathbbm{i}}\cdot \mathrm{sgn}(k-l)\sqrt{(k+1)(k+2)(2J-k)(2J-k-1)}\delta(l-k-2)\\
  \bra{\xi_{k,J}}S_{xz}\ket{\xi_{l,J}}&=&(2J-1-2l)\sqrt{(2J-l)(l+1)}\delta(k-l-1) + (2J-1-2k)\sqrt{(2J-k)(k+1)}\delta(l-k-1)\\
  \bra{\xi_{k,J}}S_{yy}\ket{\xi_{l,J}}&=&[2k(2J-k)-m]\delta(k-l) - \sqrt{(l+1)(l+2)(2J-l)(2J-l-1)}\delta(k-l-2)\\
  && - \sqrt{(k+1)(k+2)(2J-k)(2J-k-1)}\delta(l-k-2)\\
  \bra{\xi_{k,J}}S_{yz}\ket{\xi_{l,J}}&=&{\mathbbm{i}}\cdot \mathrm{sgn}(k-l)(2J-1-2l)\sqrt{(2J-l)(l+1)}\delta(k-l-1)\\
  &&+{\mathbbm{i}}\cdot \mathrm{sgn}(k-l)(2J-1-2k)\sqrt{(2J-k)(k+1)}\delta(l-k-1)\\
  \bra{\xi_{k,J}}S_{zz}\ket{\xi_{l,J}}&=&[(2J-2k)^2-2J-m]\delta(k-l)
 \end{array}
 \end{equation}
 where $0\leq k,l\leq m$, $\delta$ is the Kronecker delta function and $\mathrm{sgn}$ is the sign function.
 \begin{equation}
 \delta(a)=\left\{
 \begin{array}{ccc}
  1&\mbox{if}&a=0\\
  0&\mbox{else}.&
 \end{array}
 \right.
  \label{eq:def-KroneckerDelta}
 \end{equation}
\end{thm}
\begin{proof}
We shall start by obtaining the expression of (\ref{eq:def-sk,skl}) in the computational basis. For this purpose, it will be convenient to introduce a bit of notation to keep the calculations simple. We will denote the elements of a $2^n\times 2^n$ matrix $A$ as $(A)^{\mathbf{i}}_{\mathbf{j}}$ ($\mathbf{i}$-th row, $\mathbf{j}$-th column, $0 \leq \mathbf{i}, \mathbf{j} < 2^n$), \textit{i.e.,} $\bra{\mathbf{i}}A\ket{\mathbf{j}}$. We shall also use the binary representation of $\mathbf{i}=i_0i_1\ldots i_{n-1}$ with $i_k \in \{0,1\}$ and the same for $\mathbf{j}$, and we shall refer to $\mathbf{i}$ and $\mathbf{j}$ as words. We also define the weight of $\mathbf{i}$, denoted $|\mathbf{i}|$, as the number of $1$'s in the binary representation of $\mathbf{i}$: $|\mathbf{i}|:=|\{i_k: i_k=1\}|$. We shall also use the bit-wise XOR function (equivalently, bitwise addition modulo $2$) $\mathbf{i}\oplus \mathbf{j}:=(i_0\oplus j_0)(i_1\oplus j_1)\cdots (i_{n-1}\oplus j_{n-1})$. We will say that two words $\mathbf{i}$
and $\mathbf{j}$ are at Hamming distance $k$ if they differ exactly in $k$ bits; equivalently, if $|\mathbf{i}\oplus \mathbf{j}|=k$.
\begin{itemize}
 \item $(S_x)^{\mathbf{i}}_{\mathbf{j}} = \delta(|\mathbf{i}\oplus \mathbf{j}|-1)$.\\
 This expression stems from the fact that $(\sigma_x^{(v)})^{\mathbf{i}}_{\mathbf{j}}$ has nonzero entries (actually, their value is $1$) on those elements whose indices fulfill $i_a\oplus j_a = \delta(a-v)$ for all $0 \leq a < n$. In other words, $\mathbf{i}$ and $\mathbf{j}$ only differ in the bit corresponding to the party in which $\sigma_x$ is applied. This bit difference is a consequence of the form of $\sigma_x = \ket{0}\bra{1} + \ket{1}\bra{0}$. Summing over all parties $v$ leads to the result.
 \item $(S_y)^{\mathbf{i}}_{\mathbf{j}} = {\mathbbm{i}}\cdot \mathrm{sgn}(\mathbf{i} - \mathbf{j})\delta(|\mathbf{i}\oplus \mathbf{j}|-1)$.\\
 In this case we apply the same reasoning as in $S_x$, but because of the form of $\sigma_y = -\mathbbm{i} \ket{0}\bra{1} + \mathbbm{i}\ket{1}\bra{0}$, the elements in the upper half are multiplied by $-\mathbbm{i}$ whereas the elements in the lower half are multiplied by $\mathbbm{i}$.
 \item $(S_z)^{\mathbf{i}}_{\mathbf{j}} = (n-2|\mathbf{i}|)\delta(\mathbf{i}-\mathbf{j})$.\\
 This easily follows from the fact that $\sigma_z$ is diagonal and, when we sum $\sigma_z^{(v)}$ over all the parties $v$, those indices $\mathbf{i}$ with a $0$ in the $v$-th bit (\textit{i.e.,} $i_v=0$) get a $+1$ contribution, whereas the bits with $i_v=1$ get a $-1$ contribution.
 \item $(S_{xx})^{\mathbf{i}}_{\mathbf{j}} = 2\delta(|\mathbf{i}\oplus \mathbf{j}|-2)$.\\
 Observe first that $(\sigma_x^{(v)} \sigma_x^{(w)})^{\mathbf{i}}_{\mathbf{j}}=\delta((\mathbf{i}\oplus \mathbf{j})-(2^v+2^w))$. This is precisely the condition that bits $v$ and $w$ have to be different. Summing over all (ordered) pairs of parties $v$ and $w$, the result follows.
 \item $(S_{xy})^{\mathbf{i}}_{\mathbf{j}} = {\mathbbm{i}}\cdot \mathrm{sgn}(\mathbf{i} - \mathbf{j})[1-(-1)^\frac{|\mathbf{i}|-|\mathbf{j}|}{2}]\delta(|\mathbf{i}\oplus \mathbf{j}|-2)$.\\
 For this case, we shall compute first the value of $(\sigma_x^{(v)}\sigma_y^{(w)})^{\mathbf{i}}_{\mathbf{j}}$. Observe that $(\sigma_x^{(v)})^{\mathbf{i}}_{\mathbf{j}}= \delta((\mathbf{i}\oplus \mathbf{j})-2^v)$. and $(\sigma_y^{(w)})^{\mathbf{i}}_{\mathbf{j}}= -\mathbbm{i}\mathrm{sgn}(\mathbf{i}-\mathbf{j})\delta((\mathbf{i}\oplus \mathbf{j})-2^w)$. Then,
 \begin{equation}
  (\sigma_x^{(v)}\sigma_y^{(w)})^{\mathbf{i}}_{\mathbf{j}}=\sum_{\mathbf{k}} (\sigma_x^{(v)})^{\mathbf{i}}_{\mathbf{k}}(\sigma_y^{(w)})^{\mathbf{k}}_{\mathbf{j}}=-\mathbbm{i}\sum_{\mathbf{k}}\delta((\mathbf{i}\oplus \mathbf{k})-2^v)\delta((\mathbf{k}\oplus \mathbf{j})-2^w)\mathrm{sgn}(\mathbf{k}-\mathbf{j}).
 \end{equation}
Consider the factor $\delta((\mathbf{k}\oplus \mathbf{j})-2^w)\mathrm{sgn}(\mathbf{k}-\mathbf{j})$, since $\mathbf{k}$ and $\mathbf{j}$ only differ in the $w$-th bit, one has that the sign of $\mathbf{k}-\mathbf{j}$ is positive whenever $j_w=0$ ($k_w=1$) and negative whenever $j_w=1$ ($k_w=0$). Then, one can rewrite $\delta((\mathbf{k}\oplus \mathbf{j})-2^w)\mathrm{sgn}(\mathbf{k}-\mathbf{j})= \delta((\mathbf{k}\oplus \mathbf{j})-2^w)[\delta(j_w)-\delta(j_w-1)]=\delta((\mathbf{k}\oplus \mathbf{j})-2^w)(-1)^{j_w}$.
  \begin{equation}
  (\sigma_x^{(v)}\sigma_y^{(w)})^{\mathbf{i}}_{\mathbf{j}}=-\mathbbm{i}(-1)^{j_w}\sum_{\mathbf{k}}\delta((\mathbf{i}\oplus \mathbf{k})-2^v)\delta((\mathbf{k}\oplus \mathbf{j})-2^w)=-\mathbbm{i}(-1)^{j_w}(\sigma_x^{(v)} \sigma_x^{(w)})^{\mathbf{i}}_{\mathbf{j}}.
 \end{equation}
 Hence,
 \begin{equation}
 \label{eq:tmp30}
  (\sigma_x^{(v)}\sigma_y^{(w)})^{\mathbf{i}}_{\mathbf{j}}=-\mathbbm{i}(-1)^{j_w}\delta((\mathbf{i}\oplus \mathbf{j})-(2^v+2^w)).
 \end{equation}
Now we sum (\ref{eq:tmp30}) over all ordered pairs of particles $(v,w)$. A convenient way to perform this sum is to join term $(v,w)$ with $(w,v)$:
 \begin{equation}
  \label{eq:tmp31}
  (S_{xy})^{\mathbf{i}}_{\mathbf{j}} =\sum_{v\neq w} (\sigma_x^{(v)}\sigma_y^{(w)})^{\mathbf{i}}_{\mathbf{j}} = -\mathbbm{i}\sum_{0\leq v < w < n}\left[(-1)^{j_w}+(-1)^{j_v}\right]\delta((\mathbf{i}\oplus \mathbf{j})-(2^v+2^w)).
 \end{equation}
 The term $(-1)^{j_w} + (-1)^{j_v}$ can take $3$ different values. It will be $0$ whenever $j_v\neq j_w$, it will be $2$ whenever $j_v=j_w=0$ and it will be $-2$ whenever $j_v=j_w=1$. Observe that only terms for which $i_v\neq j_v$ and $i_w \neq j_w$ are counted in the sum (\ref{eq:tmp31}). Hence, one can substitute (\ref{eq:tmp31}) by the following expression:
 \begin{equation}
  \label{eq:tmp32}
  (S_{xy})^{\mathbf{i}}_{\mathbf{j}} = -\mathbbm{i}\sum_{0\leq v < w < n}\left[1 - (-1)^{\frac{|\mathbf{i}|-|\mathbf{j}|}{2}}\right]\mathrm{sgn}(\mathbf{i}-\mathbf{j})\delta((\mathbf{i}\oplus \mathbf{j})-(2^v+2^w)) = -\frac{\mathbbm{i}}{2}\left[1 - (-1)^{\frac{|\mathbf{i}|-|\mathbf{j}|}{2}}\right]\mathrm{sgn}(\mathbf{i}-\mathbf{j})(S_{xx})^{\mathbf{i}}_{\mathbf{j}}.
 \end{equation}
 Following the same argument as for $(S_{xx})^{\mathbf{i}}_{\mathbf{j}}$, the result follows.
 \item $(S_{xz})^{\mathbf{i}}_{\mathbf{j}} = (n-|\mathbf{i}|-|\mathbf{j}|)\delta(|\mathbf{i}\oplus \mathbf{j}|-1)$.\\
 In this case, we use the fact that $S_{xz}=S_{x}S_{z}-\sum_{v=0}^{n-1}(\sigma_x\sigma_z)^{(v)}$. Note that
 \begin{equation}
  \label{eq:tmp33}
  \sigma_x \sigma_z = \left(
  \begin{array}{cc}
   0&-1\\1&0
  \end{array}
\right),
 \end{equation}
which means that $\sum_{v=0}^{n-1}(\sigma_x\sigma_z)^{(v)}$ behaves essentially like $S_x$ but with a $-1$ instead of a $1$ on the upper diagonals. Hence, we can write
\begin{equation}
 \label{eq:tmp34}
 \left(\sum_{v=0}^{n-1}(\sigma_x \sigma_z)^{v}\right)^{\mathbf{i}}_{\mathbf{j}}=\mathrm{sgn}(\mathbf{i}-\mathbf{j})\delta(|\mathbf{i}\oplus \mathbf{j}|-1)=(|\mathbf{i}|-|\mathbf{j}|)\delta(|\mathbf{i}\oplus \mathbf{j}|-1).
\end{equation}
The last equality stems from the fact that, given a $\mathbf{j}$, a number $\mathbf{i}$ greater than $\mathbf{j}$, but at Hamming distance $1$ from this $\mathbf{j}$ can only be constructed by changing one of the $0$ bits of $\mathbf{j}$ to $1$, hence $|\mathbf{i}|=|\mathbf{j}|+1$. The opposite happens when we want to get a lower number: we turn one of the $1$ bits of $\mathbf{j}$ to $0$, hence lowering its weight by $1$.
At this point we just need to calculate
\begin{equation}
\begin{array}{ccl}
 (S_{xz})^{\mathbf{i}}_{\mathbf{j}} &=& \displaystyle\sum_{\mathbf{k}}\left[(S_{x})^{\mathbf{i}}_{\mathbf{k}} (S_{z})^{\mathbf{k}}_{\mathbf{j}}\right] - (|\mathbf{i}|-|\mathbf{j}|)\delta(|\mathbf{i}\oplus \mathbf{j}|-1)\\
 &=&\displaystyle\sum_{\mathbf{k}}\left[\delta(|\mathbf{i}\oplus\mathbf{k}|-1)(n-2|\mathbf{k}|)\delta(\mathbf{k}-\mathbf{j})\right] - (|\mathbf{i}|-|\mathbf{j}|)\delta(|\mathbf{i}\oplus \mathbf{j}|-1)\\
 &=&(n-2|\mathbf{j}|)\delta(|\mathbf{i}\oplus\mathbf{j}|-1)-(|\mathbf{i}|-|\mathbf{j}|)\delta(|\mathbf{i}\oplus\mathbf{j}|-1)=(n-|\mathbf{i}|-|\mathbf{j}|)\delta(|\mathbf{i}\oplus\mathbf{j}|-1).
\end{array}
\end{equation}
 \item $(S_{yy})^{\mathbf{i}}_{\mathbf{j}} = 2(-1)^\frac{|\mathbf{i}|-|\mathbf{j}|}{2}\delta(|\mathbf{i}\oplus \mathbf{j}|-2)$.\\
 In this case, we proceed in a similar way as in $(S_{xy})^{\mathbf{i}}_{\mathbf{j}}$:
 \begin{equation}
 \begin{array}{ccl}
  (\sigma_y^{(v)}\sigma_y^{(w)})^{\mathbf{i}}_{\mathbf{j}}&=&\displaystyle\sum_{\mathbf{k}}(S_{y})^{\mathbf{i}}_{\mathbf{k}} (S_{y})^{\mathbf{k}}_{\mathbf{j}} = -\displaystyle\sum_{\mathbf{k}}\delta(\mathbf{i}\oplus\mathbf{k} -2^v)\delta(\mathbf{k}\oplus \mathbf{j}-2^w)\mathrm{sgn}(\mathbf{i}-\mathbf{k})\mathrm{sgn}(\mathbf{k}-\mathbf{j})\\
  &=&\displaystyle\sum_{\mathbf{k}}\delta(\mathbf{i}\oplus\mathbf{k} -2^v)\delta(\mathbf{k}\oplus \mathbf{j}-2^w)(-1)^{i_v+j_w}=(-1)^{i_v+j_w}(\sigma_x^{(v)}\sigma_x^{(w)})^{\mathbf{i}}_{\mathbf{j}}.\\
 \end{array}
 \label{eq:tmp35}
 \end{equation}
Summing (\ref{eq:tmp35}) over all parties, we obtain the result, since $(-1)^{i_v+j_w}$ is equal to $1$ whenever $i_v=j_w$ and to $-1$ whenever $i_v\neq j_w$. The Kronecker deltas in (\ref{eq:tmp35}) enable us to substitute $(-1)^{i_v+j_w}$ by $(-1)^{\frac{|\mathbf{i}|-|\mathbf{j}|}{2}}$ in this case.
 \item $(S_{yz})^{\mathbf{i}}_{\mathbf{j}} = {\mathbbm{i}}\cdot \mathrm{sgn}(\mathbf{i} - \mathbf{j}) (S_{xz})^{\mathbf{i}}_{\mathbf{j}}$.\\
 In this case, it suffices to change the signs the upper half of the matrix and multiply by $\mathbbm{i}$, as we did to go to $(S_{y})^{\mathbf{i}}_{\mathbf{j}}$ from $(S_{x})^{\mathbf{i}}_{\mathbf{j}}$.
 \item $(S_{zz})^{\mathbf{i}}_{\mathbf{j}} = [n(n-1)-4|\mathbf{i}|(n-|\mathbf{i}|)]\delta(\mathbf{i}-\mathbf{j})$.\\
 For this last case, we apply that $S_{zz}=S_{z}S_{z}-\sum_{v=0}^{n-1}(\sigma_z\sigma_z)^{(v)}$.
 \begin{equation}
  (S_{zz})^{\mathbf{i}}_{\mathbf{j}}=\sum_{\mathbf{k}}(S_{z})^{\mathbf{i}}_{\mathbf{k}}(S_{z})^{\mathbf{k}}_{\mathbf{j}}-\left[\sum_{v}(\sigma_z\sigma_z)^{(v)}\right]^{\mathbf{i}}_{\mathbf{j}} = \sum_{\mathbf{k}}\left[(n-2|\mathbf{k}|)^2\delta(\mathbf{i}-\mathbf{k})\delta(\mathbf{k}-\mathbf{j})\right]-n\delta(\mathbf{i}-\mathbf{j}).
 \end{equation}
 Because $\delta(\mathbf{i}-\mathbf{k})\delta(\mathbf{k}-\mathbf{j})=\delta(\mathbf{i}-\mathbf{j})\delta(\mathbf{k}-\mathbf{j})$, we obtain the result:
 \begin{equation}
  (S_{zz})^{\mathbf{i}}_{\mathbf{j}}=\delta(\mathbf{i}-\mathbf{j})\left[(n-2|\mathbf{i}|)^2-n\right].
 \end{equation}
\end{itemize}

Let us now express $\ket{D^k_{2J}}\ket{\psi^-}^{\otimes m/2}$ in the computational basis. It will be convenient to denote $\Delta(\mathbf{i'})=\prod_{a=0}^{m/2-1}\delta(i'_{2a}-i'_{2a+1}-1)$. Let us also denote by $\mathbf{\aleph}=10101010...10$ and by $\mathbf{i'}\cdot \mathbf{\aleph}$ the scalar product $\sum_{a=0}^{m/2-1}i_{2a+1}'$. Then, we have that:
\begin{equation}
 \ket{D^k_{2J}}=\sum_{\mathbf{i}=0}^{2^{2J}-1} {2J\choose k}^{-1/2} \delta(|\mathbf{i}|-k)\ket{\mathbf{i}}, \qquad \ket{\psi^{-}}^{\otimes m/2}=\sum_{\mathbf{i'}=0}^{2^m-1} 2^{-m/4}(-1)^{\mathbf{i'}\cdot \mathbf{\aleph}}\Delta(\mathbf{i'})\ket{\mathbf{i'}}.
\end{equation}

The matrix for the change of basis, denoted $p$, is then
\begin{equation}
 p = \sum_{\mathbf{i}=0}^{2^{2J}-1}\sum_{\mathbf{i'}=0}^{2^m-1}\sum_{j=0}^{2J}
 \frac{(-1)^{\mathbf{i'}\cdot\mathbf{\aleph}}}{\sqrt{{2J \choose j}2^{m/2}}}\delta(|\mathbf{i}|-j)\Delta(\mathbf{i'})\ket{\mathbf{i}}\ket{\mathbf{i'}}\bra{D^k_{2J}}\bra{\psi^-}^{\otimes m/2}.
\end{equation}
Observe that $p^Tp={\mathbbm 1}_{2J+1}$, whereas $pp^T$ is the projector onto the space ${\cal H}_J$ in (\ref{eq:Schur-Weyl}). We shall denote $\ket{\underline{\mathbf i}}=\ket{\mathbf i}\ket{\mathbf{i'}}$ and $\underline{\mathbf{i}}=2^m\mathbf{i}+\mathbf{i'}$ in order to make the expressions more compact.\\
\begin{itemize}
 \item $(p^TS_xp)^k_l$.
  \begin{equation}
  \begin{array}{ccl}
   (p^TS_xp)^k_l&=&\displaystyle\sum_{\underline{\mathbf{i}}, \underline{\mathbf{j}}}p_k^{\underline{\mathbf{i}}} (S_x)^{\underline{\mathbf{i}}}_{\underline{\mathbf{j}}} p^{\underline{\mathbf{j}}}_l\\
   &=&\displaystyle\sum_{\underline{\mathbf{i}}, \underline{\mathbf{j}}}\frac{(-1)^{\mathbf{i'}\cdot \mathbf{\aleph}}(-1)^{\mathbf{j'}\cdot \mathbf{\aleph}}}{\sqrt{2^m{2J \choose k}{2J \choose l}}}\Delta(\mathbf{i'})\Delta(\mathbf{j'})\delta(|\underline{\mathbf{i}}\oplus\underline{\mathbf{j}}|-1)\delta(|\mathbf{i}|-k)\delta(|\mathbf{j}|-l).
  \end{array}\nonumber
 \end{equation}
 Using the fact that $|\underline{\mathbf{i}}\oplus\underline{\mathbf{j}}|=|\mathbf{i}\oplus \mathbf{j}|+|\mathbf{i'}\oplus \mathbf{j'}|$, it is convenient to distinguish the two possible cases for which $\delta(|\underline{\mathbf{i}}\oplus\underline{\mathbf{j}}|-1)$ is nonzero: $|\mathbf{i}\oplus \mathbf{j}|=0$ and $|\mathbf{i}\oplus \mathbf{j}|=1$. Equivalently, $\mathbf{i}=\mathbf{j}$ and $\mathbf{i'}=\mathbf{j'}$.
 We can write $\delta(|\underline{\mathbf{i}}\oplus\underline{\mathbf{j}}|-1) = \delta(\mathbf{i}-\mathbf{j})\delta(|\mathbf{i'}\oplus \mathbf{j'}|-1) + \delta(|\mathbf{i}\oplus \mathbf{j}|-1)\delta(\mathbf{i'}-\mathbf{j'})$.
 The first term does not lead to any contribution, because if $\mathbf{i'}$ and $\mathbf{j'}$ differ exactly in one bit, then either $|\mathbf{i'}|$ or $|\mathbf{j'}|$ cannot be $m/2$, and $\Delta(\mathbf{i'})\Delta(\mathbf{j'})$ neutralizes any index with weight other than $m/2$. Hence, we are left with the last summand, which is
 \begin{equation}
  \begin{array}{ccl}
   (p^TS_xp)^k_l&=&\displaystyle\sum_{\underline{\mathbf{i}}, \underline{\mathbf{j}}}p_k^{\underline{\mathbf{i}}} (S_x)^{\underline{\mathbf{i}}}_{\underline{\mathbf{j}}} p^{\underline{\mathbf{j}}}_l\\
   &=&\displaystyle\sum_{\underline{\mathbf{i}}, \underline{\mathbf{j}}}\frac{(-1)^{\mathbf{i'}\cdot\mathbf{\aleph}}(-1)^{\mathbf{j'}\cdot\mathbf{\aleph}}}{\sqrt{2^m{2J \choose k}{2J \choose l}}}\Delta(\mathbf{i'})\Delta(\mathbf{j'})\delta(|\mathbf{i}\oplus\mathbf{j}|-1)\delta(\mathbf{i'}-\mathbf{j'})\delta(|\mathbf{i}|-k)\delta(|\mathbf{j}|-l)\\
   &=&\displaystyle\sum_{\mathbf{i}, \mathbf{j}, \mathbf{i'}}\frac{(-1)^{2\mathbf{i'}\cdot \mathbf{\aleph}}}{\sqrt{2^m{2J \choose k}{2J \choose l}}}\Delta(\mathbf{i'})\delta(|\mathbf{i}\oplus\mathbf{j}|-1)\delta(|\mathbf{i}|-k)\delta(|\mathbf{j}|-l)\\
   &=&\displaystyle\sum_{\mathbf{i}, \mathbf{j}}\frac{1}{\sqrt{2^m{2J \choose k}{2J \choose l}}}\delta(|\mathbf{i}\oplus\mathbf{j}|-1)\delta(|\mathbf{i}|-k)\delta(|\mathbf{j}|-l)\sum_{\mathbf{i'}}\Delta(\mathbf{i'})\\
   &=&\displaystyle\sum_{\mathbf{i}, \mathbf{j}}\frac{1}{\sqrt{{2J \choose k}{2J \choose l}}}\delta(|\mathbf{i}\oplus\mathbf{j}|-1)\delta(|\mathbf{i}|-k)\delta(|\mathbf{j}|-l).
  \end{array}\nonumber
 \end{equation}
 Now we use the following fact:
 \begin{equation}
  \sum_{\mathbf{j}=0}^{2^{2J}-1}\delta(|\mathbf{j}|-l)\delta(|\mathbf{i}\oplus\mathbf{j}|-1)=|\mathbf{i}|\delta(|\mathbf{i}|-l-1)+(2J-|\mathbf{i}|)\delta(|\mathbf{i}|-l+1),
  \label{eq:tmp36}
 \end{equation}
 which is explained as follows: The left hand side of (\ref{eq:tmp36}) counts, among all words $\mathbf{j}$, which of them have weight $l$ and are at Hamming distance from a given $\mathbf{i}$ (they differ exactly in one bit). There are $|\mathbf{i}|$ of them which have weight $l=|\mathbf{i}|-1$ (\textit{i.e.}, we turn a $1$ into a $0$) and there are $2J-|\mathbf{i}|$ of them which have weight $l=|\mathbf{i}|+1$ (\textit{i.e.}, we turn a $0$ into a $1$).
 \begin{equation}
  \begin{array}{ccl}
   (p^TS_xp)^k_l&=&\displaystyle\sum_{\mathbf{i}}\frac{\delta(|\mathbf{i}|-k)}{\sqrt{{2J \choose k}{2J \choose l}}}(k\delta(k-l-1)+ (2J-k)\delta(k-l+1))\\
   &=&\frac{k{2J\choose k}\delta(k-l-1)+(2J-k){2J\choose k}\delta(k-l+1)}{\sqrt{{2J\choose k}{2J \choose l}}}\\
   &=&(l+1)\sqrt{{2J \choose l+1}{2J \choose l}}\delta(k-l-1)+(2J-k)\sqrt{{2J \choose k}/{2J \choose k+1}}\delta(l-k-1)\\
   &=&\sqrt{(l+1)(2J-l)}\delta(k-l-1)+\sqrt{(k+1)(2J-k)}\delta(l-k-1).
  \end{array}\nonumber
 \end{equation}

 \item For $(p^TS_yp)^k_l$, the same reasoning as $(p^TS_xp)^k_l$ follows, with the additional fact that the term $\mathrm{sgn}(\mathbf{i}-\mathbf{j})$ turns into $\mathrm{sgn}(k-l)$ because
 \begin{equation}
  \mathrm{sgn}(\underline{\mathbf{i}}-\underline{\mathbf{j}})\delta(\mathbf{i'}-\mathbf{j'})\delta(|\mathbf{i}\oplus\mathbf{j}|-1)\delta(|\mathbf{i}|-k)\delta(|\mathbf{j}|-l)=
  \mathrm{sgn}(\mathbf{i}-\mathbf{j})\delta(|\mathbf{i}\oplus\mathbf{j}|-1)\delta(|\mathbf{i}|-k)\delta(|\mathbf{j}|-l)=k-l,\nonumber
 \end{equation}
 and $(k-l)[\delta(k-l-1)+\delta(l-k-1)]=\mathrm{sgn}(k-l)$.

 \item Let us calculate $(p^TS_zp)^k_l$.
 \begin{equation}
  \begin{array}{ccl}
   (p^TS_zp)^k_l&=&\displaystyle\sum_{\underline{\mathbf{i}}, \underline{\mathbf{j}}}p_k^{\underline{\mathbf{i}}} (S_z)^{\underline{\mathbf{i}}}_{\underline{\mathbf{j}}} p^{\underline{\mathbf{j}}}_l\\
   &=&\displaystyle\sum_{\underline{\mathbf{i}}, \underline{\mathbf{j}}}\frac{(n-2|\underline{\mathbf{i}}|)(-1)^{\mathbf{i'}\mathbf{\aleph}}(-1)^{\mathbf{j'}\mathbf{\aleph}}}{\sqrt{2^m{2J \choose k}{2J \choose l}}}\Delta(\mathbf{i'})\Delta(\mathbf{j'})\delta(\underline{\mathbf{i}}-\underline{\mathbf{j}})\delta(|\mathbf{i}|-k)\delta(|\mathbf{j}|-l)\\
   &=&\displaystyle\sum_{\underline{\mathbf{i}}}\frac{(n-2|\underline{\mathbf{i}}|)(-1)^{(\mathbf{i'}+\mathbf{i'})\cdot \mathbf{\aleph}}}{\sqrt{2^m{2J \choose k}{2J \choose l}}}\Delta(\mathbf{i'})\delta(|\mathbf{i}|-k)\delta(|\mathbf{i}|-l)\\
   &=&\displaystyle\sum_{\mathbf{i}}\frac{\delta(|\mathbf{i}|-k)\delta(k-l)}{\sqrt{{2J \choose k}{2J \choose l}}}\sum_{\mathbf{i'}}\frac{(n-2|\mathbf{i}|-2|\mathbf{i'}|)}{\sqrt{2^m}}\Delta(\mathbf{i'})\\
   &=&\displaystyle\frac{{2J \choose k}\delta(k-l)}{\sqrt{{2J \choose k}{2J \choose k}}}\sum_{\mathbf{i'}}\frac{(n-2k-2m/2)}{\sqrt{2^m}}\Delta(\mathbf{i'})\\
   &=&\displaystyle\delta(k-l)\sum_{\mathbf{i'}}\frac{2J-2k}{2^{m/2}}\Delta(\mathbf{i'})=(2J-2k)\delta(k-l).\\
  \end{array}\nonumber
 \end{equation}
 We have used the fact that any function $f$ of $|\mathbf{i'}|$, when multiplied by $\Delta(\mathbf{i'})$, is a function of $m/2$, because only words with weight $m/2$ survive $\Delta(\mathbf{i'})$.\\
Hence, we obtain $(p^TS_zp)^k_l=2(J-k)\delta(k-l)$.

 \item Let us now move on to $(p^TS_{xx}p)^k_l$.
 \begin{equation}
  \begin{array}{ccl}
   (p^TS_{xx}p)^k_l&=&\displaystyle\sum_{\underline{\mathbf{i}}, \underline{\mathbf{j}}}p_k^{\underline{\mathbf{i}}} (S_{xx})^{\underline{\mathbf{i}}}_{\underline{\mathbf{j}}} p^{\underline{\mathbf{j}}}_l\\
   &=&2\displaystyle\sum_{\underline{\mathbf{i}}, \underline{\mathbf{j}}}\frac{(-1)^{\mathbf{i'}\cdot \mathbf{\aleph}}(-1)^{\mathbf{j'}\cdot \mathbf{\aleph}}}{\sqrt{2^m{2J \choose k}{2J \choose l}}}\Delta(\mathbf{i'})\Delta(\mathbf{j'})\delta(|\underline{\mathbf{i}}\oplus\underline{\mathbf{j}}|-2)\delta(|\mathbf{i}|-k)\delta(|\mathbf{j}|-l).
  \end{array}\nonumber
 \end{equation}
 Similarly to the case of $S_x$, now we use the fact that $\delta(|\underline{\mathbf{i}}\oplus\underline{\mathbf{j}}|-2) = \delta(\mathbf{i}-\mathbf{j})\delta(|\mathbf{i'}\oplus\mathbf{j'}|-2)+\delta(|\mathbf{i}\oplus\mathbf{j}|-1)\delta(|\mathbf{i'}\oplus\mathbf{j'}|-1)+\delta(|\mathbf{i}\oplus\mathbf{j}|-2)\delta(\mathbf{i'}-\mathbf{j'})$. The term in the middle is irrelevant, as the delta factors $\Delta(\mathbf{i'})\Delta(\mathbf{j'})$ will cancel it out (because of the term $\delta(|\mathbf{i'}\oplus\mathbf{j'}|-1)$, not both of them can have weight $m/2$).
 Let us consider the remaining terms:
 \begin{itemize}
  \item For the case that $\mathbf{i}=\mathbf{j}$, its contribution is
  \begin{equation}
  \begin{array}{ccl}
   &&2\displaystyle\sum_{\mathbf{i}, \mathbf{i'}, \mathbf{j'}}\frac{(-1)^{\mathbf{i'}\cdot \mathbf{\aleph}}(-1)^{\mathbf{j'}\cdot \mathbf{\aleph}}}{\sqrt{2^m{2J \choose k}{2J \choose l}}}\Delta(\mathbf{i'})\Delta(\mathbf{j'})\delta(|\mathbf{i'}\oplus\mathbf{j'}|-2)\delta(|\mathbf{i}|-k)\delta(|\mathbf{i}|-l)\\
   &=&2\delta(k-l)\displaystyle\sum_{\mathbf{i'}, \mathbf{j'}}\frac{(-1)^{\mathbf{i'}\cdot\mathbf{\aleph}}(-1)^{\mathbf{j'}\cdot\mathbf{\aleph}}}{\sqrt{2^m}}\Delta(\mathbf{i'})\Delta(\mathbf{j'})\delta(|\mathbf{i'}\oplus\mathbf{j'}|-2)\\
   &=&2\delta(k-l)\displaystyle\sum_{\mathbf{i_L'}, \mathbf{j_L'}=0}^{2^{m/2}-1}\frac{(-1)}{\sqrt{2^m}}\delta(|\mathbf{i_L'}\oplus\mathbf{j_L'}|-1)=
   -2\delta(k-l)\displaystyle\sum_{\mathbf{i_L'}, \mathbf{j_L'}=0}^{2^{m/2}-1}\frac{1}{\sqrt{2^m}}\delta(|\mathbf{i_L'}\oplus\mathbf{j_L'}|-1)\\
   &=&-2\delta(k-l)\displaystyle\sum_{\mathbf{i_L'}=0}^{2^{m/2}-1}\frac{m/2}{\sqrt{2^m}}=-m\delta(k-l),
  \end{array}\nonumber
  \end{equation}
  where we have used the encoding for the logical qubits $\ket{0_L}=\ket{01}$ and $\ket{1_L}=\ket{10}$. Observe that $\Delta(\mathbf{i'})\Delta(\mathbf{j'})\delta(|\mathbf{i'}\oplus\mathbf{j'}|-2)=\delta(|\mathbf{i_L'}\oplus\mathbf{j_L'}|-1)$, because a pair $01$ has to go to a pair $10$ in order to be considered. Observe as well that $\sum_{ \mathbf{j_L'}}\delta(|\mathbf{i_L'}\oplus\mathbf{j_L'}|-1)=m/2$, as there are $m/2$ logical bits to flip.
 \item And for the case $\mathbf{i'}=\mathbf{j'}$, its contribution is
  \begin{equation}
  \begin{array}{ccl}
   &&2\displaystyle\sum_{\mathbf{i}, \mathbf{j}, \mathbf{i'}}\frac{(-1)^{\mathbf{i'}\cdot{\mathbf{\aleph}}}(-1)^{\mathbf{i'}\cdot\mathbf{\aleph}}}{\sqrt{2^m{2J \choose k}{2J \choose l}}}\Delta(\mathbf{i'})\delta(|\mathbf{i}\oplus\mathbf{j}|-2)\delta(|\mathbf{i}|-k)\delta(|\mathbf{j}|-l)\\
   &=&2\displaystyle\sum_{\mathbf{i}, \mathbf{j}}\frac{1}{\sqrt{{2J \choose k}{2J \choose l}}}\delta(|\mathbf{i}\oplus\mathbf{j}|-2)\delta(|\mathbf{i}|-k)\delta(|\mathbf{j}|-l)\\
   &=&\frac{2}{\sqrt{{2J \choose k}{2J \choose l}}}\displaystyle\sum_{\mathbf{i}}\delta(|\mathbf{i}|-k)\left[{|\mathbf{i}|\choose 2}\delta(|\mathbf{i}|-2-l)+|\mathbf{i}|(2J-|\mathbf{i}|)\delta(|\mathbf{i}|-l)+{2J-|\mathbf{i}|\choose 2}\delta(|\mathbf{i}|+2-l)\right]\\
   &=&\frac{2{2J\choose k}}{\sqrt{{2J \choose k}{2J \choose l}}}\left[{k \choose 2}\delta(k-2-l)+k(2J-k)\delta(k-l)+{2J-k\choose 2}\delta(k+2-l)\right]\\
   &=&2k(2J-k)\delta(k-l)+2{k \choose 2}\sqrt{{2J\choose k}/{2J \choose l}}\delta(k-2-l)+2{2J-k \choose 2}\sqrt{{2J\choose k}/{2J\choose l}}\delta(k+2-l)\\
   &=&2k(2J-k)\delta(k-l)+\sqrt{(l+2)(l+1)(2J-l)(2J-l-1)}\delta(k-2-l)\\
   &&+\sqrt{(2J-k)(2J-k-1)(k+1)(k+2)}\delta(l-k-2).\\
  \end{array}\nonumber
  \end{equation}

  We have used that, given a word $\mathbf{i}$, the words $\mathbf{j}$ at Hamming distance $2$ from $\mathbf{i}$ can be classified into $3$ subsets: ${|\mathbf{i}| \choose 2}$ of them having weight $|\mathbf{j}|=|\mathbf{i}|-2$; ${|\mathbf{i}|\choose 1}{2J-|\mathbf{i}|\choose 1}$ of them having weight $|\mathbf{j}|=|\mathbf{i}|$; and ${2J-|\mathbf{i}|\choose 2}$ of them having weight $|\mathbf{j}|=|\mathbf{i}|+2$. One can check that the relation ${2J \choose 2}={|\mathbf{i}| \choose 2}+{|\mathbf{i}|\choose 1}{2J-|\mathbf{i}|\choose 1}+{2J-|\mathbf{i}|\choose 2}$ holds, as it counts all the possibilities (all the ways to pick the $2$ different bits between $\mathbf{i}$ and $\mathbf{j}$).\\
  %Note that this expression is symmetric in $k$ and $l$.
 \end{itemize}

 \item The case $(p^TS_{xy}p)^k_l$ is similar to $(p^TS_{xx}p)^k_l$, so we shall point out only their differences:
 \begin{equation}
  \begin{array}{ccl}
   (p^TS_{xy}p)^k_l&=&\displaystyle\sum_{\underline{\mathbf{i}}, \underline{\mathbf{j}}}p_k^{\underline{\mathbf{i}}} (S_{xy})^{\underline{\mathbf{i}}}_{\underline{\mathbf{j}}} p^{\underline{\mathbf{j}}}_l\\
   &=&\mathbbm{i}\displaystyle\sum_{\underline{\mathbf{i}}, \underline{\mathbf{j}}}\frac{(-1)^{\mathbf{i'}\cdot\mathbf{\aleph}}(-1)^{\mathbf{j'}\cdot \mathbf{\aleph}}}{\sqrt{2^m{2J \choose k}{2J \choose l}}}\Delta(\mathbf{i'})\Delta(\mathbf{j'})\mathrm{sgn}(\underline{\mathbf{i}}-\underline{\mathbf{j}})
   \left(1-(-1)^{\frac{|\underline{\mathbf{i}}|-|\underline{\mathbf{j}}|}{2}}\right)
   \delta(|\underline{\mathbf{i}}\oplus\underline{\mathbf{j}}|-2)\delta(|\mathbf{i}|-k)\delta(|\mathbf{j}|-l).
  \end{array}\nonumber
 \end{equation}
 Again, we divide into the cases where $\mathbf{i}=\mathbf{j}$ or $\mathbf{i'}=\mathbf{j'}$.
 \begin{itemize}
  \item Case $\mathbf{i}=\mathbf{j}$ and $|\mathbf{i'}\oplus \mathbf{j'}|=2$:\\
  The sign function is then determined by the part corresponding to the singlet: $\mathrm{sgn}(\mathbf{i'}-\mathbf{j'})$, and we have
   \begin{equation}
  \begin{array}{ccl}
   &&\mathbbm{i}\displaystyle\sum_{\mathbf{i}, \mathbf{i'}, \mathbf{j'}}\frac{(-1)^{\mathbf{i'}\cdot\mathbf{\aleph}}(-1)^{\mathbf{j'}\cdot\mathbf{\aleph}}}{\sqrt{2^m{2J \choose k}{2J \choose l}}}\Delta(\mathbf{i'})\Delta(\mathbf{j'})\mathrm{sgn}(\mathbf{i'}-\mathbf{j'})
   \left[1-(-1)^{\frac{|\mathbf{i'}|-|\mathbf{j'}|}{2}}\right]
   \delta(|\mathbf{i'}\oplus\mathbf{j'}|-2)\delta(|\mathbf{i}|-k)\delta(|\mathbf{i}|-l).
  \end{array}\nonumber
 \end{equation}
 But if $\mathbf{i'}$ and $\mathbf{j'}$ differ exactly in two bits, because of the factors $\Delta(\mathbf{i'})\Delta(\mathbf{j'})$, we only have to consider when a pair $01$ turns into a pair $10$ and viceversa. So, $|\mathbf{i'}|=|\mathbf{j'}|$ and the whole expression is zero.
 \item Case $\mathbf{i'}=\mathbf{j'}$ and $|\mathbf{i}\oplus \mathbf{j}|=2$:
 Now the sign function is determined by the part corresponding to the Dicke states: $\mathrm{sgn}(\mathbf{i}-\mathbf{j})$, so we have
  \begin{equation}
  \begin{array}{ccl}
   &&\mathbbm{i}\displaystyle\sum_{\mathbf{i}, \mathbf{j}, \mathbf{i'}}\frac{(-1)^{2\mathbf{i'}\cdot\mathbf{\aleph}}}{\sqrt{2^m{2J \choose k}{2J \choose l}}}\Delta(\mathbf{i'})\mathrm{sgn}(\mathbf{i}-\mathbf{j})
   \left[1-(-1)^{\frac{|\mathbf{i}|-|\mathbf{j}|}{2}}\right]
   \delta(|\mathbf{i}\oplus\mathbf{j}|-2)\delta(|\mathbf{i}|-k)\delta(|\mathbf{j}|-l)\\
   &=&\mathbbm{i}\displaystyle\sum_{\mathbf{i}, \mathbf{j}}\frac{1}{\sqrt{{2J \choose k}{2J \choose l}}}\mathrm{sgn}(\mathbf{i}-\mathbf{j})
   \left[1-(-1)^{\frac{|\mathbf{i}|-|\mathbf{j}|}{2}}\right]
   \delta(|\mathbf{i}\oplus\mathbf{j}|-2)\delta(|\mathbf{i}|-k)\delta(|\mathbf{j}|-l)\displaystyle\sum_{\mathbf{i'}}\frac{\Delta(\mathbf{i'})}{2^{m/2}}\\
   &=&\mathbbm{i}\displaystyle\sum_{\mathbf{i}, \mathbf{j}}\frac{1}{\sqrt{{2J \choose k}{2J \choose l}}}\mathrm{sgn}(\mathbf{i}-\mathbf{j})
   \left[1-(-1)^{\frac{|\mathbf{i}|-|\mathbf{j}|}{2}}\right]
   \delta(|\mathbf{i}\oplus\mathbf{j}|-2)\delta(|\mathbf{i}|-k)\delta(|\mathbf{j}|-l).\\
  \end{array}\nonumber
 \end{equation}
 We observe now that there are only $3$ possibilities: $|\mathbf{i}|=|\mathbf{j}|\pm 2$ and $|\mathbf{i}|=|\mathbf{j}|$, but this last one clearly leads to $0$. If $|\mathbf{i}|=|\mathbf{j}|+2$ , then $\mathbf{i}>\mathbf{j}$, because we must change exactly two bits, and they must be two ones turning into two zeroes, as we have the $\delta(|\mathbf{i}\oplus\mathbf{j}|-2)$ factor. Similarly, if $|\mathbf{i}|=|\mathbf{j}|-2$, then $\mathbf{i}<\mathbf{j}$. We can simplify then $\mathrm{sgn}(\mathbf{i}-\mathbf{j})=(|\mathbf{i}|-|\mathbf{j}|)/2$. Hence, we  can write the following:
  \begin{equation}
  \begin{array}{ccl}
   &&2\mathbbm{i}\displaystyle\sum_{\mathbf{i}, \mathbf{j}}\frac{1}{\sqrt{{2J \choose k}{2J \choose l}}}\frac{k-l}{2}
   \delta(|\mathbf{i}\oplus\mathbf{j}|-2)\delta(|\mathbf{i}|-k)\delta(|\mathbf{j}|-l).\\
  \end{array}\nonumber
 \end{equation}
 Now, by following the same discussion as in $(p^TS_{xx}p)^k_l$, we arrive at the expression
 \begin{equation}
  \begin{array}{ccl}
   &&\mathbbm{i}\frac{k-l}{2}\left[2k(2J-k)\delta(k-l) + \sqrt{(l+2)(l+1)(2J-l)(2J-l-1)}\delta(k-l-2)\right.\\
   &&\hspace{1cm}+\left.\sqrt{(2J-k)(2J-k-1)(k+2)(k+1)}\delta(l-k-2)\right].\\
  \end{array}\nonumber
 \end{equation}
 The first term vanishes and the factor $(k-l)/2$ can be substituted by $\mathrm{sgn}(k-l)$ because of the Kronecker deltas, leading to the result.
 \end{itemize}

 \item Let us now turn to $(p^TS_{xz}p)^k_l$.
  \begin{equation}
  \begin{array}{ccl}
   (p^TS_{xz}p)^k_l&=&\displaystyle\sum_{\underline{\mathbf{i}}, \underline{\mathbf{j}}}p_k^{\underline{\mathbf{i}}} (S_{xz})^{\underline{\mathbf{i}}}_{\underline{\mathbf{j}}} p^{\underline{\mathbf{j}}}_l\\
   &=&\displaystyle\sum_{\underline{\mathbf{i}}, \underline{\mathbf{j}}}\frac{(-1)^{\mathbf{i'}\cdot\mathbf{\aleph}}(-1)^{\mathbf{j'}\cdot \mathbf{\aleph}}}{\sqrt{2^m{2J \choose k}{2J \choose l}}}\Delta(\mathbf{i'})\Delta(\mathbf{j'})(n-|\underline{\mathbf{i}}|-|\underline{\mathbf{j}}|)
   \delta(|\underline{\mathbf{i}}\oplus\underline{\mathbf{j}}|-1)\delta(|\mathbf{i}|-k)\delta(|\mathbf{j}|-l).
  \end{array}\nonumber
 \end{equation}
 As in $(p^TS_{x}p)^k_l$, it is convenient to distinguish $\delta(|\underline{\mathbf{i}}\oplus\underline{\mathbf{j}}|-1)$ into the cases $|\mathbf{i}|=|\mathbf{j}|\pm1, \mathbf{i'}=\mathbf{j'}$ or $\mathbf{i}=\mathbf{j}, |\mathbf{i'}|=|\mathbf{j'}|\pm1$. The second one is neutralized by the functions $\Delta(\mathbf{i'})\Delta(\mathbf{j'})$, so we only need to consider $\delta(|\mathbf{i}\oplus\mathbf{j}|-1)\delta(\mathbf{i'}-\mathbf{j'})$: Summing over $\mathbf{i'}$ we obtain
 \begin{equation}
  \begin{array}{ccl}
   (p^TS_{xz}p)^k_l&=&\displaystyle\sum_{\mathbf{i}, \mathbf{j}}\frac{1}{\sqrt{{2J \choose k}{2J \choose l}}}[n-|\mathbf{i}|-|\mathbf{j}|-2(m/2)]
   \delta(|\mathbf{i}\oplus\mathbf{j}|-1)\delta(|\mathbf{i}|-k)\delta(|\mathbf{j}|-l).
  \end{array}\nonumber
 \end{equation}
 Recall that, given $\mathbf{i}$, there are $|\mathbf{i}|$ words $\mathbf{j}$ having weight $|\mathbf{i}|-1$ and $2J-|\mathbf{i}|$ words $\mathbf{j}$ having weight $|\mathbf{i}|+1$. Thus, summing over $\mathbf{j}$ we obtain
 \begin{equation}
  \begin{array}{ccl}
   (p^TS_{xz}p)^k_l&=&\displaystyle\sum_{\mathbf{i}}\frac{1}{\sqrt{{2J \choose k}{2J \choose l}}}(2J-k-l)
   \delta(|\mathbf{i}|-k)\left[|\mathbf{i}|\delta(|\mathbf{i}|-l-1)+(2J-|\mathbf{i}|)\delta(|\mathbf{i}|-l+1)\right]\\
   &=&{2J \choose k}\frac{1}{\sqrt{{2J \choose k}{2J \choose l}}}(2J-k-l)
   \left[k\delta(k-l-1)+(2J-k)\delta(k-l+1)\right]\\
   &=&(2J-2l-1)\sqrt{(2J-l)(l+1)}\delta(k-l-1)+(2J-2k-1)\sqrt{(2J-k)(k+1)}\delta(l-k-1).
  \end{array}\nonumber
 \end{equation}

 \item Now we analyze $(p^TS_{yy}p)^k_l$.
  \begin{equation}
  \begin{array}{ccl}
   (p^TS_{yy}p)^k_l&=&\displaystyle\sum_{\underline{\mathbf{i}}, \underline{\mathbf{j}}}p_k^{\underline{\mathbf{i}}} (S_{yy})^{\underline{\mathbf{i}}}_{\underline{\mathbf{j}}} p^{\underline{\mathbf{j}}}_l\\
   &=&2\displaystyle\sum_{\underline{\mathbf{i}}, \underline{\mathbf{j}}}\frac{(-1)^{\mathbf{i'}\cdot \mathbf{\aleph}}(-1)^{\mathbf{j'}\cdot\mathbf{\aleph}}}{\sqrt{2^m{2J \choose k}{2J \choose l}}}\Delta(\mathbf{i'})\Delta(\mathbf{j'})(-1)^{\frac{|\underline{\mathbf{i}}|-|\underline{\mathbf{j}}|}{2}}
   \delta(|\underline{\mathbf{i}}\oplus\underline{\mathbf{j}}|-2)\delta(|\mathbf{i}|-k)\delta(|\mathbf{j}|-l).
  \end{array}\nonumber
 \end{equation}
 This case again looks very similar to $(p^TS_{xx}p)^k_l$. Recall that $\delta(|\underline{\mathbf{i}}\oplus\underline{\mathbf{j}}|-2) = \delta(\mathbf{i}-\mathbf{j})\delta(|\mathbf{i'}\oplus\mathbf{j'}|-2)+\delta(|\mathbf{i}\oplus\mathbf{j}|-1)\delta(|\mathbf{i'}\oplus\mathbf{j'}|-1)+\delta(|\mathbf{i}\oplus\mathbf{j}|-2)\delta(\mathbf{i'}-\mathbf{j'})$ and that the term in the middle is irrelevant, as the delta factors $\Delta(\mathbf{i'})\Delta(\mathbf{j'})$ cancel it out. Hence, we distinguish the two cases:
 \begin{itemize}
  \item Case $\mathbf{i}=\mathbf{j}$: For this case, its contribution is
  \begin{equation}
    \begin{array}{ccl}
    &&2\displaystyle\sum_{\mathbf{i}, \mathbf{i'}, \mathbf{j'}}\frac{(-1)^{\mathbf{i'}\cdot \mathbf{\aleph}}(-1)^{\mathbf{j'}\cdot \mathbf{\aleph}}}{\sqrt{2^m{2J \choose k}{2J \choose l}}}\Delta(\mathbf{i'})\Delta(\mathbf{j'})(-1)^{\frac{|\mathbf{i'}|-|\mathbf{j'}|}{2}}
    \delta(|\mathbf{i'}\oplus\mathbf{j'}|-2)\delta(|\mathbf{i}|-k)\delta(|\mathbf{i}|-l)\\
    &=&2\delta(k-l)\displaystyle\sum_{\mathbf{i'}, \mathbf{j'}}\frac{(-1)^{(\mathbf{i'}+\mathbf{j'})\cdot \mathbf{\aleph}}}{\sqrt{2^m}}\Delta(\mathbf{i'})\Delta(\mathbf{j'})(-1)^{\frac{|\mathbf{i'}|-|\mathbf{j'}|}{2}}
    \delta(|\mathbf{i'}\oplus\mathbf{j'}|-2).\\
    \end{array}\nonumber
  \end{equation}
 Observe that $(-1)^{(\mathbf{i'}+\mathbf{j'})\cdot  \mathbf{\aleph}}=(-1)^{(\mathbf{i'}\oplus\mathbf{j'})\cdot  \mathbf{\aleph}}$ and that $(-1)^{(\mathbf{i'}\oplus\mathbf{j'})\cdot \mathbf{\aleph}}\Delta(\mathbf{i'})\Delta(\mathbf{j'})\delta(|\mathbf{i'}\oplus\mathbf{j'}|-2)$ is nonzero only in those terms for which $\mathbf{i'}\oplus \mathbf{j'}=0000\cdots001100\cdots0000$, which means that $(\mathbf{i'}\oplus\mathbf{j'})\cdot \mathbf{\aleph}=1$. Also, the combination $\Delta(\mathbf{i'})\Delta(\mathbf{j'})\delta(|\mathbf{i'}\oplus\mathbf{j'}|-2)$ only allows for words with equal weight to be considered. Hence,
  \begin{equation}
    \begin{array}{ccl}
    &&2\delta(k-l)\displaystyle\sum_{\mathbf{i'}, \mathbf{j'}}\frac{(-1)}{\sqrt{2^m}}\Delta(\mathbf{i'})\Delta(\mathbf{j'})(-1)^{\frac{|\mathbf{i'}|-|\mathbf{j'}|}{2}}
    \delta(|\mathbf{i'}\oplus\mathbf{j'}|-2)\\
    &=&-2\delta(k-l)\displaystyle\sum_{\mathbf{i'}, \mathbf{j'}}\frac{1}{\sqrt{2^m}}\Delta(\mathbf{i'})\Delta(\mathbf{j'})
    \delta(|\mathbf{i'}\oplus\mathbf{j'}|-2)\\
    &=&-2\delta(k-l)\displaystyle\sum_{\mathbf{i_L'}, \mathbf{j_L'}}\frac{1}{\sqrt{2^m}}\delta(|\mathbf{i_L'}\oplus\mathbf{j_L'}|-1)\\
    &=&-2\delta(k-l)\displaystyle\sum_{\mathbf{i_L'}=0}^{2^{m/2}-1}\frac{m/2}{\sqrt{2^m}}=-m\delta(k-l).
    \end{array}\nonumber
  \end{equation}

 \item Case $\mathbf{i'}=\mathbf{j'}$: For this case, it will be convenient to consider
 \begin{equation}
  \sum_{\mathbf{j}}(-1)^{\frac{|\mathbf{i}|-|\mathbf{j}|}{2}}\delta(|\mathbf{i}\oplus\mathbf{j}|-2)\delta(|\mathbf{j}|-l).\nonumber
 \end{equation}
 There are ${|\mathbf{i}|\choose 2}$ words $\mathbf{j}$ at Hamming distance $2$ from a given $\mathbf{i}$ and with weight $|\mathbf{j}|=|\mathbf{i}|-2$, which are counted with the factor $(-1)^{(|\mathbf{i}|-{|\mathbf{j}|})/2}=-1$; there are $|\mathbf{i}|(2J-|\mathbf{i}|)$ words with weight $|\mathbf{j}|=|\mathbf{i}|$ which are counted with a factor $(-1)^{(|\mathbf{i}|-{|\mathbf{j}|})/2}=1$; and there are ${2J-|\mathbf{i}|\choose 2}$ words with weight $|\mathbf{j}|=|\mathbf{i}|+2$, hence counted with a factor $(-1)^{(|\mathbf{i}|-{|\mathbf{j}|})/2}=-1$.
 \end{itemize}
 Now, the same argument we used for $(p^TS_{xx}p)^k_l$ applies, and we get a contribution from this term equal to
 \begin{equation}
  \begin{array}{ccl}
   &&2k(2J-k)\delta(k-l)-\sqrt{(l+2)(l+1)(2J-l)(2J-l-1)}\delta(k-2-l)\\
   &&-\sqrt{(2J-k)(2J-k-1)(k+1)(k+2)}\delta(l-k-2).
  \end{array}\nonumber
 \end{equation}
Summing the two contributions, we arrive at the result.

 \item For the case $(p^TS_{yz}p)^k_l$, we use the same argument than that for $(p^TS_{xz}p)^k_l$ and, similarly to the $(p^TS_{y}p)^k_l$ case, we gain the extra factor $\mathbbm{i}\mathrm{sgn}(k-l)$.

 \item Finally, for the case $(p^TS_{zz}p)^k_l$, we have to consider
 \begin{equation}
  \begin{array}{ccl}
   (p^TS_{zz}p)^k_l&=&\displaystyle\sum_{\underline{\mathbf{i}}, \underline{\mathbf{j}}}p_k^{\underline{\mathbf{i}}} (S_{zz})^{\underline{\mathbf{i}}}_{\underline{\mathbf{j}}} p^{\underline{\mathbf{j}}}_l\\
   &=&\displaystyle\sum_{\underline{\mathbf{i}}, \underline{\mathbf{j}}}\frac{(-1)^{\mathbf{i'}\cdot \mathbf{\aleph}}(-1)^{\mathbf{j'}\cdot\mathbf{\aleph}}}{\sqrt{2^m{2J \choose k}{2J \choose l}}}\Delta(\mathbf{i'})\Delta(\mathbf{j'})
   [(n-2|\underline{\mathbf{i}}|)^2-n]
   \delta(\underline{\mathbf{i}}-\underline{\mathbf{j}})\delta(|\mathbf{i}|-k)\delta(|\mathbf{j}|-l)\\
   &=&\displaystyle\sum_{\underline{\mathbf{i}}}\frac{(-1)^{\mathbf{i'}\cdot \mathbf{\aleph}}(-1)^{\mathbf{i'}\cdot\mathbf{\aleph}}}{\sqrt{2^m{2J \choose k}{2J \choose l}}}\Delta(\mathbf{i'})\Delta(\mathbf{i'})
   [(n-2|\underline{\mathbf{i}}|)^2-n]
   \delta(|\mathbf{i}|-k)\delta(|\mathbf{i}|-l)\\
   &=&\delta(k-l)\displaystyle\sum_{\mathbf{i}}\frac{1}{\sqrt{{2J \choose k}{2J \choose l}}}
   [(n-2|\mathbf{i}|-2m/2)^2-n]
   \delta(|\mathbf{i}|-k)\sum_{\mathbf{i'}}\frac{\Delta(\mathbf{i'})}{\sqrt{2^{m}}}\\
   &=&\delta(k-l)\displaystyle\sum_{\mathbf{i}}\frac{1}{\sqrt{{2J \choose k}{2J \choose l}}}
   [(2J-2|\mathbf{i}|)^2-n]
   \delta(|\mathbf{i}|-k)\\
   &=&\delta(k-l)[(2J-2k)^2-n]=[(2J-2k)^2-2J-m]\delta(k-l).
  \end{array}\nonumber
 \end{equation}
\end{itemize}
\end{proof}

\section{Experimental errors and 2-body Bell inequalities}\label{AppC}

Here we discuss the role of experimental errors and imperfections in detection
of nonlocality with our two-body Bell inequalities by measuring the collective observables.
Our analysis is motivated by the recent experiment aiming at detection of entanglement in the Dicke states \cite{Klempt}. We analyse how errors in the quantities $S_\theta^2$ affect the ability to observe quantum violation with our Bell inequalities.

In Sec. \ref{sec:Experimental} we have shown how our permutationally invariant Bell inequalities
in the case when all parties perform the same pair of measurements can be
tested with the aid of total spin components $S_i$ $(i=x,y,z)$ and their combinations
$\textbf{m}\cdot\vec{S}$ along $\mathbf{m}$.
Here we take a bit more simplified approach and express our Bell inequalities in terms of only
$S_z$, $S_x$, $S_z^2$, $S_x^2$ and $S_{\pi/4}^2$, where $S_{\pi/4}=(S_x+S_z)/\sqrt{2}$.

To be more precise, let us recall that any permutationally invariant two-body Bell inequality
can be written as
\begin{equation}
\label{eq:eq1}
\beta_c + \alpha {\cal S}_0 + \beta {\cal S}_1 + \frac{\gamma}{2}{\cal S}_{00} + \delta {\cal S}_{01} + \frac{\varepsilon}{2}{\cal S}_{11} \geq 0,
\end{equation}
where $\mathcal{S}_k$ and ${\cal S}_{kl}$ are defined in Eqs. (\ref{eq:def1bodysym}) and (\ref{eq:def2bodysym}), and at each site we use the same pair of measurements of the form ${\cal M}_k^{(i)}=\cos (\theta_k) \sigma_z^{(i)} + \sin(\theta_k)\sigma_x^{(i)}$.
For simplicity, in what follows, we denote $\theta_0 = \varphi, \theta_1 = \theta$.
Denoting then by $\mathcal{S}_u$ and $\mathcal{S}_{uv}$ with $u,v=x,z$ the permutationally invariant
expectation values $\mathcal{S}_k$ and $\mathcal{S}_{kl}$ in which the measurements are taken to be simply
$\sigma_z$ and $\sigma_x$, we can now rewrite (\ref{eq:eq1}) as
%
%Using the results in the Supplementary Material of Ref. \cite{Science} we can re-express all the ${\cal S}_k, {\cal %S}_{kl}$ in terms of ${\cal S}_{u}, {\cal S}_{uv}$, where $k,l \in \{0,1\}$ and $u,v \in \{x,z\}$.
%The Bell inequality (\ref{eq:eq1}) then becomes
\begin{equation}
\label{eq:eq2}
\beta_c + A {\cal S}_z + A'{\cal S}_x + \frac{B}{2}{\cal S}_{zz} + D {\cal S}_{xz}+ \frac{C}{2}{\cal S}_{xx} \geq 0,
\end{equation}
where
$$
\begin{array}{lll}
A&=&\alpha \cos \varphi + \beta \cos \theta,\quad
A'=\alpha \sin \varphi + \beta \sin \varphi,\\
B&=&\gamma \cos^2\varphi + 2 \delta \cos \varphi \cos \theta + \varepsilon \cos^2 \theta,\\
C&=&\gamma \sin^2\varphi + 2 \delta \sin \varphi \sin \theta + \varepsilon \sin^2 \theta,\\
D&=&\gamma \cos \varphi \sin \varphi + \delta \cos \varphi \sin \theta + \delta \cos \theta \sin \varphi + \varepsilon \cos \theta \sin \theta.
\end{array}
$$
%and ${\cal S}_{uv}:=\sum_{i \neq j}\sigma_u^{(i)}\otimes \sigma_v^{(j)}$.

We can clearly express the quantities $\mathcal{S}_u$ and $\mathcal{S}_{uv}$
in terms of the total spin components $S_i$ as ${\cal S}_u=2\langle S_u\rangle$ and ${\cal S}_{uu}=4\langle S_u^2\rangle-n$ with $u,v=x,z$, and
%
%Now, we will use the following relations between the quantities ${\cal S}_u, {\cal S}_{uv}$ and the $S_u, S_{u}^2, %S_{\pi/4}^2$:
%
%$$
\begin{equation}
{\cal S}_{xy}=2\langle\{S_x,S_z\}\rangle=2\left(2\langle S_{\pi/4}^2\rangle-\langle S_x^2\rangle-\langle S_z^2\rangle\right).
\end{equation}
%$$
%
This allows us to rewrite (\ref{eq:eq2}) in the following form
\begin{equation}
\label{eq:eq3}
\beta_c + 2A \langle S_z\rangle + 2A'\langle S_x\rangle+\frac{B}{2}(4\langle S_z^2\rangle-n)+\frac{C}{2}(4 \langle S_x^2 \rangle - n) + 2D(2\langle S_{\pi/4}^2\rangle-\langle S_x^2\rangle-\langle S_z^2\rangle)\geq 0,
\end{equation}
which can be further  rewritten as:
\begin{equation}
\label{eq:eq4}
\beta_c + 2A\langle S_z\rangle + 2A'\langle S_x\rangle + 2(B-D)\langle S_z^2\rangle + 2(C-D)\langle S_x^2\rangle + 4D \langle S_{\pi/4}^2\rangle - n\frac{B+C}{2}\geq 0.
\end{equation}

Having all this we can now study how the experimental errors influence the violation of our Bell inequalities. Inspired by Ref. \cite{Klempt}, we consider errors of the form
\begin{equation}
\label{eq:errormodel}
\langle S_\theta^2\rangle_{\mbox{measured}}=\kappa + \eta\langle S_{\theta}^2\rangle_{\mbox{ideal}},
\end{equation}
where $0<\eta<1$ is some visibility parameter, $\kappa$ is an offset that depends on $n$, and $\theta=x,z,\pi/4$. The parameter $\kappa$ is strongly related to the detection noise (\textit{i.e.}, the atom counting uncertainty); thus it adds independently of $\eta$.
In this way, (\ref{eq:errormodel}) captures the behavior of the experimental data of Ref. \cite{Klempt}, in which $\eta \approx 0.8$ and $\kappa \approx 100$: For $\langle S_z^2\rangle$, which is ideally $0$, $\kappa$ dominates, whereas for $\langle S_x^2 \rangle$, which is $O(n^2)$, $\eta$ dominates. For $\langle S_{\pi/4}^2 \rangle$ the results of this measurement are not available, but we expect that both effects are taken into account reasonably by (\ref{eq:errormodel}).

\subsection{Half-filled Dicke state}
In the case of $\ket{D^{n/2}_n}$ with even $n$, the ideal values of the total spin components and its moments are

\begin{itemize}
\item $\langle S_x \rangle = \langle S_z \rangle = \langle S_z^2 \rangle=0$,
\item $\langle S_x^2\rangle=n(n+2)/8$,
\item $\langle S_{\pi/4}^2\rangle=n(n+2)/16$.
\end{itemize}
Recall that the Bell inequality that detects the Dicke state $\ket{D^{n/2}_n}$ is
\begin{equation}\label{appC:inequality}
 \left\lceil \frac{n+2}{2} \right \rceil {n \choose 2} + \frac{1}{2}{n \choose 2} {\cal S}_{00} + \frac{n}{2} {\cal S}_{01} - \frac{1}{2}{\cal S}_{11} \geq 0.
\end{equation}
In Fig. \ref{fig:1} one can see the effect that errors have on the violation of Ineq. (\ref{appC:inequality}).

\begin{center}
\begin{figure}
\includegraphics[width=7cm]{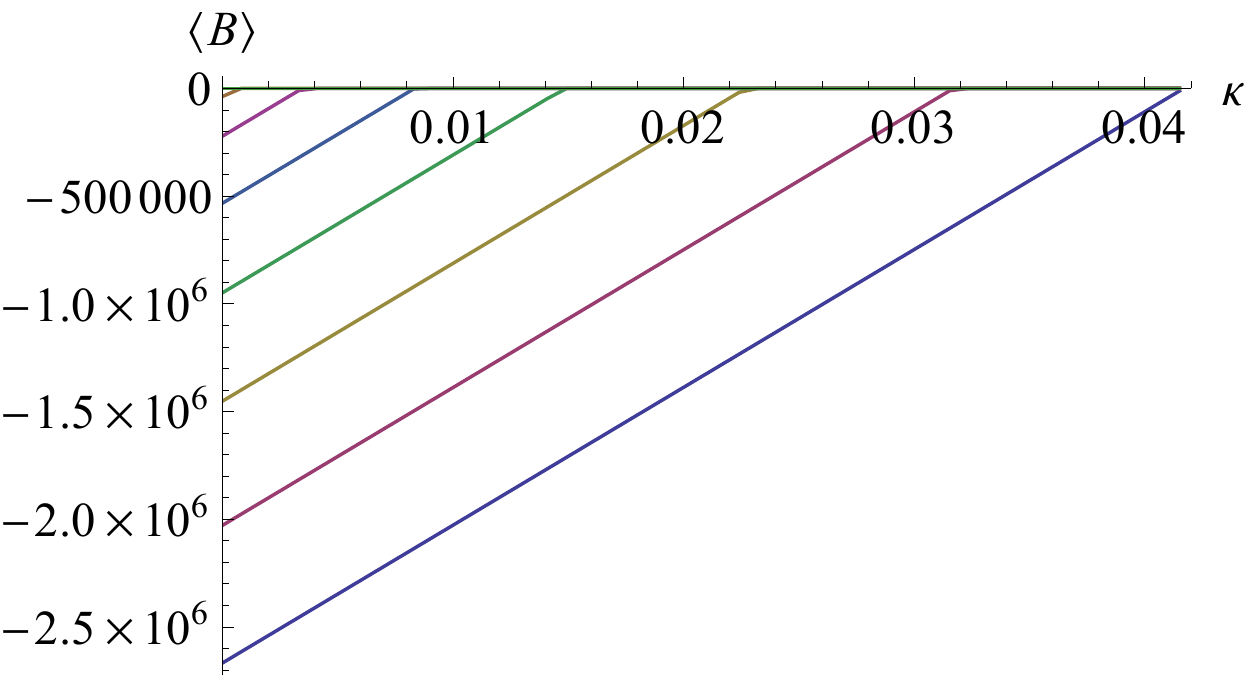}
\caption{On the vertical axis, we have the unnormalized quantum violation (the value of the left-hand side of (\ref{appC:inequality})) for $n=8000$; on the horizontal axis there is $\kappa$. The lines are for $\eta \in \{1,0.95,0.9,\ldots\}$ starting at $\eta=1$ for the blue line.
% \added{(The violation is large and makes the impression that even if there are errors it can be easily detected.)}
}\label{fig:1}
\end{figure}
\end{center}

\subsection{State for inequality (6)}
In this other case, for the Gaussian superposition of Dicke states, we have that for large $n$ the ideal values of the total spin components and its moments are
\begin{itemize}
\item $\langle S_z \rangle=1/2\sqrt{3}$,
\item $\langle S_x \rangle=n/2$,
\item $\langle S_z^2\rangle=0$,
\item $\langle S_x^2\rangle=n^2/4$,
\item $\langle S_{\pi/4}^2\rangle=n(n/2+1/\sqrt{3})/4$.
\end{itemize}
Recall that the Bell inequality that is violated by these superpositions is
\begin{equation}\label{appC:inequality2}
 2n-2{\cal S}_0+ \frac{1}{2}{\cal S}_{00} - {\cal S}_{01} + \frac{1}{2}{\cal S}_{11} \geq 0
\end{equation}
In Fig. \ref{fig:2} one can see the effect that errors have on these quantities. This case looks much more promising, as the error in $\langle S_z^2\rangle$ it can support may be almost an order of magnitude larger than what was achieved in the experiment \cite{Klempt}.
\begin{center}
\begin{figure}
\includegraphics[width=7cm]{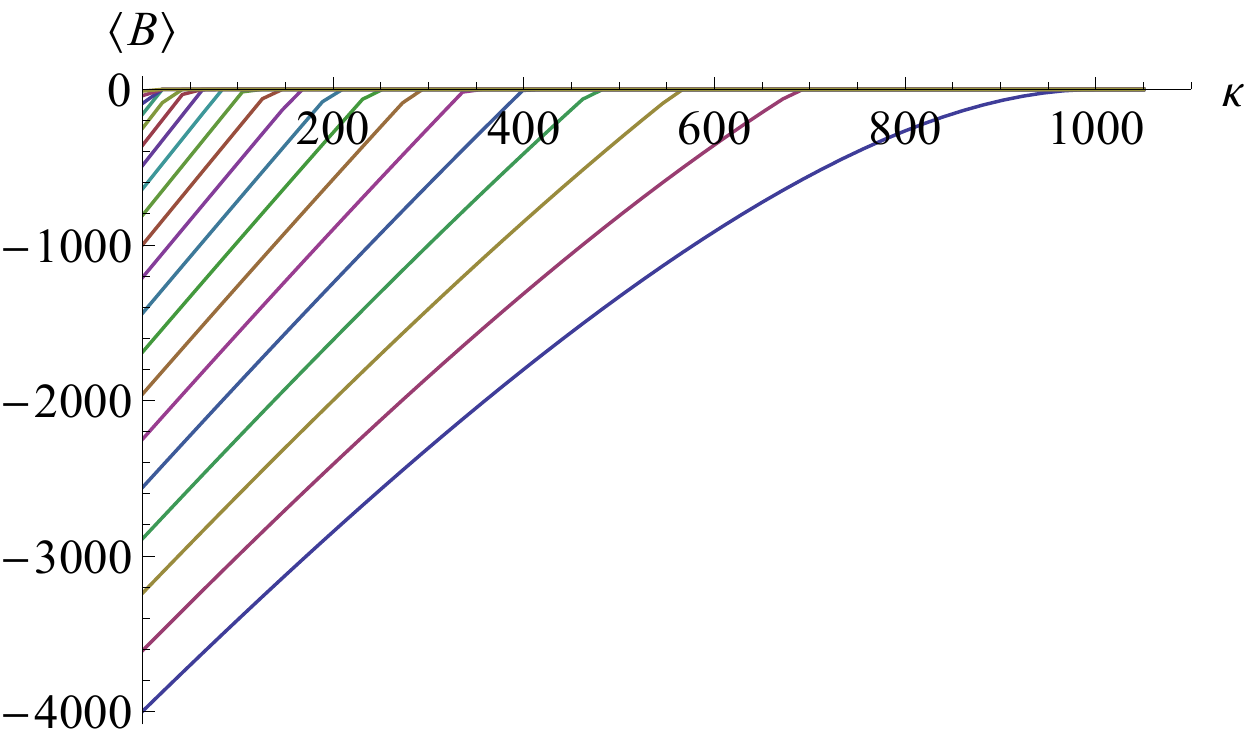}
\caption{On the vertical axis, we have the unnormalized quantum violation (the value of the left-hand side of Ineq. (\ref{appC:inequality2})) for $n=8000$; on the horizontal axis there is $\kappa$. The lines are for $\eta \in \{1,0.95,0.9,\ldots\}$ starting at $\eta=1$ for the blue line.}\label{fig:2}
\end{figure}
\end{center}
%
%\subsection{Further considerations}
%
Few remarks are in order:

\begin{itemize}

\item It seems that the considerable difference in the performance with $\kappa$ stems from the fact that in inequality (\ref{appC:inequality2}) the coefficients are \textit{balanced}, whereas in the inequality for the Dicke states (\ref{appC:inequality}) they are \textit{non-linear} functions of $n$. Observe that in both cases $\langle S_z^2\rangle_{\mbox{ideal}}=0$.

\item In inequality (\ref{appC:inequality2}), an interesting fact occurs when we take different values of the parameter $\eta$ for $\langle S_x^2\rangle$ and $\langle S_{\pi/4}^2\rangle$. Then one observes that the two terms in $n^2$ no longer cancel each other out; hence in this way, one can \textit{mimic} a huge violation of the order of $n^2$ by taking one $\eta$ larger than the other.

\item It seems that in order to find a Bell inequality for $\ket{D^{n/2}_n}$ that can tolerate a greater $\kappa$, it  would be desirable to test other Bell inequalities that are also violated by $\ket{D^{n/2}_n}$, but may have more balanced coefficients; the price to pay, of course, will be that their effective quantum violation will be smaller.

\item Last but not least, going to smaller but controllable numbers of atoms (say 100, or so) is, of course, another way to detect (maybe not so `many', but still) many-body nonlocality.

\end{itemize}

\end{document}